\documentclass[11pt]{amsart}
\usepackage{geometry}                
\usepackage[latin1]{inputenc}   
\usepackage{mathpazo,cite}  
\usepackage[usenames,dvipsnames,cmyk]{xcolor}
\usepackage{url}
\usepackage{amsmath,amssymb,latexsym,mathrsfs,amssymb,amsthm,amsfonts}
\usepackage{enumerate}
\usepackage{enumitem,overpic}
\usepackage{mathtools} 
\usepackage[all]{xy}
\usepackage{tikz, tkz-euclide}
\usepackage{setspace}
\renewcommand{\vec}{\mathbf}

\usepackage[colorlinks=true, linkcolor=MidnightBlue, citecolor=MidnightBlue, urlcolor=MidnightBlue]{hyperref}

\newcommand{\RS}{\Sigma}

\newtheorem{theorem}{Theorem}[section]
\newtheorem{lemma}[theorem]{Lemma}
\newtheorem{proposition}[theorem]{Proposition}

\theoremstyle{definition}
\newtheorem{definition}[theorem]{Definition}

\theoremstyle{remark}
\newtheorem{remark}[theorem]{Remark}
\newtheorem{example}[theorem]{Example}

\let\Re=\undefined\DeclareMathOperator{\Re}{Re}
\let\Im=\undefined\DeclareMathOperator{\Im}{Im}

\DeclareMathOperator{\diag}{diag}

\DeclareMathOperator{\wron}{Wron}
\DeclareMathOperator{\tr}{tr}
\newcommand*\dd{\mathop{}\!\mathrm{d}}
\newcommand{\ii}{\ensuremath{\mathrm{i}}}
\newcommand{\ee}{\ensuremath{\,\mathrm{e}}}

\title{Computation of large-genus solutions of the Korteweg-de Vries equation}

\author{Deniz Bilman}
\address{Department of Mathematical Sciences, University of Cincinnati, OH}
\email{bilman@uc.edu}
\urladdr{https://homepages.uc.edu/~bilman/}

\author{Patrik Nabelek}
\address{Department of Mathematics, Oregon State University, OR}
\email{nabelekp@oregonstate.edu}
\urladdr{https://math.oregonstate.edu/node/15580}

\author{Thomas Trogdon}
\address{Department of Applied Mathematics, University of Washington, WA}
\email{trogdon@uw.edu}
\urladdr{http://faculty.washington.edu/trogdon/}

\begin{document}

\begin{abstract}
    We consider the numerical computation of finite-genus solutions of the Korteweg-de Vries equation when the genus is large. Our method applies both to the initial-value problem when spectral data can be computed and to dressing scenarios when spectral data is specified arbitrarily.  In order to compute large genus solutions, we employ a weighted Chebyshev basis to solve an associated singular integral equation.  We also extend previous work to compute period matrices and the Abel map when the genus is large, maintaining numerical stability.  We demonstrate our method on four different classes of solutions.  Specifically, we demonstrate dispersive quantization for ``box'' initial data and demonstrate how a large genus limit can be taken to produce a new class of potentials.
\end{abstract}

\maketitle

\section{Introduction}
Consider the Korteweg-de Vries (KdV) equation, written in the form
\begin{align}
    q_t + 6 q q_x + q_{xxx} = 0,\quad x\in[0,L),~t>0,
    \label{KdV}
\end{align}
subject to periodic boundary conditions. A main outcome of this work is an efficient numerical method for the computation of the inverse scattering transform associated with \eqref{KdV}, that is, a numerical inverse scattering transform for the Schr\"odinger operator with a periodic, piecewise smooth, potential.

\begin{figure}[h!]  
  \centering
  \begin{overpic}[width = 0.6\linewidth]{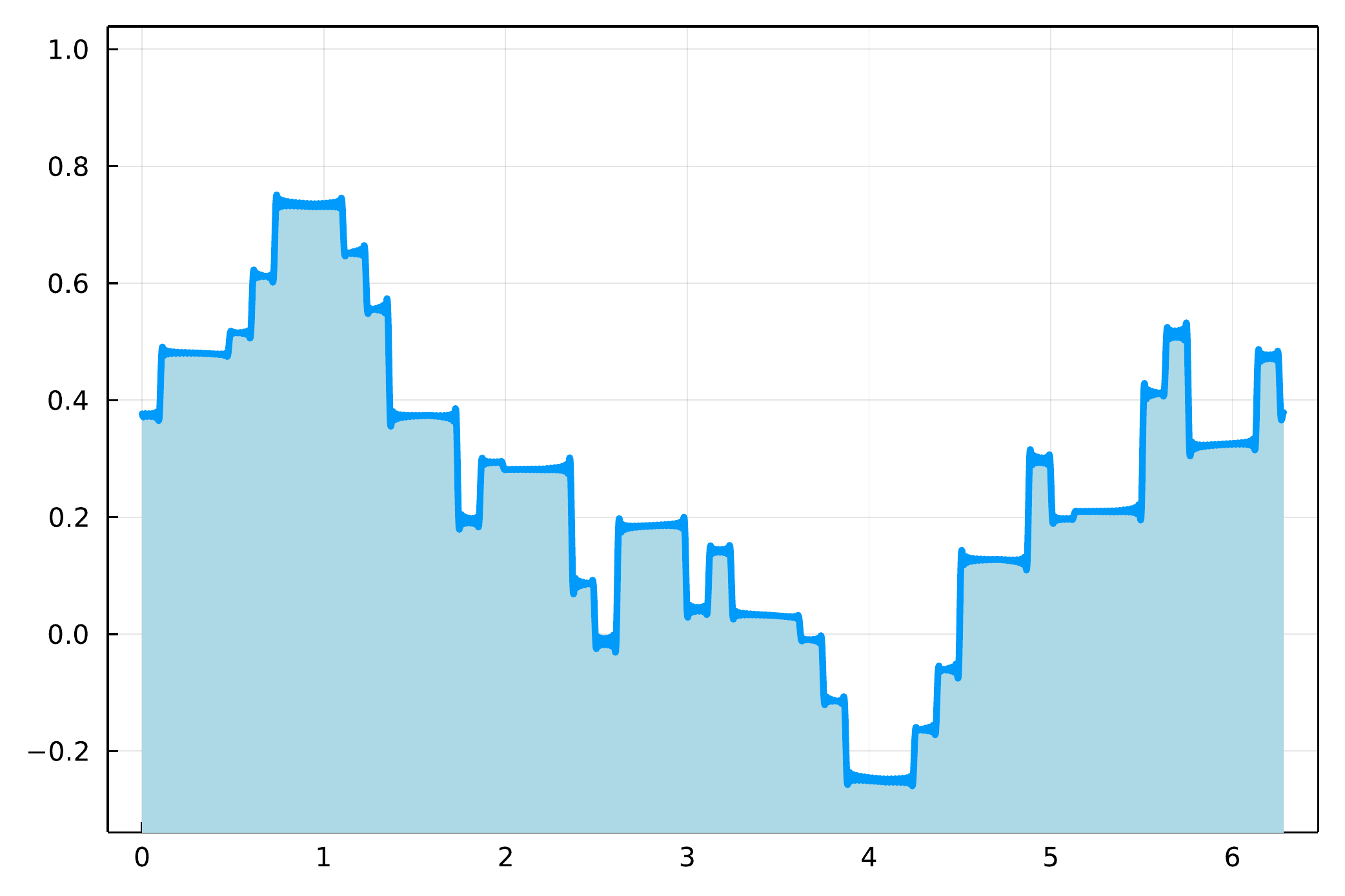}
    \put(-3,25){\rotatebox{90}{$u(x,1.03 \pi)$}}
    \put(55,-1){{$x$}}
  \end{overpic}
  \caption{\label{fig:demo} $u(x,1.03\pi)$ in \eqref{eq:ukdv} with \eqref{eq:udata} {solving KdV equation in the form $u_t-uu_x + u_{xxx}=0$ (as in \cite{Chen2014a})}. This plot shows dispersive quantization. The solution appears to be piecewise smooth at rational-times-$\pi$ times and fractal otherwise.  For the KdV equation, this was first observed by Chen and Olver, see \cite{Chen2014a}, for example. These plots are produced using a genus $g = 300$ approximation.  More details concerning the computation of this solution can be found in Section~\ref{s:numerics-box}.}
  \end{figure}

Following \cite{TrogdonFiniteGenus,McLaughlin2018,TrogdonSOBook}, we formulate a Riemann--Hilbert problem for the so-called Bloch eigenfunctions of the Schr\"odinger operator.  In the finite-gap case, two eigenfunctions are classically used to construct the associated Baker--Akhiezer function on a hyperelliptic Riemann surface. A key improvement we make here over the numerical approach in \cite{TrogdonFiniteGenus} is that through a transformation $z^2 = \lambda$ we pose Riemann--Hilbert problem with jumps supported on the gaps.  This idea was used with limited scope in \cite{TrogdonDressing}.

The key numerical innovations come from the use of the Chebyshev-U and Chebyshev-V (third and fourth kind) polynomials and their weighted Cauchy transforms.  These weighted Cauchy transforms encode the singularity structure of the solution of the Riemann--Hilbert problem we pose and allow for an extremely sparse representation of the solution.  This is demonstrated in Figure~\ref{fig:coefs} below.

The developed numerical method can handle high-genus potentials --- Riemann-Hilbert problems with jump matrices supported on thousands of intervals.  This ability stems from the choice of Chebyshev-U and V basis{, which is one of the new ideas in this work.}  But there are additional developments that are required to even pose that Riemann--Hilbert problem in the inverse scattering context.  These developments are related to computing the period matrix for a basis of holomorphic differentials when the genus is high.  Through a judicious choice of the basis of holomorphic differentials, we develop an approach that appears stable as $g \to \infty$.  Furthermore, our use of Chebyshev-U and Chebyshev-V polynomials and their weights is predicated on having a potential that produces a Baker--Akhiezer function with poles at band ends.  As this is not the generic setting, following \cite{TrogdonFiniteGenus}, we construct a parameterix Baker--Akhiezer function to move the poles to the band ends, {without loss of generality}. {A simplification explained in Section~\ref{s:move} allows for $g$ to be large with this approach.}

In this paper we do not discuss, in detail, the computation of the direct scattering transform for the Schr\"odinger operator with a periodic potential --- the computation of the periodic, anti-periodic and Dirichlet spectra.  We do accomplish this using existing standard techniques but consider any improvement on these approaches as important future research topics.  In the case of a ``box'' potential (see \eqref{eq:udata} below) we can compute the Bloch eigenfunctions explicitly and apply simple root-finding routines to compute the requisite spectra.  In Figure~\ref{fig:demo} we plot the evolution of this infinite-genus box potential to time $t = 1.03\pi$ using a genus $g = 300$ approximation.
{Realization of dispersive quantization in a nonlinear setting is clear --- the solution appears to be piecewise smooth whenever $t$ is a rational multiple of $\pi$ \cite{Chen2014a,Olver2010a,Berry1996,Talbot1836}.}

\subsection{Relation to other work}

We emphasize that the computation of finite-genus solutions is a nontrivial matter.  Lax's foundational paper \cite{Lax} includes an appendix by Hyman, where solutions of genus 2 were obtained through a variational principle. The classical approach to their computation goes through their algebro-geometric description in terms of Riemann surfaces, see \cite{Deconinck-theta} or \cite{klein}, for instance.  While very effective, this approach has only been applied to relatively small genus Riemann surfaces.

Yet another approach is by the numerical solution of the so-called Dubrovin equations \cite{Algebro,Dubrovin}. And the finite-genus solution is easily recovered from the solution of the Dubrovin equations \cite{TheoryOfSolitons, Osborne}.  We do not take this approach again because {1)} the dimensionality involved may pose possible stability issues and {2)} one has to time-step the solution to get to large times.  The Riemann--Hilbert problem we pose has $x$ and $t$ as {explicit} parameters, and therefore the complexity associated with computing the solution at any given $(x,t)$ value is independent of $(x,t)$.

As mentioned above, a numerical Riemann--Hilbert approach was introduced in \cite{TrogdonFiniteGenus} (see also \cite{TrogdonSOBook}). While the approach in \cite{TrogdonFiniteGenus} should be seen as the precursor to the current work, it was only successful for small genus solutions and was too inefficient when the genus is larger than, say, 10.  

\subsection{Outline of the paper}

The paper is laid out as follows.  In Section~\ref{s:ist} we review the inverse spectral theory for the Schr\"odinger operation with a periodic or finite-gap potential, connecting it to an underlying Riemann surface (in the finite-gap case) and the associated Baker--Akhiezer function.  In Section~\ref{s:move} we discuss the parametrix Baker--Akhiezer function that allows the movement of poles and in Section~\ref{s:form_RHP} we begin formulating a Riemann--Hilbert problem satisfied by the planar representation of the Baker--Akhiezer function.  In Section~\ref{s:adaptive-Cauchy} we convert the Riemann--Hilbert problem to a singular integral equation on a collection of intervals.  We look for solutions in a weighted $L^2$ space.  In Section~\ref{s:NIST} we discuss the numerical solution of the singular integral equation from the previous section, discussing both preconditioning and adaptivity of grid points. In Section~\ref{s:Applications} we discuss the comptuation of various solutions of the KdV equation.  Specifically, in Section~\ref{s:numerics-dressing-slow} and \ref{s:numerics-dressing-univ} we compute solutions with prescribed spectral data.  In Section~\ref{s:numerics-dressing-univ} we give a formal universality result that demonstrates how primitive potentials can be obtained in a large-genus limit.  In Section~\ref{s:numerics-smooth} we solve the initial-value problem for the KdV equation with smooth initial data.  In Section~\ref{s:numerics-box}, we give an extensive treatment of the numerical solution of the KdV equation with ``box'' initial data.  

This work gives rise to many interesting questions.  The work here, while empirically valid, comes with no rigorous error bounds and the full numerical analysis of the method  is an open problem.  Similarly, we provide no error bounds for the approximation of an infinite genus potential by one of finite genus.   
{The reconstruction formula} \eqref{eq:dubrovin-recover} appears to imply that the errors will be small if one removes gaps such that $\alpha_{j+1}-\beta_j$ is small.  But this removal has a non-trivial impact on $\gamma_k(x)$ for $k \neq j$ and that error needs to be estimated.  This leads to the question of understanding both the large $g$ limit of the period matrix of our basis of holomorphic differentials and the large $g$ limit of the singular integral equation we formulate. These issues will be addressed in future work.  
{There is also some room for improvement in the complexity of the numerical method.}
{A significant improvement on the complexity would be to put it inside a matrix-free framework using some incarnation of the fast multiple method \cite{Carrier1988}.} Code used to generate the plots in the current paper can be found at \cite{Trogdon2022}.

Before we proceed, we give a remark that details our notational conventions.
\begin{remark}[Notational conventions]
We use capital boldface letters, e.g., $\mathbf{M}$, to denote $1\times 2$ row-vectors and to denote matrices, with the exception of the Pauli matrices,
\begin{equation}
\sigma_{1}:=\begin{bmatrix}
0 & 1 \\
1 & 0
\end{bmatrix}, \quad \sigma_{2}:=\begin{bmatrix}
0 & -\mathrm{i} \\
\mathrm{i} & 0
\end{bmatrix}, \quad \sigma_{3}:=\begin{bmatrix}
1 & 0 \\
0 & -1
\end{bmatrix},
\end{equation}
and denote the identity matrix or identity operator by $\mathbb{I}$. 
We use lowercase boldface letters, e.g., $\mathbf{u}$ to denote column-vectors that are of arbitrary dimension.
We use the capitalized Greek characters, e.g., $\Psi$, to denote functions defined on a hyperelliptic Riemann surface. Given such $\Psi$, we denote by $\psi_{\pm}$ the (scalar-valued) components of its planar representation in the form of a row-vector which is denoted by the boldface version ${\bf \Psi}$ of $\Psi$.
We use superscripts $f^{\pm}(z)$ to denote the boundary values of $f$ at a point $z$ on an oriented contour taken from the left ($+$) and the right ($-$) side of the contour with respect to the orientation. 
We use fraktur $\mathfrak{a}$ and $\mathfrak{b}$ to denote the cycles on a Riemann surface.  Lastly, for a function $f: \mathbb C \to \mathbb C$ we use $f(\mathbf{u})$ to denote $f$ applied entrywise to the vector $\mathbf{u}$.
\label{rem:notation}
\end{remark}

\subsection*{Acknowledgements}
{Bilman was supported by the National Science Foundation under grant number DMS-2108029. Trogdon was partially supported by the National Science Foundation under grant number DMS-1945652. The authors thank Ken McLaughlin, Peter Miller, and Peter Olver for helpful and motivating discussions.}

\section{Inverse Scattering Transform for Periodic Solutions} \label{s:ist}
In this section we give a summary of the well-known inverse scattering transform associated with \eqref{KdV} and define the quantities relevant to the method we develop in this work, along with particular choices we make. 
The KdV equation in the form \eqref{KdV} is the $\lambda$-independent compatibility condition for the linear problems, i.e., the Lax pair,
\begin{align}
\mathcal{L}(t) \psi &= \lambda \psi, \label{Lax-x}\\
\psi_t &= \mathcal{P}(t) \psi,\label{Lax-t}
\end{align}
where is $\mathcal{L}$ is the Schr\"odinger operator
\begin{align}
\mathcal{L}(t):=-\frac{\dd^2}{\dd x^2} - q(\diamond,t)
\label{Schrodinger-op}
\end{align}
with the time-dependent potential $-q(\diamond,t)$, and $\mathcal{P}$ is the skew-symmetric operator
\begin{align}
\mathcal{P}(t):=-4\frac{\dd^3}{\dd x^3} - 6 q(\diamond,t)\frac{\dd}{\dd x} - 3 q_x(\diamond,t).
\end{align}
The compatibility condition for the system of linear problems \eqref{Lax-x}--\eqref{Lax-t} yields the operator equation, referred to as the Lax equation \cite{Lax68}, in the form
\begin{align}
    \frac{\dd}{\dd t}\mathcal{L}(t) + [\mathcal{L}(t),\mathcal{P}(t)] = 0,
    \label{Lax}
\end{align}
which is equivalent to the KdV equation \eqref{KdV} in the sense that the left-hand side defines an operator of multiplication by the function $-(q_t + 6 q q_x + q_{xxx})$, where $[\mathcal{L},\mathcal{P}]:=\mathcal{L}\mathcal{P}-\mathcal{P}\mathcal{L}$ is the operator commutator. As $q$ evolves in time according to the KdV equation \eqref{KdV}, \eqref{Lax} defines an isospectral deformation of the Schr\"odinger operator $\mathcal{L}$.

\subsection{The spectrum and the Riemann surface}
For fixed $t\geq 0$, consider the problem \eqref{Lax-x} for the Schr\"odinger operator with the time-independent potential $-q(\diamond,t) = -q(\diamond)$:
\begin{align}
-\psi_{xx} - q \psi = \lambda \psi,
\label{Schrodinger-problem}
\end{align}
for real periodic $q$ with minimal period $L>0$: $q(x+L) = q(x)$.
The Bloch spectrum $\sigma_\mathrm{B}(q)$ associated with the periodic potential $-q$ for the Schr\"odinger operator \eqref{Schrodinger-op} is 
\begin{align}
\sigma_\mathrm{B}(q):=\{\lambda\in\mathbb{C}\colon \text{there exists a solution $\psi(\diamond;\lambda)$ to \eqref{Schrodinger-problem} such that }\sup_{x\in\mathbb{R}}|\psi(x;\lambda)|<\infty \}.
\end{align}
For real-valued smooth (and periodic) $q$, the Bloch spectrum is a countable union of real intervals
\begin{equation}
\sigma_\mathrm{B}(q) = \bigcup_{k=1}^{g+1} [\alpha_k, \beta_k],\quad\text{where}~g\in\mathbb{Z}_{>0}~\text{or}~g=\infty,
\end{equation}
with 
\begin{equation}
\alpha_k < \beta_k < \alpha_{k+1},\quad k = 1, 2, \ldots.
\end{equation}
We refer to the intervals $[\alpha_k, \beta_k] \subset \sigma_\mathrm{B}(q)$ as \emph{bands} and $(\beta_k, \alpha_{k+1})$ as \emph{gaps}.
If the number of intervals $g+1$ is finite, $\beta_N := +\infty$ and the last interval is $[\alpha_{g+1}, +\infty)$, in which case the associated $-q$ is called a \emph{finite-gap} potential. The endpoints $\alpha_j$ and $\beta_j$, $j=1,2,\ldots,N$, remain invariant as $q(\diamond, t)$ evolves according to the KdV equation \eqref{KdV}, and hence $\sigma_{\mathrm{B}}(q_0) = \sigma_{\mathrm{B}}(q)$ for $q(\diamond,t)$ solving $\eqref{KdV}$ with $q(\diamond, 0)=q_0$.
The following well-known symmetry transformations associated with the KdV equation play a role in various choices we will make in this work.
\begin{remark}[Two symmetry groups of KdV] Suppose that $q(x,t)$ is a solution of \eqref{KdV}. \label{r:symmetry}
\begin{itemize}
    \item \emph{Galilean transformation:} The function
    \begin{align}
        \tilde{q}(x,t):=q(x-6ct,t)+c
    \label{Galielean-sym}
    \end{align}
is also a solution of \eqref{KdV} for any constant $c$.
    \item \emph{Scaling transformation:} The function
        \begin{align}
            \tilde{q}(x,t):=c^{2} q(c x, c^3 t)
        \label{scaling-sym}
        \end{align}
is also a solution of \eqref{KdV} for any constant $c$.
\end{itemize}
\end{remark}
Given $q=q(\diamond,t)$ and $\alpha_1 = \min(\sigma_\mathrm{B}(q))$, using the Galilean symmetry transformation \eqref{Galielean-sym} with $c=\alpha_1$ lets one map $q(\diamond, t)$ to $\tilde{q}(\diamond-6 c t, t) + c$ for which $\min (\sigma_\mathrm{B}(\tilde{q})) = 0$. Doing so becomes useful in the formulation of a Riemann-Hilbert problem (and of the associated singular integral equation). This transformation is employed in the numerical implementation of our method: once $\alpha_1 \in \sigma_{\mathrm{B}}(q_0)$ is computed for given $q_0$ at $t=0$, we perform the spectral shift described above and then invert it to obtain $q(\diamond, t)$ from $\tilde{q}(\diamond-6 c t, t) + c$ at a later time $t>0$. Accordingly, we take $\alpha_1 = 0$ without loss of generality in the remainder of this paper. 

For our (computational) purposes, we restrict the theory to the finite-gap case. For $q_0$ giving rise to $g+1$ bands 
\begin{equation}
\sigma_\mathrm{B}(q_0)= [\alpha_{g+1}, +\infty) \cup \left( \bigcup_{j=1}^g [\alpha_j, \beta_j] \right),
\end{equation}
and $g$ gaps, $g\geq 2$, consider the monic polynomial $P(\lambda)$ of degree $2g +1$ given by
\begin{equation}
P(\lambda) := (\lambda-\alpha_{g+1}) \prod_{j=1}^{g}(\lambda-\alpha_{j}) (\lambda-\beta_{j}),
\end{equation}
and define
$\RS$ to be the hyperelliptic (elliptic, if $g=1$) nonsingular Riemann surface
\begin{equation}
\RS:=\{(\lambda, w) \in\mathbb{C}^2 \colon w^2 = P(\lambda)\},
\label{RS-def}
\end{equation}
associated with the zero locus of $F(\lambda, w) := w^2 - P(\lambda)$. The points $(\alpha_j,0)$, $(\beta_j,0)$, $j=1,2,\ldots, g,$ and $(\alpha_{g+1},0)$ on $\Sigma$ are branch points for the projection $(z,w)\mapsto z$ and there is a single point at $\infty$ on $\Sigma$. 
For $P_0=(\lambda_0, w_0) \in \Sigma$ we have the following choices of a local coordinate $\zeta$:
\begin{itemize}[topsep=1em]
\item If $P_0$ is not a branch point and not $\infty\in \Sigma$, then for $(\lambda,w)$ near $P_0$ we may take essentially $\lambda$ to be a local coordinate, so we write for $|\zeta|$ sufficiently small
\begin{equation}
\lambda = \lambda(\zeta):= \lambda_0 + \zeta,\qquad w = w(\zeta):= \sqrt{P(\lambda(\zeta))}.
\label{local-regular}
\end{equation}
\item  If $P_0=E_k=(\lambda_k, 0)$ for some $k$, then for $(\lambda,w)$ near $P_0$ we may write
 \begin{equation}
\lambda = \lambda(\zeta):= \lambda_k + \zeta^2,\qquad w = w(\zeta):= \zeta \sqrt{\prod\limits_{\substack{j=1 \\ j\neq k}}^{2g_+1}(\zeta^2+ \lambda_k - \lambda_j )}.
\label{local-branch}
\end{equation}
\item Finally, if $P_0 = \infty\in\Sigma$, then for $(\lambda,w)$ near $P_0$ we may write
\begin{equation}
\lambda=\lambda(\zeta):= \frac{1}{\zeta^2},\qquad w=w(\zeta):=\frac{1}{\zeta} \sqrt{\prod\limits_{j=1}^{2g_+1}(1 -\zeta^2  \lambda_j )}.
\label{local-infty}
\end{equation}
\end{itemize}
In all three cases $\lambda(\zeta)$ and $w(\zeta)$ are locally holomorphic functions of $\zeta$ in a neighborhood of $\zeta=0$ with non-zero derivatives at $\zeta=0$, making them locally injective, and $\zeta(P_0)=0$.

Define the branch of square root $R(\lambda)$ of $P(\lambda)$ to be the (unique) single-valued function that is analytic in $\mathbb{C}\setminus \sigma_\mathrm{B}(q_0)$ and that satisfies $R(\lambda)^2 = P(\lambda)$ along with the asymptotic behavior
\begin{equation}
R(\lambda) = |\lambda|^{g+\frac{1}{2}} \left( (-1)^g \ii  + o(1) \right) ,\quad \lambda \to -\infty.
\end{equation}
Importantly, the boundary values $R(\lambda)$ taken on the bands $\sigma_\mathrm{B}(q)$ are real-valued and $R(\lambda)$ is purely imaginary on the gaps, namely for $\lambda\in (\beta_j, \alpha_{j+1})$, $j=1,2,\ldots, g$. Assuming that the bands constituting $\sigma_\mathrm{B}(q_0)$ are oriented in the increasing direction of the real line, we define the boundary values $R^\pm(\lambda)$ of $R$ taken on $\sigma_\mathrm{B}(q_0)$ as:
\begin{equation}
R^{\pm}(\lambda) := \lim_{\epsilon \downarrow 0} R(\lambda \pm \ii \epsilon),\quad \lambda\in\sigma_{\mathrm{B}}(q).
\end{equation}
We denote by $+R(\lambda)$ the function that's defined for all $\lambda\in\mathbb{C}$ which coincides with $R(\lambda)$ for $\lambda\in\mathbb{C}\setminus \sigma_\mathrm{B}(q_0)$ and also satisfies $+R(\lambda) = R^+(\lambda)$ for $\lambda\in \sigma_\mathrm{B}(q_0)$; and define
the sheets $\RS_\pm$ of $\RS$ by
\begin{equation}
\RS_\pm := \left\{ (\lambda, \pm R(\lambda)) \colon \lambda \in \mathbb{C} \right\}.
\end{equation}
The Riemann surface $\RS$ is topologically equivalent to a sphere with $g$ handles, which is obtained by gluing the two sheets $\RS_{\pm}$ along the edges of the cuts (the bands) $[\alpha_j, \beta_j]$, $j=1,2,\ldots, g$, and $[\alpha_{g+1},\infty]$. We define the cycles $\{\mathfrak{a}_j, \mathfrak{b}_j\}_{j=1}^g$, which constitute a homology basis for $\RS$, as depicted in Figure~\ref{f:Riemann-Surface-Cycles}; and we denote by $\{\nu_1, \nu_2,\ldots, \nu_g \}$ the basis of normalized holomorphic differentials on $\RS$ that satisfy
\begin{equation}
\oint_{\mathfrak{a}_k} \nu_j = 2\pi \ii \delta_{jk},\qquad 1\leq j,k \leq g.
\end{equation}
The $g \times g$ matrix $\mathbf{B}$ defined by
\begin{equation}
{B}_{j k}:= \oint_{\mathfrak{b}_k} \nu_j,\qquad 1\leq j,k \leq g
\label{Riemann-matrix}
\end{equation}
is the Riemann matrix for $\Sigma$. $\mathbf{B}$ is symmetric with $\Re(\mathbf{B})<0$. Although finite-gap solutions of \eqref{KdV} have, in principle, representations given in terms Riemann theta functions $\Theta$ which are based on the Riemann matrix $\mathbf{B}$, our method does not require at all any explicit knowledge of the Riemann matrix $\mathbf{B}$.
\begin{figure}[h]
\includegraphics{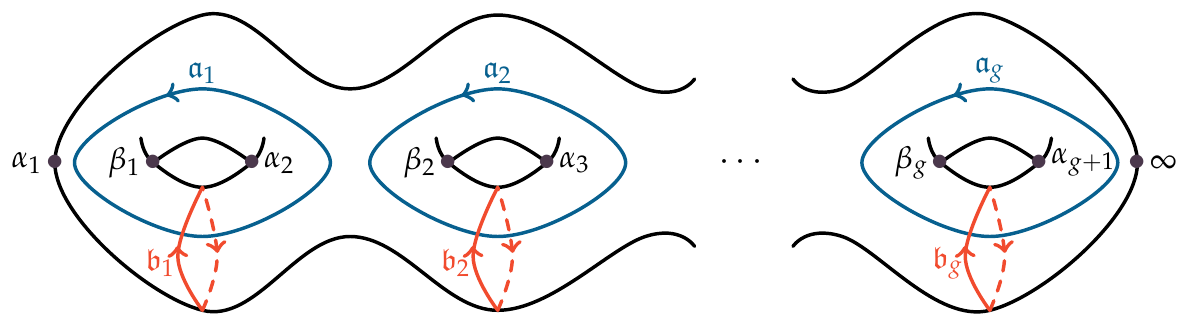}
\caption{An illustration of the hyperelliptic Riemann surface $\RS$ and the choices for the $\mathfrak{a}$- and the $\mathfrak{b}$-cycles.}
\label{f:Riemann-Surface-Cycles}
\end{figure}

We will use a particular basis of differentials that is ideal for our computational purposes, so some observations are in order. Let $\{\omega_1, \omega_2,\ldots, \omega_g \}$ be an arbitrary basis of holomorphic differentials on $\RS$ and define the $g\times g$ matrices $\mathbf{A}$ and $\tilde{\mathbf{B}}$ by
\begin{equation}
A_{j k}:= \oint_{\mathfrak{a}_k} \omega_j, \qquad \tilde{B}_{j k}:= \oint_{\mathfrak{b}_k} \omega_j,\qquad 1\leq j,k \leq g.
\label{A-B-period-matrices}
\end{equation}
It is well-known and easy to see that $\mathbf{A}$ is nonsingular since otherwise one can find a nontrivial linear combination of $\{\omega_j \}_{j=1}^g$ that has vanishing $\mathfrak{a}$-cycles, implying that the resulting holomorphic differential is identically zero, and hence contradicting the independence of the differentials $\omega_j$.  Note that there exists scalars $C_{l k}$, $1\leq l,k \leq g$, such that $\nu_l = \sum_{j=1}^g C_{lj} \omega_j$. Let $\mathbf{C}$ be the matrix of these scalars. Then we have
\begin{equation}
2\pi \ii \delta_{lk} = \oint_{\mathfrak{a}_k} \nu_l = \sum_{j=1}^g C_{lj} \oint_{\mathfrak{a}_k} \omega_j = \sum_{j=1}^g C_{lj} A_{jk} = [\mathbf{C}\mathbf{A}]_{lk},
\end{equation}
 implying that $\mathbf{C}\mathbf{A} = {2\pi\ii}\mathbb{I}$, and hence
$
 \mathbf{C}= 2\pi \ii \mathbf{A}^{-1},
$
 which yields the relation
 \begin{equation}
 \begin{bmatrix}
 \nu_1 \\ \nu_2 \\ \vdots \\ \nu_g 
 \end{bmatrix}
 = 2\pi \ii \mathbf{A}^{-1}
  \begin{bmatrix}
 \omega_1 \\  \omega_2 \\ \vdots  \\ \omega_g
 \end{bmatrix}.
 \end{equation}
 
 A classical choice for the basis $\{\omega_1,\ldots,\omega_g\}$ is
 \begin{align*}
     \omega_j = \frac{\lambda^{j-1}}{R(\lambda)} \dd \lambda.
 \end{align*}
 But from a computational point-of-view, when $g$ is large, this basis is ill-conditioned.  It is better to chose
 \begin{align}\label{eq:stable-dif}
     \omega_j = \frac{\prod_{k=1, k \neq j}^{g} (\lambda - \alpha_k)}{R(\lambda)} \dd \lambda = \left[\prod_{k=1, k \neq j}^{g} \sqrt{\frac{\lambda - \alpha_k}{\lambda-\beta_k}} \right]\frac{\dd \lambda}{\sqrt{(\lambda-\alpha_{g+1})(\lambda-\alpha_j)(\lambda - \beta_j)}}.
 \end{align}
 
 \subsection{The Baker-Akhiezer function} 
 Our approach to compute solutions of \eqref{KdV} is based on numerical solution of a RH problem satisfied by a suitable renormalization of a Baker-Akhiezer function. We now give a series of definitions and then give construction of the relevant Baker-Akhiezer function built from certain solutions of the spectral problem
  \begin{align}
   \mathcal L(0) \psi = \lambda \psi.
\label{Schrodinger-t-0}
 \end{align}
 
  \begin{definition}
   For the hyperelliptic Riemann surface $\RS$ defined by $w^2 = P(\lambda)$, a divisor $D$ on $\RS$ is a formal sum
   \begin{align}
     D = n_1 P_1 + n_2 P_2 + \cdots + n_m P_m,
   \end{align}
   where $n_j \in \mathbb Z$ and $P_j \in \RS$ for $j =1,2,\ldots,m$.  A divisor is called positive if $n_j > 0$ for all $j$, and the degree of a divisor is the number $\sum_{j=1}^m n_j$.
 \end{definition}

 \begin{definition}[Abel map]
   Fix an arbitrary point $P_0$ on the Riemann surface $\RS$ defined by $w^2 = P(\lambda)$ and let $D = n_1 P_1 + n_2 P_2 + \cdots + n_m P_m$ be a divisor on $\RS$.  The Abel map $\mathcal A(D)$ is given by
   \begin{align}
     \mathcal A(D) = \left( \sum_{j = 1}^m n_j \int_{P_0}^{P_j} \nu_\ell \right)_{1 \leq \ell \leq g},
   \end{align}
  where the path of integration from $P_0$ to $P_j$ is chosen to be the same for each $\ell =1,2,\ldots,g$.
 \end{definition}

 \begin{definition}\label{def:BA}
   For the hyperelliptic Riemann surface $\RS$ defined by $w^2 = P(\lambda)$, let $Q_1,\ldots,Q_n$ be points on $\RS$ with local parameters $\zeta_j$, $j=1,\ldots,n$, with $\zeta_j(Q_j) = 0$, as in \eqref{local-regular}--\eqref{local-infty}.  To each point $Q_j$ associate an arbitrary polynomial $q_j(\zeta_j^{-1})$ of the reciprocal of the associated local parameter.  Next, let $D = P_1 + P_2 + \cdots P_g$ be an arbitrary positive divisor on $\RS \setminus \{Q_1,\ldots,Q_n\}$ of degree $g$. Then $\mathcal V(D; Q_1,\ldots,Q_n; q_1,\ldots,q_n)$ is the linear space of functions $\Psi(P)$ on $\RS$ satisfying the following properties:
   \begin{enumerate}
   \item The function $\Psi(P)$ is meromorphic on $\RS \setminus \{Q_1,\ldots,Q_n\}$ and has poles at the points of $D$.
   \item There exists a neighborhood of every point $Q_j$, $j=1,\ldots,n$, such that the product\newline $\Psi(P) \exp \left( - q_j\left(\zeta_j(P)^{-1} \right) \right)$ is analytic in this neighborhood.
   \end{enumerate}
   Such a function $\Psi(P)$ is called a Baker-Akhiezer function. 
 \end{definition}
 
 The following theorem from \cite[Theorem 2.24]{Algebro} is quite useful in this work.  We do not define all the quantities that arise in its statement but only highlight the components that are crucial for us to proceed.  

 \begin{theorem}\label{t:one_dim}
   The space $\mathcal V(D; Q_1,\ldots,Q_n; q_1,\ldots,q_n)$ is one-dimensional for a non-special divisor\footnote{The divisors $D$ that we encounter in this work will always be  of the form $D = P_1+ P_2 + \ldots + P_g$ for distinct $P_j$.  Such divisors are non-special \cite{DubrovinNotes}.} $D$ and polynomials $q_j$ with sufficiently small coefficients.  Its basis is described explicitly by
   \begin{align}\label{eq:ratio-theta}
     \Psi_0(P) = \frac{\Theta \left( \mathcal A(P) + \mathbf{v} - \mathbf{d}; \mathbf{B}\right)}{\Theta \left( \mathcal A(P) -  \mathbf{d}; \mathbf{B}\right)} \ee^{\Omega(P)}
   \end{align}
   where $\Omega(P)$ is a normalized Abelian integral of the second kind\footnote{A normalized Abelian differential is a meromorphic differential that is normalized to integrate to zero over all the $\mathfrak a$-cycles.} with poles at the points $Q_1,\ldots,Q_n$,  the principle parts of which coincide with the polynomials $q_j(z_j)$, $j=1,\ldots,n$, $\Theta$ is Riemann's theta function, and $\mathbf{v}$ is a vector of the $b$-periods of the integrals of $\Omega(P)$:
   \begin{align}\label{eq:Omega-cycles}
     v_j = \oint_{\mathfrak{b}_j} \dd \Omega, \quad j =1,\ldots, g.
   \end{align}
   Further, $\mathbf{d} =  \mathcal A(D) + \mathbf{k}$ where $\mathcal{A}(D)$ is the Abel map and $\mathbf{k}$ is a vector of Riemann constants, and the integration path for the integrals
   \begin{align}
     \Omega(P) = \int_{P_0}^P \dd \Omega \quad \text{and} \quad \mathcal A(P)
   \end{align}
   is chosen to be the same.
 \end{theorem}
 
 We now focus our attention to solutions of $\eqref{Schrodinger-t-0}$. First, fix $x_0\in \mathbb{R}$ and define a set of fundamental solutions $\{ c(x;\lambda), s(x;\lambda)\}$ of
\eqref{Schrodinger-t-0} that are determined by
 \begin{align}
   c(x_0;\lambda) = 1, \quad c_x(x_0;\lambda) = 0,\label{eq:c}\\
   s(x_0;\lambda) = 0, \quad s_x(x_0;\lambda) = 1.\label{eq:s}
 \end{align}
 It is easy to verify that these solutions solve the following Volterra integral equations
 \begin{align}
   c(x;\lambda) = \cos\left( \sqrt{\lambda}(x-x_0) \right) - \int_{x_0}^x \frac{\sin \left( \sqrt{\lambda} (x - y) \right)}{\sqrt{\lambda}} q_0(y) c(y;\lambda) \dd y,\\
    s(x;\lambda) = \frac{\sin\left( \sqrt{\lambda}(x-x_0) \right)}{\sqrt{\lambda}} - \int_{x_0}^x \frac{\sin \left( \sqrt{\lambda} (x - y) \right)}{\sqrt{\lambda}} q_0(y) s(y;\lambda) \dd y.
 \end{align}
This demonstrates that these two solutions are entire functions of $\lambda$ for given $x$ because cosine and sine are even and odd functions, respectively, and the paths of integration are finite. We omit the parametric dependence of these solutions on $x_0$ in our notation. Next, we define the \emph{monodromy operator} $\mathcal{T}$ for \eqref{Schrodinger-t-0} by
$
(\mathcal{T}\psi)(x):=\psi(x+L)
$
and represent the action of $\mathcal{T}$ on the set of fundamental solutions constructed above. Consider the first-order system equivalent to \eqref{Schrodinger-t-0}
\begin{equation}
\label{first-order}
\frac{\dd}{\dd x}
\begin{bmatrix} 
\psi(x;\lambda) \\ \psi_x(x;\lambda) 
\end{bmatrix}
=
\begin{bmatrix}
0 & 1 \\ -(\lambda+q_0(x)) & 0
\end{bmatrix}
\begin{bmatrix} 
\psi(x;\lambda) \\ \psi_x(x;\lambda) 
\end{bmatrix}
\end{equation}
along with its fundamental solution matrix
\begin{equation}
\mathbf{F}(x;\lambda):=\begin{bmatrix} c(x;\lambda) & s(x;\lambda) \\ c_x(x;\lambda) & s_x(x;\lambda) \end{bmatrix},
\end{equation}
which is unimodular (see \cite[\S 1.6]{DubrovinNotes}) and satisfies $\mathbf{F}(x_0;\lambda)=\mathbb{I}$.
As $x\mapsto c(x+L;\lambda)$ and $x\mapsto s(x+L;\lambda)$ also define solutions of \eqref{Schrodinger-t-0} thanks to the periodicity of $q_0$,
\begin{equation}
\mathbf{F}(x+L;\lambda):=\mathbf{F}(x;\lambda)\mathbf{T}(\lambda)
\end{equation}
for some $x$-independent $2\times 2$ matrix $\mathbf{T}(\lambda)$. Evaluating both sides at $x=x_0$ yields
\begin{equation}
\mathbf{T}(\lambda) = \mathbf{F}(x_0+L;\lambda) = \begin{bmatrix} c(x_0+L;\lambda) & s(x_0+L;\lambda) \\ c_x(x_0+L;\lambda) & s_x(x_0+L;\lambda) \end{bmatrix},
\end{equation}
which called the \emph{monodromy matrix}. It is an entire matrix-valued function of $\lambda$ and $\det(\mathbf{T}(\lambda,x_0,L) ) = 1$. Then, for a solution $\psi(\diamond;\lambda)$ of \eqref{Schrodinger-t-0}, we have
\begin{equation}
\begin{bmatrix} 
\psi(x;\lambda) \\ \psi_x(x;\lambda) 
\end{bmatrix}
=\mathbf{F}(x,\lambda)\mathbf{a}(\lambda),
\label{psi-from-F}
\end{equation}
for some $\mathbf{a}(\lambda)\in\mathbb{C}^{2\times 1}$, and hence
\begin{equation}
\mathcal{T}^n
\left(\begin{bmatrix} 
\psi(\diamond;\lambda) \\ \psi_x(\diamond;\lambda) 
\end{bmatrix}\right)(x)
=\mathbf{F}(x,\lambda)\mathbf{T}(\lambda)^n\mathbf{a}(\lambda).
\end{equation}
for any positive integer $n$.
This and the unimodularity of $\mathbf{F}(x;\lambda)$ imply that for given $\lambda$, $\mu(\lambda)$ is an eigenvalue of $\mathcal{T}$ with an eigenfunction $\psi(\diamond;\lambda)$ in the (two dimensional) solution space of \eqref{Schrodinger-t-0} if and only if $\mu(\lambda)$ is an eigenvalue of $\mathbf{T}(\lambda)$ with an eigenvector $\mathbf{a}(\lambda)$. It can be shown that $\lambda\in\sigma_\mathrm{B}(q_0)$ if and only if $\mathcal{T}$ has an eigenvalue $\mu(\lambda)$ (in the solution space of \eqref{Schrodinger-t-0}) with $|\mu(\lambda)|=1$, and that $|\mu(\lambda)|=1$ implies $\lambda\in\mathbb{R}$ (see \cite[Lemma 1.6.4, Lemma 1.6.7]{DubrovinNotes}). To find $\mu(\lambda)$, set
$\Delta(\lambda):=\tfrac{1}{2}\tr(\mathbf{T}(\lambda))$ and note that $\Delta(\lambda)$ is entire and does not depend on $x_0$. Then the eigenvalues of $\mathbf{T}(\lambda)$ are $\mu_{\pm}(\lambda) := \Delta(\lambda) \pm \ii \sqrt{1-\Delta(\lambda)^2}$, where the square root is taken as the principal branch (positive for $-1< \Delta(\lambda)<1$), and the Bloch spectrum $\sigma_{\mathrm{B}}(q_0)$ consists of $\lambda\in\mathbb{R}$ such that $|\Delta(\lambda)|\leq 1$. 
The Bloch eigenfunctions $\psi_{\pm}(\diamond;\lambda)$ are bounded solutions of \eqref{Schrodinger-t-0} that are eigenfunctions of $\mathcal{T}$ for $\lambda\in\sigma_{\mathrm{B}}(q_0)$ with the normalization $\psi_{\pm}(x_0;\lambda)=1$, hence they are obtained by choosing the first row of $\mathbf{a}(\lambda)$ in \eqref{psi-from-F} to be equal to $1$
in solving
\begin{equation}
\mathbf{T}(\lambda) \mathbf{a}_{\pm}(\lambda) = \mu_{\pm}(\lambda)\mathbf{a}_{\pm}(\lambda)
\end{equation}
to get
\begin{equation}
\begin{split}
\psi_{\pm}(x;\lambda):=& c(x;\lambda) + \frac{\pm \ii \sqrt{1-\Delta(\lambda)^2} + \frac{1}{2}\left(T_{22}(\lambda) - T_{11}(\lambda)\right)}{T_{12}(\lambda)}s(x;\lambda)\\
=&c(x;\lambda) + \frac{\mu_{\pm}(\lambda) - c(x_0+L;\lambda)}{s(x_0+L;\lambda)}s(x;\lambda).
\end{split}
\end{equation}
Setting $\chi_{\pm}(x;\lambda):=-\ii (\partial_x \psi_{\pm}(x;\lambda))/\psi_{\pm}(x;\lambda)$, it follows that $\chi_{\pm}(x;\lambda)$ are periodic functions of $x$ with period $L$ and are independent of the choice of $x_0$ (see \cite[Lemma 1.1]{Dubrovin75-fan2}). Using the independence of the Wronskian $\wron(\psi_{\pm}(x;\lambda), c(x;\lambda))$ from $x$ yields the formula
\begin{equation}
\psi_{\pm}(x;\lambda)=c(x;\lambda) +\ii \chi_{\pm}(x_0;\lambda)  s(x;\lambda).
\label{psi-in-chi}
\end{equation}
For the finite-gap case (i.e., $g<\infty$) we are considering, $\Delta(\lambda)^2-1=0$ has finite (and odd) number of simple roots, which are the band endpoints \cite{Dubrovin75-fan2}. From the representation \cite[Eqn.\@ (1.8)]{Dubrovin75-fan2} of $\psi_{\pm}(\lambda)$ in terms of $\Re(\chi_{\pm}(x;\lambda))$ and the known asymptotic behavior (\cite[Lemma 1.1]{Dubrovin75-fan2}) of $\Re(\chi_{\pm}(x;\lambda))$ as $\lambda\to\infty$ within $\sigma_{\mathrm{B}}(q_0)$ we have that
  \begin{equation}
    \psi_{\pm}(x;\lambda) = \ee^{\pm \ii \sqrt{\lambda}^{+} (x - x_0)}(1 + o(1)), \quad \lambda  \to + \infty,
    \label{Bloch-asymptotics}
 \end{equation}
 where the branch cut for the square root is taken to be $[0,+\infty)$ and the branch is chosen so that $\sqrt{\lambda} = \ii |\lambda|^{1/2} + o (1)$ as $\lambda\to - \infty$. Here $ \sqrt{\lambda}^{+} $ for $\lambda>0$ denotes the boundary value of this branch of square root from the upper half-plane. Moreover, $\chi(x;\lambda)$ extends as a single-valued algebraic function on the Riemann surface $\Sigma$ defined in \eqref{RS-def}, with
\begin{equation}
\Re(\chi_{\pm}(x;\lambda)) = \frac{\pm R(\lambda)}{\prod_{j=1}^g (\lambda- \gamma_j(x))}
\end{equation}
where $\gamma_j(x)$ are located in gaps or their endpoints: $\beta_j \leq \gamma_j(x) \leq {\alpha_{j+1}}$. One also has the identity
\begin{equation}
{\psi_+(x;\lambda)\psi_-(x;\lambda)} = \frac{\prod_{j=1}^g (\lambda- \gamma_j(x))}{\prod_{j=1}^g (\lambda- \gamma_j(x_0))},
\label{psi-product-identity}
\end{equation}
see \cite[Theorem 2.1]{Dubrovin75-fan2}, and also \cite{DubrovinN74}.
It then follows (see \cite[Theorem 2.3]{Dubrovin75-fan1}) that the function $\Psi(x; P)$ defined on the Riemann surface $\Sigma$ by
 \begin{align} \label{eq:def-Psi}
   \Psi(x;P) = \begin{cases} 
     {\psi_+}(x;\lambda), \quad P = (\lambda, + R(\lambda)),\\
     {\psi_-}(x;\lambda), \quad P = (\lambda, - R(\lambda)).\end{cases}
 \end{align}
extends as a single-valued meromorphic (for $x\neq x_0$) function on $\Sigma\setminus \{\infty\}$ with poles at locations where $\chi(x_0;\lambda)$ has its simple poles (see \eqref{psi-in-chi}), namely, at $\lambda=\gamma_j(x_0)$, $j=1,2,\ldots, g$. The identity \eqref{psi-product-identity} implies that $\Psi(x;P)$ has a pole only on one of the sheets: $P_j:=(\gamma_j(x_0), \sigma_j R(\gamma_j(x_0)))$, one in each of the gaps, where $\sigma_j$ is either $1$ or $-1$, $j=1,2,\ldots,g$. $\Psi(x;P)$ also has an essential singularity at $P=\infty$ and its behavior for $P$ near $\infty$ is given by
\begin{equation}
\Psi(x;P)\ee^{-\ii z(P)(x-x_0)} = 1+ o(1),
\end{equation} 
where $z(P)$ denotes the reciprocal of the local coordinate \eqref{local-infty} near $\infty$: $z(P)^2 = \lambda$. Recalling Definition~\ref{def:BA}, these facts show that $\Psi(x;P)$ is a Baker-Akhiezer function on $\Sigma$ with $n=1$, $Q_1=\infty$ with the associated polynomial $q_1(z):=z$, and with the non-special divisor $D = P_1 + P_2 + \cdots + P_g$. Moreover, these conditions uniquely determine $\Psi(x;P)$ by Theorem~\ref{t:one_dim}.

\begin{remark}
The zeros of $\Psi(x;P)$ are at the points where $\lambda=\gamma_j(x)$, and they lie also in the gaps. It's well-known that the potential $q_0(x)$ can be recovered via the formula
\begin{equation}\label{eq:dubrovin-recover}
q_0(x) = 2\sum_{j=1}^g \gamma_j(x) - \sum_{j=1}^{g}(\alpha_j+\beta_j) - \alpha_{g+1},
\end{equation}
see, for example, \cite{DubrovinN74}. Our method for obtaining $q_0$ from $\Psi$ makes no reference to this formula, and hence avoids root-finding.
\end{remark}

 \subsubsection{Time dependence}
{The Bloch solutions $\psi_{\pm}$ of \eqref{Lax-x} can be constructed at a given fixed time $t$ as $q(\diamond,t)$ evolves according to the KdV equation. 
Let $\psi^{[t]}_{\pm}(x;\lambda)$ denote these solutions and we have $\psi^{[0]}_{\pm}(x;\lambda)=\psi_{\pm}(x;\lambda)$ which were studied in the previous subsection. 
While $\psi^{[t]}_{\pm}(x;\lambda)$ solve \eqref{Lax-x} with \eqref{Bloch-asymptotics} and the normalization $\psi^{[t]}_\pm(x_0;\lambda)=1$, they do not provide a set of simultaneous solutions of \eqref{Lax-x}--\eqref{Lax-t} as they do not satisfy \eqref{Lax-t}. 
A calculation identical to \cite[Proposition 6.2]{McLaughlin2018} shows that 
\begin{equation}
(\mathcal{L}(t) - \lambda)\left(\frac{\partial}{\partial t}\psi^{[t]}_{\pm}(x;\lambda) - \mathcal{P}(t)\psi^{[t]}_{\pm}(x;\lambda)\right) = 0
\end{equation}
which implies that
\begin{equation}
\begin{split}
    \frac{\partial}{\partial t} \psi^{[t]}_{\pm}(x;\lambda) + d_{\pm}(t;\lambda)\psi^{[t]}_{\pm}(x;\lambda) &= \mathcal{P}(t)\psi^{[t]}_{\pm}(x;\lambda)\\
    &=  (4\lambda-2q(x,t))\frac{\partial}{\partial x}\psi^{[t]}_{\pm}(x;\lambda) + q_x(x,t)\psi^{[t]}_{\pm}(x;\lambda)
\end{split}
\label{psi-time}
\end{equation}
for $x$-independent coefficients $d_{\pm}(t;\lambda)$ that are given by
\begin{equation}
    d_{\pm}(t;\lambda)= (4\lambda-2q(x_0,t))\frac{\mu_{\pm}(\lambda) - c(x_0+L;\lambda)}{s(x_0+L;\lambda)} +  q_x(x_0,t).
\end{equation}
These are obtained by evaluating \eqref{psi-time} at $x=x_0$. Again from \cite[Proposition 6.2]{McLaughlin2018} (see also \cite[Proposition 3.3]{McLaughlin2018}) we have the asymptotic behavior
\begin{equation}
    d_{\pm}(t;\lambda) = \pm 4\ii(\sqrt{\lambda})^3 + O\left(\frac{1}{\sqrt{\lambda}}\right),\quad \lambda\to\infty.
\end{equation}
Following \cite{McLaughlin2018}, one uses the solution $\phi_{\pm}(t;\lambda)$ of
\begin{equation}
    \frac{\partial}{\partial t}\phi_{\pm}(t;\lambda) = d_{\pm}(t;\lambda)\phi_{\pm}(t;\lambda)
\end{equation}
satisfying $\phi_{\pm}(0;\lambda)=1$. Then 
\begin{equation}
    \psi_{\pm}(x,t;\lambda) := \psi^{[t]}_{\pm}(x;\lambda)\phi_{\pm}(t;\lambda)
    \label{psi-simultaneous}
\end{equation}
define a set of simultaneous solutions of \eqref{Lax-x}--\eqref{Lax-t}. As proved in \cite[Proposition 6.3]{McLaughlin2018}, $\phi_{\pm}(t;\lambda)$ satisfy
\begin{equation}
    \phi_{\pm}(t;\lambda) = \ee^{\pm 4 \ii (\sqrt{\lambda})^{3} t}\left(1+O\left(\frac{1}{\sqrt{\lambda}}\right)\right)\quad \lambda \rightarrow \infty.
\end{equation}
Moreover, the product in \eqref{psi-simultaneous} fixes the poles of $\psi_\pm$ in time, see \cite[Proposition 6.3]{McLaughlin2018}.
Thus, with $\psi_{\pm}(x,t;\lambda)$ we introduce a Baker-Akhiezer function $\Psi(x,t;P)$ on the Riemann surface with all the same properties as \eqref{eq:def-Psi} with the exception of the replacement of the asymptotics with
 \begin{align}
\Psi(x,t;P) = \ee^{\ii z(P) (x - x_0) + 4 \ii z(P)^3 t}( 1 + o(1)), \quad z(P) \to \infty.
 \end{align}
}

 
\subsection{Moving poles to the band endpoints}\label{s:move} The procedure described here resembles what was employed in the earlier works \cite{TrogdonFiniteGenus, Trogdon2013a} (see also \cite[Chapter 11]{TrogdonSOBook}). However, we make an observation that enables the treatment of the case when the genus $g$ is large.  For $z \in \mathbb C \setminus [\alpha_1,\infty)$ define
\begin{align*}
  \boldsymbol \Delta(\lambda; \mathbf d, \mathbf v) = \begin{bmatrix} \displaystyle \frac{\Theta(\mathcal A(\lambda) + \mathbf v - \mathbf d; \mathbf B)}{\Theta(\mathcal A(\lambda) - \mathbf d; \mathbf B)} & \displaystyle \frac{\Theta(-\mathcal A(\lambda) + \mathbf v - \mathbf d; \mathbf B)}{\Theta(-\mathcal A(\lambda) - \mathbf d; \mathbf B)} \end{bmatrix},
\end{align*}  
where $\mathcal A(\lambda) = \mathcal A ( \lambda, R(\lambda))$ is the Abel map restricted to the first sheet.  Note that $- \mathcal A(\lambda) = \mathcal A ( \lambda, -R(\lambda))$ is then the Abel map restricted to the second sheet.  The following properties of the theta function are now needed
\begin{align*}
  \Theta(\mathbf z + 2 \pi \ii \mathbf e_j; \mathbf B) &=  \Theta(\mathbf z; \mathbf B),\\
  \Theta(\mathbf z + \mathbf B \mathbf e_j; \mathbf B) &= \exp \left( - \frac 1 2 B_{jj} - z_j \right)   \Theta(\mathbf z; \mathbf B),
\end{align*}
where $\mathbf e_j$ is the $j$th column of the $g\times g$ identity matrix and $\mathbf B$ is the Riemann matrix.  Then note that
\begin{align*}
  \mathcal A^+(\lambda) + \mathcal A^-(\lambda) = \left( 2 \sum_{k = 1}^{j-1} \int_{\beta_k}^{\alpha_{k+1}} \nu_\ell\right)_{\ell=1}^g = \left( \sum_{k = 1}^{j-1} \oint_{\mathfrak a_k}  \nu_\ell\right)_{\ell=1}^g = 2 \pi \ii \mathbf n, \quad \lambda \in (\alpha_j,\beta_j),
\end{align*}
for a vector $\mathbf n$ of ones and zeros.  Then we compute
\begin{align*}
  \mathcal A^+(\lambda) - \mathcal A^-(\lambda) = \left( 2 \sum_{k = 1}^{j} \int_{\alpha_k}^{\beta_{k}} \nu_\ell\right)_{\ell=1}^g = \left(  \oint_{\mathfrak b_{j}} \nu_\ell\right)_{\ell=1}^g = \mathbf B \mathbf e_{j}, \quad \lambda \in (\beta_j,\alpha_{j+1}).
\end{align*}
Note that from \eqref{eq:def-Psi}, for a non-special divisor $D = \sum_{j=1}^g P_j$,
 \begin{align}
   \Theta(\mathcal A(P) - \mathcal A(D) - \mathbf k; \mathbf B) = 0 \quad \text{if and only if} \quad P \in \{P_1,\ldots,P_g\},
 \end{align}
 where, as before, $\mathbf k$ is the vector of Riemann constants with base point $\alpha_1$.  So, for two non-special divisors $D= \sum_{j=1}^g P_j$ and $D'= \sum_{j=1}^g P_j'$ we choose
 {$\mathbf{v}=\mathbf{v}(D,D')$ and $\mathbf{d}=\mathbf{d}(D,D')$ by}
 \begin{align}
   \mathbf v - \mathbf d &= - \mathcal A(D) - \mathbf k,\label{v-minus-d}\\
   - \mathbf d &= - \mathcal A(D') - \mathbf k.\label{just-d}
 \end{align}
 For $P = (\lambda,w) \in \Sigma$, define $\pi(\lambda,w) = \lambda$.  Then it follows that if $P_j$ is on the first (second) sheet of $\Sigma$ then $\boldsymbol \Delta(\lambda;\mathbf d,\mathbf v)$ has a zero at $\pi(P_j)$ in its first (second) column.  Similarly, if $P_j'$ is on the first (second) sheet of $\Sigma$ then $\boldsymbol \Delta(\lambda;\mathbf d,\mathbf v)$ has a pole at $\pi(P_j')$ in its first (second) column.

 Now suppose that $\lambda$ is not a pole of either column of $\boldsymbol \Delta$.  Then
 \begin{align*}
  \boldsymbol \Delta^+(\lambda; \mathbf d, \mathbf v) = \begin{cases} \boldsymbol\Delta^-(\lambda; \mathbf d, \mathbf v) 
  {\sigma_1}
&  {\lambda \in\left(\alpha_{g+1},+\infty\right)\cup\left(\bigcup_{j=1}^{g}\left(\alpha_{j}, \beta_{j}\right)\right)},\\ 
    \boldsymbol\Delta^-(\lambda; \mathbf d, \mathbf v) \begin{bmatrix} \ee^{-v_j} & 0 \\ 0 & \ee^{v_j} \end{bmatrix} & \lambda \in (\beta_j,\alpha_{j+1}),\\
     \boldsymbol \Delta^-(\lambda; \mathbf d, \mathbf v) & \lambda \in (-\infty,\alpha_1).   \end{cases}
\end{align*}

 Choose the divisor
 \begin{align}\label{eq:Dprime}
   D' = (\alpha_2,0) + (\alpha_3,0) + \cdots + (\alpha_{g+1},0),
 \end{align}
 and let $D$ be the divisor of the poles of the Baker-Akhiezer function $\Psi(P,t)$. Then {consider $\mathbf{v}=\mathbf{v}(D,D')$ and $\mathbf{d}=\mathbf{d}(D,D')$ with these choices as in \eqref{v-minus-d} and \eqref{just-d}}. The function
 \begin{align}
   {\Xi}(P;x,t) := \Psi(P,t) \frac{\Delta(\infty;\mathbf d, \mathbf v)}{\Delta(P;\mathbf d, \mathbf v)}
   \label{Xi-def}
 \end{align}
 {now has poles at the right endpoints of the gaps, namely  the points where $\lambda=\alpha_2,\alpha_3,\ldots,\alpha_{g+1}$.}
We arrive at the following proposition.
 \begin{proposition}\label{p:xi}
   The sectionally analytic vector-valued function
   \begin{align}
     \boldsymbol{\Xi}(\lambda;x,t) =: \begin{bmatrix} \xi_{+}(\lambda;x,t) & \xi_{-}(\lambda;x,t) \end{bmatrix}
   \end{align}
   satisfies the following jump conditions away from poles
   \begin{align}
     \boldsymbol{\Xi}^{+}(\lambda;x,t) = \boldsymbol{\Xi}^{-}(\lambda;x,t) \sigma_1, \quad \lambda \in (\alpha_{g+1}, +\infty) \cup \left( \bigcup_{j=1}^g (\alpha_j, \beta_j) \right),\\
     \boldsymbol{\Xi}^{+}(\lambda;x,t) = \boldsymbol{\Xi}^{-}(\lambda;x,t) \begin{bmatrix} \ee^{v_j} & 0 \\ 0 & \ee^{-v_j} \end{bmatrix}, \quad \lambda \in (\beta_j,\alpha_{j+1})
   \end{align}
   where $\boldsymbol{\Xi}^{\pm}(\lambda;x,t) = \lim_{\epsilon \downarrow 0} \boldsymbol{\Xi}^{\pm}(\lambda \pm \ii \epsilon;x,t)$.  The asymptotics
   \begin{align}
      \boldsymbol{\Xi}(\lambda;x,t)  \ee^{-(\ii \sqrt{\lambda} x + 4\ii \lambda^{3/2} t)\sigma_3} = \begin{bmatrix} 1 & 1 \end{bmatrix} + O\left( \frac{1}{\sqrt{\lambda}}\right), \quad |\lambda| \to \infty,
   \end{align}
   also hold.
 \end{proposition}

\begin{figure}[h]
\includegraphics{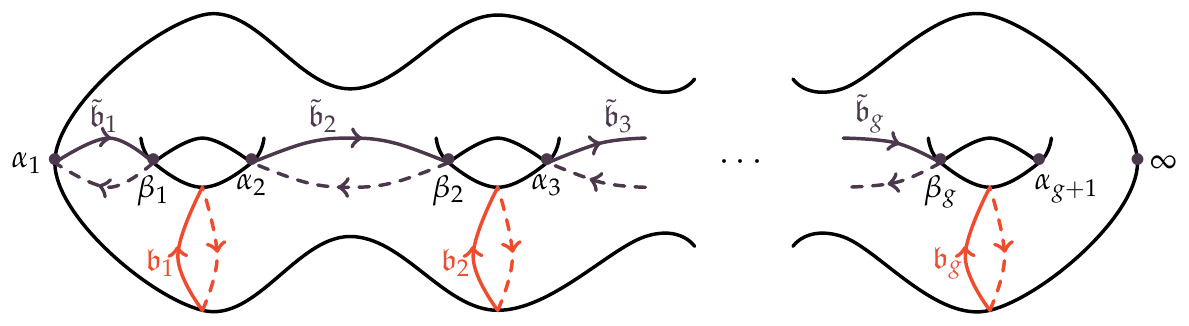}
\caption{An illustration of the choices for the $\mathfrak{b}$- and the $\tilde{\mathfrak{b}}$-cycles on the hyperelliptic Riemann surface $\RS$.}
\label{f:Riemann-Surface-Cycles-tilde}
\end{figure}

\subsection{The Riemann--Hilbert problem} \label{s:form_RHP}
Let ${\boldsymbol\Psi}(\diamond;x,t) \colon \mathbb{C}^{1\times 2}\setminus \sigma_\mathrm{B}(q) \to \mathbb{C}$ denote the row vector planar representation of the Baker-Akhiezer function associated with $q(x,t)$:
\begin{equation}
{\boldsymbol\Psi}(\lambda;x,t) := \begin{bmatrix} {\psi_+}(\lambda;x,t) & {\psi_-}(\lambda;x,t) \end{bmatrix},
\label{bold-Psi}
\end{equation}
which satisfies the ``twist'' jump condition
\begin{equation}
{\boldsymbol\Psi}^+(\lambda;x,t)={\boldsymbol\Psi}^-(\lambda;x,t) \sigma_1
,\quad \lambda \in\Sigma
\end{equation}
and has the asymptotic behavior
\begin{equation}
{\boldsymbol\Psi}(\lambda;x,t) = \begin{bmatrix} \ee^{\ii \lambda^{\frac{1}{2}}(x + 4 \lambda t)}& \ee^{-\ii \lambda^{\frac{1}{2}}(x + 4 \lambda t)}\end{bmatrix}\left(\mathbb{I} + o(1) \right),\quad \lambda\to\infty.
\end{equation}
Here the power function $\lambda\mapsto \lambda^{\frac{1}{2}}$ is defined to be analytic on $\mathbb{C}\setminus [ 0, +\infty)$, satisfying $\lambda^{\frac{1}{2}} = \ii |\lambda|^{\frac{1}{2}} + o(1)$ as $\lambda\to -\infty$. Set $\theta(\lambda;x,t):= \lambda^{\frac{1}{2}}(x + 4 \lambda t)$ and observe that $\theta(\lambda;x,t)$ has a jump discontinuity across the half-line $(0, +\infty)$, which we orient from $\lambda=0$ to $\lambda=+\infty$.
Based on the considerations in the previous section, we assume that the poles of $\mathbf \Psi$ occur at the points in the divisor \eqref{eq:Dprime}.

Define the renormalized row-vector-valued function
\begin{equation}\label{eq:Mdef}
\mathbf{M}(\lambda;x,t) := {\boldsymbol\Psi}(\lambda;x,t) \ee^{-\ii \theta(\lambda;x,t)\sigma_3}.
\end{equation}
As $\theta^+(\lambda;x,t) + \theta^-(\lambda;x,t) = 0$ for $\lambda\in [0,+\infty)$, the jump conditions satisfied by $\mathbf{M}(\lambda;x,t)$ take the form 
\begin{alignat}{2}
\mathbf{M}^+(\lambda;x,t) &= \mathbf{M}^-(\lambda;x,t)\sigma_1,&&\quad \lambda\in (\alpha_{g+1}, +\infty) \cup \left( \bigcup_{j=1}^g (\alpha_j, \beta_j)\right),\\
\mathbf{M}^+(\lambda;x,t) &= \mathbf{M}^-(\lambda;x,t) \ee^{-2\ii  \theta^+(\lambda;x,t)\sigma_3}
,&&\quad \lambda\in\bigcup_{j=1}^g (\beta_j, \alpha_{j+1}),
\end{alignat}
and $\mathbf{M}(\lambda;x,t)$ satisfies
\begin{equation}
\mathbf{M}(\lambda;x,t) = \begin{bmatrix} 1 & 1 \end{bmatrix}(\mathbb{I} + o(1)), \quad \lambda\to\infty.
\end{equation}

\begin{remark}\label{r:dress}
An important calculation to make here is to define
\begin{align*}
    \mathbf{K}(z;x,t) = \begin{cases} 
    \mathbf{M}(z^2;x,t) & \Im z > 0,\\
    \mathbf{M}(z^2;x,t) \sigma_1 & \Im z < 0. \end{cases}
\end{align*}
Then apply \cite[Theorem~2.1]{TrogdonDressing} to see that $\mathbf K$ (and therefore $\mathbf M$, and hence $\Psi(x,t;P)$) is a simultaneous solution of an appropriate version of the Lax pair for the KdV equation.
\end{remark}

To control the oscillatory factors in the jump matrices above, we seek a function $G(\lambda;x,t)$ that is analytic for $\lambda\in \mathbb{C}\setminus [0, +\infty)$ satisfying
\begin{alignat}{2}
G^+(\lambda;x,t) + G^-(\lambda;x,t) &= 0, &&\quad \lambda\in (\alpha_{g+1}, +\infty) \cup \left( \bigcup_{j=1}^g (\alpha_j, \beta_j)\right)\label{G-jump-bands}\\
G^+(\lambda;x,t) - G^-(\lambda;x,t) + 2 \theta^+(\lambda;x,t)&= \Omega_j, &&\quad \lambda\in (\beta_{j}, \alpha_{j+1}), \quad j=1,2,\ldots,g, \label{G-jump-gaps}
\end{alignat}
for some constants $\Omega_j$, and normalized to satisfy $G(\lambda) \to 0$ as $\lambda\to \infty$. It is easy to see that
\begin{equation}
G(\lambda;x,t) := \frac{R(\lambda)}{2\pi \ii}\sum_{j=1}^{g}\int_{\beta_{j}}^{\alpha_{j+1}}\frac{\Omega_j -2\theta^+(\zeta;x,t) }{R(\zeta)(\zeta-\lambda)}\dd \zeta
\end{equation}
is analytic for $\lambda\in \mathbb{C}\setminus [0, +\infty)$, admits continuous boundary values on $[0, +\infty)$ which satisfy the jump conditions \eqref{G-jump-bands}--\eqref{G-jump-gaps}. Observe that
\begin{equation}
\frac{G(\lambda)}{R(\lambda)} = \sum_{k=1}^{g} m_k(x,t) \lambda^{-k} + O(\lambda^{-g-1}), \quad \lambda\to \infty, 
\label{G-over-R-expand}
\end{equation}
where
\begin{equation}\label{eq:mk}
m_k(x,t) :=  -\frac{1}{2\pi \ii} \sum_{j=1}^{g}\int_{\beta_{j}}^{\alpha_{j+1}} \left( \Omega_j  - 2\theta^+(\zeta;x,t) \right) \frac{\zeta^{k-1}}{R(\zeta)} \dd \zeta, \quad k=1,2,\ldots,g.
\end{equation}
Thus, in order to have $G(\lambda) = o(1)$ as $\lambda\to \infty$, we need to have $m_k\equiv 0$ for $k=1,2,\ldots,g$, which yields the conditions
\begin{equation}
\sum_{j=1}^{g}  \Omega_j  \int_{\beta_{j}}^{\alpha_{j+1}} \frac{\zeta^{k-1}}{R(\zeta)} \dd \zeta = \sum_{j=1}^{g}\int_{\beta_{j}}^{\alpha_{j+1}} 2\theta^+(\zeta;x,t) \frac{\zeta^{k-1}}{R(\zeta)} \dd \zeta, \quad k=1,2,\ldots,g.
\label{Omega-linear-system}
\end{equation}
This is a linear system of $g$ equations for the constants $\Omega_j=\Omega_j(x,t)$, $j=1,2,\ldots, g$.  Taking a linear combination of these equations we can instead consider
\begin{equation}
\sum_{j=1}^{g}  \Omega_j  \int_{\beta_{j}}^{\alpha_{j+1}} \frac{p_{k-1}(\zeta)}{R(\zeta)} \dd \zeta = \sum_{j=1}^{g}\int_{\beta_{j}}^{\alpha_{j+1}} 2\theta^+(\zeta;x,t) \frac{p_{k-1}(\zeta)}{R(\zeta)} \dd \zeta, \quad k=1,2,\ldots,g,
\label{Omega-linear-system-general}
\end{equation}
for any basis $\{p_0,\ldots,p_{g-1}\}$ for polynomials of degree at most $g-1$.
Taking into account the orientation of the $\mathfrak{a}$-cycles depicted in Figure~\ref{f:Riemann-Surface-Cycles} and the sign change that occurs from passing from one sheet to the other, we have
\begin{equation}
 \int_{\beta_{j}}^{\alpha_{j+1}} \frac{p_{k-1}(\zeta)}{R(\zeta)} \dd \zeta =  \int_{\beta_{j}}^{\alpha_{j+1}} \frac{p_{k-1}(\zeta)}{w} \dd \zeta = -\frac{1}{2}\oint_{\mathfrak{a}_j} \omega_k = -\frac{1}{2} A_{kj},\quad k=1,2,\ldots,g,
\end{equation}
from \eqref{A-B-period-matrices}, choosing $p_{j-1}$ so that $p_j(\zeta)/R(\zeta) \dd \zeta = \omega_j$, and hence \eqref{Omega-linear-system} reads
\begin{equation}
\sum_{j=1}^{g} A_{kj} \Omega_j  =  - 4  \sum_{j=1}^{g}\int_{\beta_{j}}^{\alpha_{j+1}} \theta^+(\zeta;x,t) \frac{p_{k-1}(\zeta)}{R(\zeta)} \dd \zeta, \quad k=1,2,\ldots,g.
\end{equation}
Thus, the coefficient matrix for the linear system \eqref{Omega-linear-system} is nothing but a constant multiple of the matrix $\mathbf{A}$ of $\mathfrak{a}$-cycles of the basis of differentials $\{\omega_k \}_{k=1}^g$, which is nonsingular. Therefore, the system \eqref{Omega-linear-system} is uniquely solvable. Moreover, because $\theta^+(\zeta;x,t)$ is real-valued and $R(\zeta)$ is purely imaginary for $\zeta \in (\beta_j, \alpha_{j+1})$, $j=1,2,\ldots,g$, it follows that $\Omega_j$, $j=1,2,\ldots,g$, are all real valued.
Using the basis of differentials in \eqref{eq:stable-dif} results in a linear system which can be solved in a numerically stable fashion as $g$ becomes large.  

\begin{remark}
The computation of the integrals that appear in this section and the computation of the Abel map is discussed in \cite{TrogdonFiniteGenus}.  There is a numerical subtlety here.  If one computes the Abel map $\mathcal A(\lambda)$ for $\lambda$ near a branch point and $\lambda$ is known to within an error $\epsilon$, that error may be amplified to be on the order of $\sqrt{\epsilon}$.  So, if $q_0$ is such that $\gamma_j(x)$ is near a branch point, the computation of the of $v_j$ in Proposition~\ref{p:xi} may suffer increased errors.  In practice, one can choose $x_0$ in the initial scattering theory to move this away from the branch point.  Various schemes can be employed to find a good choice of $x_0$.  Choosing $x_0$ randomly is often sufficient.
\end{remark}

Now define
\begin{equation}\label{eq:N}
\mathbf{N}(\lambda;x,t) :=\mathbf{M}(\lambda;x,t) \ee^{-\ii G(\lambda;x,t)\sigma_3},
\end{equation}
and observe that $\mathbf{N}(\lambda;x,t)$ satisfies the following jump conditions:
\begin{alignat}{2}
\mathbf{N}^+(\lambda;x,t) &= \mathbf{N}^-(\lambda;x,t)\sigma_1,&&\quad \lambda\in (\alpha_{g+1}, +\infty) \cup \left( \bigcup_{j=1}^g (\alpha_j, \beta_j)\right),\\
\mathbf{N}^+(\lambda;x,t) &= \mathbf{N}^-(\lambda;x,t) \ee^{-\ii  \Omega_j(x,t)\sigma_3}
,&&\quad \lambda\in(\beta_j, \alpha_{j+1}),\quad j =1,2,\ldots,g.
\end{alignat}

\section{A singular integral equation for the solution of the a Riemann-Hilbert problem on cuts}
\label{s:adaptive-Cauchy}

In this section we describe a numerical method to approximate \eqref{eq:N}. But before we do that, we need to make one more transformation to remove the non-trivial jump that has infinite extent.  Recall that from the discussion following Remark~\ref{r:symmetry} that we take $\alpha_1 = 0$ without loss of generality.  Define
\begin{align}
  \mathbf{S}(z;x,t) := \begin{cases} \mathbf{N}(z^2;x,t) & \Im z > 0,\\
    \mathbf{N}(z^2;x,t) \sigma_1 & \Im z < 0. \end{cases}
\end{align}
For convenience we write 
\begin{align}
\alpha_j &=: a_j^2,\quad a_j\geq 0, \quad j=1,2,\ldots,g+1,\\
\beta_j &=: b_j^2,\quad b_j>0,\quad j=1,2,\ldots,g.
\end{align}
Note that for $\Re (z) > 0$ if $\Im (z) > 0$ then $\Im (z^2) > 0$.  And similarly, if $\Re (z) > 0$,  $\Im (z) < 0$ implies $\Im (z^2) < 0$.  And if $\Re (z) < 0$, then these implications are flipped.  From this we find that $\mathbf{S}(z;x,t)$ only has jumps on the (symmetric) collection of intervals
\begin{align}
  (b_j, a_{j+1}), \quad (-a_{j+1},-b_j), \quad j =1,2,\ldots, g,
  \label{intervals}
\end{align}
where it satisfies:
\begin{alignat}{2}
  \mathbf{S}^+(z;x,t) &=  \mathbf{N}^+(z^2;x,t) = \mathbf{N}^-(z^2;x,t) \ee^{- \ii \Omega_j(x,t) \sigma_3} = \mathbf{S}^-(z;x,t) \sigma_1 \ee^{- \ii \Omega_j(x,t) \sigma_3}, &&\quad z \in (b_j, a_{j+1}), \\
  \mathbf{S}^+(z;x,t) &=  \mathbf{N}^-(z^2;x,t) = \mathbf{N}^+(z^2;x,t) \ee^{\ii \Omega_j(x,t) \sigma_3} = \mathbf{S}^-(z;x,t) \sigma_1 \ee^{ \ii \Omega_j(x,t) \sigma_3}, &&\quad z \in (-a_{j+1},-b_j),
\end{alignat}
{where we have reoriented the intervals $(-a_{j+1},-b_j)$ so that all of the intervals in \eqref{intervals} are oriented from their left endpoint to the right endpoint.}
{Moreover}, $\mathbf{S}(z;x,t)$ is normalized so that $\mathbf{S}(z;x,t) = \begin{bmatrix} 1 & 1 \end{bmatrix} + O(z^{-1})$ as $|z| \to \infty$.  

Next, we will want a formula to recover $q(x,t)$, the solution of the KdV equation \eqref{KdV}, directly from a representation of $\mathbf{S}(\lambda;x,t)$ as function of $\lambda$.  As $|z| \to \infty$, write
\begin{align}\label{eq:S}
\mathbf{S}(z;x,t) = \begin{bmatrix} 1 & 1 \end{bmatrix} + \frac{1}{z}\begin{bmatrix} s_1(x,t)& s_2(x,t) \end{bmatrix} + O(z^{-2}).
\end{align}
Supposing this limit is taken in the upper-half of the $z$-plane, this then implies that
\begin{align}
  \mathbf{N}(\lambda;x,t) = \begin{bmatrix} 1 & 1 \end{bmatrix} + \frac{1}{\sqrt{\lambda}}\begin{bmatrix} s_1(x,t)& s_2(x,t) \end{bmatrix} + O(\lambda^{-1}),\quad \lambda\to\infty.
\end{align}
Then we recall that
\begin{align}
  \mathbf{N}(\lambda;x,t) = \boldsymbol{\Psi}(\lambda;x,t) \ee^{-\ii(G(\lambda;x,t) + \theta(\lambda;x,t)) \sigma_3},
\end{align}
where the entries of the $1\times 2$ vector $\boldsymbol{\Psi}(\lambda,x,t)$ are solutions of $\mathcal L(t) \psi = \lambda \psi$, see \eqref{bold-Psi}.  So, consider the function $m(\lambda;x,t) = \psi^+(\lambda;x,t)e^{-\ii \theta(\lambda;x,t)}$:
\begin{align}
  \partial_x m(\lambda;x,t) &= \partial_x\psi^+(\lambda;x,t)e^{-\ii \theta(\lambda;x,t)} - \ii \sqrt{\lambda} \psi^+(\lambda;x,t)e^{-\ii \theta(\lambda;x,t)},\\
   \partial_{xx} m(\lambda;x,t) &= \partial_{xx}\psi^+(\lambda;x,t)e^{-\ii \theta(\lambda;x,t)} - 2\ii \sqrt{\lambda} \partial_x\psi^+(\lambda;x,t)e^{-\ii \theta(\lambda;x,t)} - \lambda \psi^+(\lambda;x,t)e^{-\ii \theta(\lambda;x,t)}.
\end{align}
Adding these so as to eliminate the $\partial_x \psi^+$ term, we find
\begin{equation}
\begin{aligned}
   \partial_{xx} m(\lambda;x,t) + 2 \ii \sqrt{\lambda} \partial_x m(\lambda;x,t) &= 
   \partial_{xx}\psi^+(\lambda;x,t)e^{-\ii \theta(\lambda;x,t)} + \lambda \psi_{+}(\lambda; x,t) \ee^{-\ii \theta(\lambda;x,t)} \\
   &= - q(x,t) m(\lambda;x,t).
  \end{aligned}
\end{equation}
It follows from \eqref{eq:ratio-theta} that both $\partial_{xx} m(\lambda;x,t)$ and $\partial_{x} m(\lambda;x,t)$ decay at infinity, giving the recovery formula
\begin{align}
  -\lim_{\lambda \to \infty} 2 \ii \sqrt{\lambda} \partial_x m(\lambda;x,t) =  q(x,t).
\end{align}
In other words, we have {as $\lambda\to\infty$}
\begin{align}
  \mathbf{N}(\lambda,x,t) = \begin{bmatrix} 1 & 1 \end{bmatrix} + \frac{1}{2 \ii \sqrt{\lambda}} \begin{bmatrix} -\int^x q(s,t) \dd s  & \int^x q(s,t) \dd s \end{bmatrix} - \frac{\ii}{\sqrt{\lambda}} \begin{bmatrix} m_{g+1}(x,t) & - m_{g+1}(x,t) \end{bmatrix} + O(\lambda^{-1}),
\end{align}
where $m_{g+1}(x,t)$ {denotes the coefficient of the term proportional to $\lambda^{-g-1}$ in the expansion \eqref{G-over-R-expand}}.
Thus, we arrive at
\begin{align}
  q(x,t) = -2 \ii \partial_x s_1(x,t) + 2\partial_x m_{g+1}(x,t).
  \label{recovery-q}
\end{align}

\subsection{Weighted spaces}

We now formulate a singular integral equation on a direct sum of weighted $L^2$ spaces. 
Define
\begin{align}
  I_j := \begin{cases}  (b_j, a_{j+1}) & j \in \{1,2,\ldots,g\},\\
    (-a_{j+1},-b_j) & j \in \{-1,-2,\ldots,-g\}. \end{cases}
\end{align}
Then set, for ${y} \in I_j$
\begin{equation}
  w_j({y}) :=  \begin{cases}\dfrac 1 \pi \sqrt{\dfrac{y - b_j}{a_{j+1} - y}} & j \in \{1,2,\ldots,g\},\\ \\
   \dfrac 1 \pi \sqrt{\dfrac{ - b_j - y}{y + a_{j+1}}} & j \in \{-1,-2,\ldots,-g\}, \end{cases}
   \label{weights}
\end{equation}
where {each weight} $w_j(y)$ is understood to vanish outside its domain of definition.  
Then define
\begin{align}
  L^2_w\left( \bigcup_j I_j\right) := \bigoplus_{\substack{j=-g\\ j\neq 0}}^g L^2_{w_j} (I_j), \qquad w :=  \sum_j w_j.
\end{align}
It is convenient to order the component functions (each of which is {$2\times 1$ row-}vector-valued) for $\mathbf{U} \in  L^2_w\left( \bigcup_j I_j\right)$ as $\mathbf U = (\mathbf{U}_1,\mathbf{U}_{-1},\mathbf{U}_2,\mathbf{U}_{-2},\ldots,\mathbf{U}_{g},\mathbf{U}_{-g})$.  Define the operators
\begin{align}
  \mathcal R_j \mathbf{U} := \mathbf{U} J_j, \quad (\mathcal W \mathbf{U})|_{I_j} := w_j^{-1} \mathbf{U}|_{I_j}
\end{align}
i.e., right multiplication by the jump matrix 
{
\begin{align}
\mathbf{J}_j := \sigma_1 \ee^{-\mathrm{sgn}(j) \Omega_{|j|}(x,t) \sigma_3},\label{jump-Jc}
\end{align}
}
and division by the weight on $I_j$, respectively. 

Suppose\footnote{This suffices for our purposes, but in general one can consider Carleson curves \cite{Bottcher1997}.} $\Gamma$ is a union of line segments.  
{For a weight function $w\colon \Gamma \to [0,\infty)$ supported on $\Gamma$, define the weighted $L^p$ space
\begin{equation}
L_{w}^{p}(\Gamma):=\left\{f: \Gamma \rightarrow \mathbb{C}: \int_{\Gamma}|f(z)|^{p} w(z)|d z|<\infty\right\}
\end{equation}
and the weighted Cauchy transform
\begin{equation}
\mathcal{C}_{\Gamma, w} u(z)=\frac{1}{2 \pi \ii} \int_{\Gamma} \frac{u\left(z^{\prime}\right)}{z^{\prime}-z} w\left(z^{\prime}\right) d z^{\prime}, \quad z \in \mathbb{C} \backslash \Gamma .
\label{weighted-Cauchy}
\end{equation}
We define the boundary values of \eqref{weighted-Cauchy} whenever the following limits exist:
\begin{equation}
    \mathcal{C}_{\Gamma, w}^{\pm} u(z)=\lim _{\epsilon \downarrow 0} \mathcal{C}_{\Gamma, w} u(z \pm \ii \epsilon), \quad z \in \Gamma.
\end{equation}
When the domain of the weight $w$ is clear from context we write $\mathcal{C}_{w}, \mathcal{C}_{w}^{\pm}$. When $w \equiv 1$ we write $\mathcal{C}_{\Gamma}$, $\mathcal{C}_{\Gamma}^{\pm}$.
These operators are understood to apply to vectors component-wise.}

\begin{definition}\label{d:precise}
  A function $\mathbf{S}(z;x,t)$ is a solution of the Riemann-Hilbert problem
  \begin{align}
    \mathbf{S}^+(z;x,t) &= \mathbf{S}^-(z;x,t) \mathbf{J}_j, \quad z\in I_j,\label{eq:jc}\\
    \mathbf{S}(\infty;x,t) &= \mathbf C \in \mathbb C^{m \times 2},\notag
  \end{align}
  if
  \begin{align}
    \mathbf{S}(z;x,t) = \mathbf C + \sum_{\substack{j=-g\\ j \neq 0}}^g \mathcal C_{w_j} u_j(z),
  \end{align}
  for $\mathbf{U} \in L^2_w\left( \bigcup_j I_j\right)$ and the jump condition \eqref{eq:jc} is satisfied for a.e. $z \in I_j$ for each $j$.   Further, for $j \neq k$ we use the notation
\begin{align}
  \mathcal C_{w_j}\Big\rvert_{I_k} \mathbf{U} = (\mathcal C_{w_j}\mathbf{U} )\Big\rvert_{I_k}.
\end{align}
\end{definition}

\begin{theorem}\label{t:sufficient}
  Suppose $\mathbf{S}(z;x,t)$ satisfies the following\footnote{Here $\|\cdot\|$ is any norm on $\mathbb C^{m \times 2}$.}
  \begin{enumerate}
  \item For some $1 < p < 2$
    \begin{align}
      \sup_{\rho > 0} \int_{-\infty}^\infty \| \mathbf{S}(z \pm \ii \rho ;x,t) - \mathbf C\|^p \dd z < \infty.
    \end{align}
    
  \item The jump condition \eqref{eq:jc} is satisfied for a.e. $z \in I_j$ for each $j$.
  \item $\mathbf{U} = \mathcal W(\mathbf{S}^+ - \mathbf{S}^-) \in L^2_w\left( \bigcup_j I_j\right)$.
  \end{enumerate}
  Then $\mathbf{S}(z;x,t)$ is a solution of the Riemann-Hilbert problem in the sense of Definition~\ref{d:precise} with $u = \mathcal W(\mathbf{S}^+ - \mathbf{S}^-)$.
\end{theorem}
\begin{proof}
  The first condition imposes that $\mathbf{S} - \mathbf C$ is an element of the Hardy space of the upper-half and lower-half planes \cite{Duren}.  This implies that the boundary values from above and below exist a.e.  Furthermore, it also implies that $\mathbf{S}(z;x,t)$ is given by the Cauchy integral of its boundary values:
  \begin{align}
    \mathbf{S}(z;x,t)  = \mathbf C +  \sum_{\substack{j=-g\\ j \neq 0}}^g \mathcal C_{w_j} \mathbf{U}_j(z).
  \end{align}
\end{proof}

Imposing the jump condition $\mathbf S^+(z;x,t) = \mathbf S^-(z;x,t) \mathbf{J}_j$, $z \in I_j$ for each $j$ results in the following system of singular integral equations that are satisfied by $\mathbf{U}_k$, $k \in \{ \pm 1, \pm 2, \ldots, \pm g\}$,
\begin{align}\label{eq:SIE}
  \mathcal C_{w_k}^+ \mathbf{U}_k(z) + \sum_{\substack{j=-g\\ j \neq 0,k}}^g \mathcal C_{w_j} \mathbf{U}_j(z) -  \left[\mathcal C_{w_k}^- \mathbf{U}_k(z) + \sum_{\substack{j=-g\\ j \neq 0,k}}^g \mathcal C_{w_j} \mathbf{U}_j(z)\right] J_k = \begin{bmatrix} 1 & 1 \end{bmatrix} (\mathbf{J}_k - \mathbb{I}), \quad z \in I_k.
\end{align}
It is important to note that $\mathcal C_{w_j} \mathbf{U}_j(z) = \mathcal C_{w_j}^- \mathbf{U}_j(z) = \mathcal C_{w_j}^+ \mathbf{U}_j(z)$ if $z \not\in I_j$.

Then we consider the following block operator on $L^2_w\left( \bigcup_j I_j\right)$

\resizebox{.99\hsize}{!}{
\begin{minipage}{\linewidth}
  \begin{align*}
  &\mathcal S := \\
& \begin{bmatrix} \mathcal C_{w_1}^+ - \mathcal R_1 \circ  \mathcal C_{w_1}^- &  (I - \mathcal R_1) \circ  \mathcal C_{w_{-1}}|_{I_1} & (I - \mathcal R_1)  \circ \mathcal C_{w_{2}}|_{I_1} & (I - \mathcal R_1)  \circ \mathcal C_{w_{-2}}|_{I_1} & \cdots & (I - \mathcal R_1)  \circ \mathcal C_{w_{-g}}|_{I_1}\\
   (I - \mathcal R_{-1})  \circ \mathcal C_{w_{1}} |_{I_{-1}} & \mathcal C_{w_{-1}}^+ - \mathcal R_{-1} \circ  \mathcal C_{w_{-1}}^- &  (I - \mathcal R_{-1})  \circ \mathcal C_{w_{2}}|_{I_{-1}} & (I - \mathcal R_{-1})  \circ \mathcal C_{w_{-2}}|_{I_{-1}} & \cdots & (I - \mathcal R_{-1})  \circ \mathcal C_{w_{-g}}|_{I_{-1}} \\
   (I - \mathcal R_{2})  \circ \mathcal C_{w_{1}} |_{I_{2}} & (I - \mathcal R_{2})  \circ \mathcal C_{w_{-1}}|_{I_{2}} & \mathcal C_{w_2}^+ - \mathcal R_2 \circ  \mathcal C_{w_2}^- & (I - \mathcal R_{2})  \circ \mathcal C_{w_{-2}}|_{I_{2}} & \cdots & (I - \mathcal R_{2})  \circ \mathcal C_{w_{-g}}|_{I_{2}} \\
   (I - \mathcal R_{-2})  \circ \mathcal C_{w_{1}}|_{I_{-2}} & (I - \mathcal R_{-2})  \circ \mathcal C_{w_{-1}}|_{I_{-2}} & (I - \mathcal R_{-2})  \circ \mathcal C_{w_{2}}|_{I_{-2}} & \mathcal C_{w_{-2}}^+ - \mathcal R_{-2} \circ  \mathcal C_{w_{-2}}^- & \cdots & (I - \mathcal R_{-2})  \circ \mathcal C_{w_{-g}}|_{I_{-2}}\\
   \vdots & \vdots & \vdots & \vdots & \ddots & \vdots\\
    (I - \mathcal R_{-g})  \circ \mathcal C_{w_{1}}|_{I_{-g}} & (I - \mathcal R_{-g})  \circ \mathcal C_{w_{-1}}|_{I_{-g}} & (I - \mathcal R_{-g})  \circ \mathcal C_{w_{2}}|_{I_{-g}} & (I - \mathcal R_{-g})  \circ \mathcal C_{w_{-2}}|_{I_{-g}}  & \cdots & \mathcal C_{w_{-g}}^+ - \mathcal R_{-g} \circ  \mathcal C_{w_{-g}}^-
  \end{bmatrix}. 
  \end{align*}
\end{minipage}}\\

\noindent Note that $\mathcal S$ as an operator is completely described by $I_j$, $1 \leq j \leq g$ and $\Omega_j(x,t)$ for $1 \leq j \leq g$.  So, we write
\begin{align}
  \mathcal S = \mathcal S(I_1,\ldots, I_g; \Omega_1,\ldots, \Omega_g).
\end{align}

{We now state some observations that motivate the preconditioning we employ in solving \eqref{eq:SIE} numerically, which is described in Section~\ref{s:precond}. The linear system obtained from discretization of \eqref{eq:SIE} upon preconditioning ends up being extremely well-conditioned; see Figure~\ref{fig:precond}. First, one can prove the following lemma.
\begin{lemma}\label{l:diag_inv}
The operator $\mathcal W\mathrm{diag}(\mathcal S)$ is boundedly invertible on $L^2_w\left( \bigcup_j I_j\right)$.
\end{lemma}
The following is then immediate and is the heuristic that motivates the use of the aforementioned preconditioner in the numerical procedure.
\begin{lemma}\label{l:compact}
    $(\mathcal W\mathrm{diag}(\mathcal S))^{-1} \mathcal W \mathcal S - \mathcal{I}
  $, where $\mathcal{I}$ is the identity operator,
  is a compact operator  on $L^2_w\left( \bigcup_j I_j\right)$.
\end{lemma}
We will present the proof of Lemma~\ref{l:diag_inv} along with an analytical justification of the preconditioning and the convergence of the numerical method proposed in this work to solve \eqref{eq:SIE} in a forthcoming paper.
}

\section{Numerical Inverse Scattering}
\label{s:NIST}

In this section we develop a numerical method to solve the Riemann-Hilbert problem in Definition~\ref{d:precise}.  We consider the Chebyshev-V and Chebyshev-W polynomials which are also known as the Chebyshev polynomials of the third and fourth kind, respectively.  The polynomials $(V_n(y))_{n \geq 0}$ satisfy $0 < \lim_{y \to \infty} y^{-n} V_n(y) < \infty$ as well as
\begin{align}
  \int_{-1}^1 V_n(y) V_m(y)\sqrt{\frac{1+y}{1-y}} \frac{\dd y}{\pi} = \delta_{n,m},
\end{align}
for the Kronecker delta, $\delta_{n,m}$.  Similarly, the polynomials $(W_n(y))_{n \geq 0}$ satisfy $0 < \lim_{y \to \infty} y^{-n} W_n(y) < \infty$ as well as
\begin{align}
  \int_{-1}^1 W_n(y) W_m(y)  \sqrt{\frac{1-y}{1+y}} \frac{\dd y}{\pi}  = \delta_{n,m}.
\end{align}

For general $a < b$ we wish to find a basis of polynomials on $[a,b]$ using the transformation $T_{a,b}(y) = \frac{b -a}{2} y + \frac{b+a}{2}$, $T_{a,b} : [-1,1] \to [a,b]$.  Taking into account the singularity structure of the weights $w_j$ as defined {in \eqref{weights}}, for $a > 0$ define
\begin{align}
P_n(y;[a,b]) :=  V_n(T^{-1}_{a,b}(y)),\quad n \geq 0, 
\end{align}
which are orthogonal (but not normalized) polynomials on $[a,b]$ with respect to ${w_{a,b}(y) = \frac{1}{\pi} \sqrt{ \frac{y-a}{b-y}}}$. Similarly, for $b < 0$ define
\begin{align}
P_n(y;[a,b]) := W_n(T^{-1}_{a,b}(y)),\quad n \geq 0,
\end{align}
which are orthogonal polynomials on $[a,b]$ with respect to ${w_{a,b}(y) = \frac{1}{\pi} \sqrt{ \frac{b-y}{y-a}}}$.

This construction has the convenient benefit that for $w(y) = \frac{1}{\pi} \sqrt{ \frac{1-y}{y + 1}}$ defined on $[-1,1]$ and $\tilde w(y) = \frac{1}{\pi} \sqrt{ \frac{b-y}{y - a}}$ defined on $[a,b]$ {(for the case $a<b<0$)} we have
\begin{align}
  \mathcal C_{\tilde w} u (z) = C_{w} (u \circ T_{a,b}) ( T^{-1}_{a,b}(z)), \quad z \not \in [a,b].
\end{align}
{The same identity also holds for the case $b>a>0$.}
In other words, Cauchy integrals over general intervals with these weights can be computed by first mapping a function to the interval $[-1,1]$, computing the Cauchy integral for the mapped function and then mapping back.  We do note that
\begin{align}
  \int_a^b P_n(y;[a,b])^2 w_{a,b}(y) \dd y = \frac{b-a}{2}, \quad n \geq 0.
\end{align}

\subsection{Computing Cauchy integrals}

As is well-known, real orthonormal polynomials $(p_n)_{n \geq 0}$ (with positive leading coefficients), on the real axis, with respect to a probability measure $\mu$ satisfy a three-term recurrence relation
\begin{align}
  {y} p_n(y) = {A}_n p_n(y) + {B}_n p_{n+1}(y) + {B}_{n-1} p_{n-1},\\
  p_{-1}(y) \equiv 0, \quad p_0(y) \equiv 1, \quad {B}_{-1} = -1,
\end{align}
for recurrence coefficients ${(A_n)_{n \geq 0}}$, ${(B_n)_{n \geq 0}}$.  What is maybe less well-known is the weighted Cauchy transforms
\begin{align}
  c_n(z) = \mathcal C_{\mu} p_n(z) := \frac{1}{2 \pi \ii} \int_{\mathbb R} \frac{p_n(y)}{y - z} \mu(\dd y), \quad n \geq 0,
\end{align}
satisfy the same recurrence with different initial conditions, 
{and in particular}
\begin{align}
  c_{-1}(z) = \frac{1}{2 \pi \ii}, \quad c_0(z) = \frac{1}{2 \pi \ii} \int_{\mathbb R} \frac{\mu(\dd y)}{y - z}.
\end{align}

\begin{remark}
  For Chebyshev-V and Chebyshev-W polynomials we have, respectively,
  \begin{align}
    {A}_0 &= 1/2, \quad {A}_n = 0, \quad n \geq 1, \quad {B}_n = 1/2, \quad n \geq 0,\\
    {A}_0 &= -1/2, \quad {A}_n = 0, \quad n \geq 1, \quad {B}_n = 1/2, \quad n \geq 0.
  \end{align}
\end{remark}

There are some subtleties in solving the recurrence for the Cauchy transforms.  For $z$ in the complex plane, away from the support of $\mu$, $(p_n(z))_{n \geq 0}$ represents an exponentially growing solution of the three-term recurrence while $(c_n(z))_{n \geq 0}$  is an exponentially decreasing solution.  Thus, evaluating $c_n(z)$ by forward recurrence is inherently unstable.  Consider the case where $\mu(\dd y)$ has its support on $[-1,1]$.  In practice, the following is effective \cite{Olver2020}:
\begin{enumerate}
\item For $z$ inside a Bernstein ellipse with minor axis $O(1/n)$, solve for $c_n(z)$ by forward recurrence allowing one to easily compute the boundary values of $c_n$ on $[-1,1]$ from above and below.
\item For $z$ ouside a Bernstein ellipse with minor axis $O(1/n)$, solve the boundary value problem
  {\begin{align}
    \begin{bmatrix} A_0 -z & B_0 \\
      B_0 & A_1 -z  & B_1 \\
      & B_1 & A_2 -z & \ddots \\
      && \ddots & \ddots
   \end{bmatrix} \begin{bmatrix} c_0(z) \\ c_1(z) \\ \vdots \end{bmatrix} = \begin{bmatrix} \frac{1}{2 \pi \ii} \\ 0 \\ \vdots \end{bmatrix},
  \end{align}
  }
  with the adaptive QR algorithm \cite{Olver2013a}.
\end{enumerate}
When $\mu$ has a density $w$ and the support of $\mu$ is clear from context, we write $c_n(z;w) = c_n(z)$.

\begin{remark}
  It turns out that the recurrence for $c_n(z)$ in the case of Chebyshev-V and Chebyshev-W polynomials can be solve explicitly and this general procedure can be avoided, if necessary.
\end{remark}

\subsection{Discretizing \eqref{eq:SIE}}

Define the Chebyshev points of the first kind
\begin{equation}
  \mathbb C_n := \left\{x_j = \cos \left( \frac{j + 1/2}{n+1} \pi \right) :\quad 0 \leq j \leq n\right\}{\subset(-1,1)}, \quad n \geq 1.
\end{equation}
We also define projection operator of evaluation of a function at points
\begin{align}
\mathcal E_{S} f := (f(x))_{x \in S},
\end{align}
where to be truly precise, $S$ should be an ordered set.

Now suppose $f : [a,b] \to \mathbb C$ can be written as
\begin{align}
  f({y}) = \sum_{n=0}^m \gamma_n P_n({y};[a,b]),
\end{align}
{where the choice of $m$ for our purposes is discussed in Section~\ref{sec:adaptive}.}
The discretized versions of $\mathcal C_w^\pm$, $w = w_{a,b}$ are given by
\begin{equation}
\begin{aligned}
  \mathcal E_{T_{a,b}(\mathbb C_n)} \mathcal C_w^\pm f &= \begin{bmatrix} c_0^\pm( T_{a,b}(x_0);w ) & c_1^\pm( T_{a,b}(x_0);w ) & c_2^\pm( T_{a,b}(x_0);w ) & \hdots & c_m^\pm( T_{a,b}(x_0);w ) \\
    c_0^\pm( T_{a,b}(x_1);w ) & c_1^\pm( T_{a,b}(x_1);w ) & c_2^\pm( T_{a,b}(x_1);w ) & \hdots & c_m^\pm( T_{a,b}(x_1);w ) \\
    \vdots & \vdots & \vdots && \vdots \\
    c_0^\pm( T_{a,b}(x_n);w ) & c_1^\pm( T_{a,b}(x_n);w ) & c_2^\pm( T_{a,b}(x_n);w ) & \hdots & c_m^\pm( T_{a,b}(x_n);w )
    \end{bmatrix} \begin{bmatrix} \gamma_0 \\ \gamma_1 \\ \gamma_2 \\ \vdots \\ \gamma_m \end{bmatrix}\\
    &{=:C_{[a,b]}^\pm(m,n)\begin{bmatrix} \gamma_0 \\ \gamma_1 \\ \gamma_2 \\ \vdots \\ \gamma_m \end{bmatrix}}
    ,
\end{aligned}
\end{equation}
and for $c < d$, $[c,d] \cap [a,b] = \varnothing$, we have
\begin{equation}
\begin{aligned}
  \mathcal E_{T_{c,d}(\mathbb C_n)} \mathcal C_w f &= \begin{bmatrix} c_0( T_{c,d}(x_0);w ) & c_1( T_{c,d}(x_0);w ) & c_2( T_{c,d}(x_0);w ) & \hdots & c_m( T_{c,d}(x_0);w ) \\
    c_0( T_{c,d}(x_1);w ) & c_1( T_{c,d}(x_1);w ) & c_2( T_{c,d}(x_1); w ) & \hdots & c_m( T_{c,d}(x_1);w ) \\
    \vdots & \vdots & \vdots && \vdots \\
    c_0( T_{c,d}(x_n);w ) & c_1( T_{c,d}(x_n);w ) & c_2( T_{c,d}(x_n);w ) & \hdots & c_m( T_{c,d}(x_n);w )
    \end{bmatrix} \begin{bmatrix} \gamma_0 \\ \gamma_1 \\ \gamma_2 \\ \vdots \\ \gamma_m \end{bmatrix}\\
    &{=:C_{[a,b] \to [c,d]}(m,n) \begin{bmatrix} \gamma_0 \\ \gamma_1 \\ \gamma_2 \\ \vdots \\ \gamma_m \end{bmatrix}}.
\end{aligned}
\end{equation}
Note that each row of these matrices can be constructed either by forward recurrence or via the back substitution step of the adaptive QR algorithm, depending on where the evaluation points are located in the complex plane.

We now demonstrate the discretization of \eqref{eq:SIE} in the case $g = 1$.

\begin{example}\label{ex:genus_one}
  We seek vector-valued functions $\mathbf{U}_{\pm 1} : I_{\pm 1} \to \mathbb C^{1 \times 2}$.  So, write
  \begin{align}
    \mathbf{U}_{\pm 1} = \begin{bmatrix} u_{\pm 1, 1} & u_{\pm 1, 2} \end{bmatrix}.
  \end{align}
  We write out the full system of equations for scalar-valued functions explicitly:
  For $z \in I_1$
  \begin{align}
    \mathcal C_{w_1}^+ u_{1,1}(z) + \mathcal C_{w_{-1}} u_{-1,1}(z) - \ee^{-\ii \Omega_1(x,t)} \left[ \mathcal C_{w_1}^- u_{1,2}(z) + \mathcal C_{w_{-1}} u_{-1,2}(z) \right] &= \ee^{-\ii \Omega_1(x,t)} -1,\\
    \mathcal C_{w_1}^+ u_{1,2}(z) + \mathcal C_{w_{-1}} u_{-1,2}(z) - \ee^{\ii \Omega_1(x,t)} \left[ \mathcal C_{w_1}^- u_{1,1}(z) + \mathcal C_{w_{-1}} u_{-1,1}(z) \right] &= \ee^{\ii \Omega_1(x,t)} -1,
  \end{align}
  and for $z \in I_{-1}$
  \begin{align}
    \mathcal C_{w_{-1}}^+ u_{-1,1}(z) + \mathcal C_{w_{1}} u_{1,1}(z) - \ee^{\ii \Omega_1(x,t)} \left[ \mathcal C_{w_{-1}}^- u_{-1,2}(z) + \mathcal C_{w_{1}} u_{1,2}(z) \right] &= \ee^{\ii \Omega_1(x,t)} -1,\\
    \mathcal C_{w_{-1}}^+ u_{-1,2}(z) + \mathcal C_{w_{1}} u_{1,2}(z) - \ee^{-\ii \Omega_1(x,t)} \left[ \mathcal C_{w_{-1}}^- u_{-1,1}(z) + \mathcal C_{w_{1}} u_{1,1}(z) \right] &= \ee^{-\ii \Omega_1(x,t)} -1.
  \end{align}
  In block-operator form:
  \begin{multline}
    \begin{bmatrix} \mathcal C_{w_1}^+ &  -e^{-\ii \Omega_1(x,t)} \mathcal C_{w_1}^- &  \mathcal C_{w_{-1}}|_{I_1} & -e^{-\ii \Omega_1(x,t)}  \mathcal C_{w_{-1}}|_{I_1} \\
      -e^{\ii \Omega_1(x,t)} \mathcal C_{w_1}^- & \mathcal C_{w_1}^+ &  -e^{\ii \Omega_1(x,t)} \mathcal C_{w_{-1}}|I_1 & \mathcal C_{w_{-1}}|I_1 \\
      \mathcal C_{w_{1}}|I_{-1} & -e^{\ii \Omega_1(x,t)}\mathcal C_{w_{1}}|I_{-1} & \mathcal C_{w_{-1}}^+ &  -e^{\ii \Omega_1(x,t)} \mathcal C_{w_{-1}}^- \\
      -e^{-\ii \Omega_1(x,t)}\mathcal C_{w_{1}}|I_{-1} & \mathcal C_{w_{1}}|I_{-1} & -e^{-\ii \Omega_1(x,t)}\mathcal C_{w_{-1}}^- & \mathcal C_{w_{-1}}^+
    \end{bmatrix} \begin{bmatrix} u_{1,1} \\ u_{1,2} \\ u_{-1,1} \\ u_{-1,2} \end{bmatrix}\\ = \begin{bmatrix} \ee^{-\ii \Omega_1(x,t)} -1 \\ \ee^{\ii \Omega_1(x,t)} -1 \\ \ee^{\ii \Omega_1(x,t)} -1 \\ \ee^{-\ii \Omega_1(x,t)} -1 \end{bmatrix}.
  \end{multline}
  The discretized version is then
  \begin{equation}
    {\mathbf{A}} \boldsymbol \gamma = \boldsymbol \omega,
    \label{discrete-lin-sys}
  \end{equation}
  where 
  \begin{equation}
    \begin{aligned}
    {\mathbf{A}} &= {\mathbf{A}}(x,t,n,m) \\
      &= \begin{bmatrix} C^+_{I_1}(m,n) &  -\ee^{-\ii \Omega_1(x,t)} C^-_{I_1}(m,n) & C_{I_1 \to I_{-1}}(m,n) & -\ee^{-\ii \Omega_1(x,t)} C_{I_1 \to I_{-1}}(m,n) \\
      -\ee^{\ii \Omega_1(x,t)} C^+_{I_1}(m,n) & C^-_{I_1}(m,n) &  -\ee^{\ii \Omega_1(x,t)}  C_{I_1 \to I_{-1}}(m,n) &  C_{I_1 \to I_{-1}}(m,n) \\
       C_{I_{-1} \to I_{1}}(m,n) & -\ee^{\ii \Omega_1(x,t)}C_{I_{-1} \to I_{1}}(m,n) & C^+_{I_{-1}}(m,n) &  -\ee^{\ii \Omega_1(x,t)} C^-_{I_{-1}}(m,n) \\
      -\ee^{-\ii \Omega_1(x,t)}C_{I_{-1} \to I_{1}}(m,n) & C_{I_{-1} \to I_{1}}(m,n) & -\ee^{-\ii \Omega_1(x,t)}C^-_{I_{-1}}(m,n) & C^+_{I_{-1}}(m,n)
    \end{bmatrix}, \label{eq:Amatrix}
    \end{aligned}
    \end{equation}
    and
    \begin{equation}
    \boldsymbol \gamma = \boldsymbol \gamma(x,t) = \begin{bmatrix} \boldsymbol \gamma_{1,1} \\ \boldsymbol \gamma_{1,2} \\ \boldsymbol \gamma_{-1,1} \\ \boldsymbol \gamma_{-1,2} \end{bmatrix},\qquad 
    \boldsymbol \omega = \boldsymbol \omega(x,t) = \begin{bmatrix} \ee^{-\ii \Omega_1(x,t)} -1 \\ \ee^{\ii \Omega_1(x,t)} -1 \\ \ee^{\ii \Omega_1(x,t)} -1 \\ \ee^{-\ii \Omega_1(x,t)} -1 \end{bmatrix},
  \end{equation}
  where each entry in the right-hand side vector in \eqref{discrete-lin-sys} is a constant vector {for given $(x,t)$}. 
\end{example}

Lastly, we need to consider the computation of $\partial_x s_1(x,t)$ from \eqref{eq:S} {in order to use \eqref{recovery-q}}.  First, we note that if $u_{j,1}$ in \eqref{eq:SIE} is given by
\begin{align}\label{eq:expand}
 u_{j,1}({y}) = \sum_{n = 0}^\infty \gamma_{n,j}(x,t) P_n({y};I_j),
\end{align}
then by the orthogonality of the polynomials
\begin{align}\label{eq:recovery}
  s_1(x,t) = - \frac{1}{2 \pi \ii} \sum_{ \substack{j = -g\\ j \neq 0}}^g \frac{a_{j+1} - b_j}{2} \gamma_{0,j}(x,t),
\end{align}
implying that we need to solve for the $x$-derivative of the coefficients in the expansion \eqref{eq:expand}.  To do this, the linear system ${\mathbf{A}} \boldsymbol \gamma = \boldsymbol \omega$ can be differentiated to find
\begin{align}
  {\mathbf{A}} (\partial_x \boldsymbol \gamma) = \partial_x \boldsymbol \omega - (\partial_x {\mathbf{A}})  \boldsymbol \gamma.
\end{align}

\subsection{Preconditioning}
\label{s:precond}
The discretization described in Example~\ref{ex:genus_one} is easily extended to find a discretization of the operators $\mathcal S$ and $\mathrm{diag}({\mathcal{S}})$ where we expect the discretization of $\mathrm{diag}({\mathcal{S}})$ to become a good preconditioner for $\mathcal S$ in light of Lemma~\ref{l:compact}.  In practice, we find it works well to use a discretization of
\begin{align}
  \diag( \mathcal S(I_1;\Omega_1), \mathcal S(I_2;\Omega_2), \ldots, \mathcal S(I_g;\Omega_g)) \label{eq:preconder}
\end{align}
as a block-diagonal preconditioner.  We find that with this preconditioner, for a fixed tolerance, the GMRES algorithm requires a bounded number of iterations, independent of $x$ and $t$.  We explore this more in Figure~\ref{fig:precond}.

\subsection{Adaptivity}\label{sec:adaptive}

In the discretization of $\mathcal S$ following the procedure outline{d} in Example~\ref{ex:genus_one}, {an} important question {is} that of choosing $n,m$.  And, in general, different choices for $n$ and $m$ should be made for each block of $\mathcal S$ under the constraint that the resulting matrix is square.

It can be shown that the solution $\mathbf{S}(z;x,t)$ can, in an appropriate sense, be analytically continued off the interval $[-1,1]$ \cite{TrogdonDressing}. For example, one expects the solution $u_j$ on $I_j$ of \eqref{eq:SIE} to have an analytic continuation to any ellipse with foci at the endpoints of $I_j$ provided that ellipse does not intersect any other $I_\ell$ for $\ell \neq j$.  And, it is well known that rate of exponential convergence of a Chebyshev interpolant can be estimated based on this ellipse \cite{TrefethenATAP}:

\begin{theorem}
  Suppose $f: [-1,1] \to \mathbb C$  can be analytically continued to the open Bernstein ellipse
  \begin{align}
    B_\rho = \{ (z^{-1} + z)/2 : |z| < \rho\}, \quad \rho > 1,
  \end{align}
   Then for
  \begin{align}
    \gamma_k = \int_{-1}^1 f({y}) T_k({y}) \frac{\dd {y}}{\pi \sqrt{1 -{y}^2}},
  \end{align}
  one has
  \begin{align}
    |\gamma_k| \leq 2 \sup_{z \in B_\rho} |f(z)| \rho^{-k},
  \end{align}
  and consequently
  \begin{align}
    \max_{{y\in}[-1,1]} \left|f({y}) - \sum_{j=1}^m \gamma_j T_j({y})\right| \leq \frac{4 M \rho^{-m}}{\rho - 1}.
  \end{align}
\end{theorem}

So, if\footnote{If $j = 1$ we compare with $I_{1}$ with $I_{-1}$ and $I_2$.  Then $j < 0$ is taken care of by symmetry.} $j > 1$, $I_j = [a,b]$, $I_{j-1} = [a',b']$ and $I_{j+1} = [a'',b'']$, we map $[a,b]$ to $[-1,1]$ using $T^{-1}_{a,b}$ which, in turn, maps
\begin{align}
  [a',b'] \to \left[ \frac{2}{b-a} a' + \frac{b+a}{b-a}, \frac{2}{b-a} b' + \frac{b+a}{b-a}  \right],
\end{align}
and similarly for $[a'',b'']$.  So, set
\begin{align}
  \delta = \min\left\{ -1 - \frac{2}{b-a} b' + \frac{b+a}{b-a}, \frac{2}{b-a} a'' + \frac{b+a}{b-a} -1  \right\}{>0}.
\end{align}
So, we expect $u_j \circ T_{a,b}$ to have an analytic extension to $B_\rho$ for any $\rho$ such that $B_\rho \cap \mathbb R \subset [-1-\delta, 1+\delta]$.  To be conservative in our estimates, we use $\delta/2$ instead and find that $\rho$ should be chosen to be:
\begin{align}
  (\rho^{-1} + \rho)/2 = 1 + \delta/2 \Rightarrow \rho(I_j) = \frac{1}{2} \left( 2 + \delta + \sqrt{\delta(4+\delta)} \right) > 1.
\end{align}
So, given an estimate for $M$, and a tolerance $\epsilon>0$, we can choose $m$ so that
\begin{align}
\frac{4 M \rho(I_j)^{-m}}{\rho(I_j) - 1} < \epsilon,
\end{align}
and this provides an \emph{a priori} guide as to how to choose $m$ in the discretizations $C^\pm_{[a,b]}(m,n)$ and $C_{[a,b] \to [c,d]}(m,n)$.

\section{Applications}
\label{s:Applications}
\subsection{Solutions with dressing: Slowly shrinking gaps}
\label{s:numerics-dressing-slow}

With the methodology set out, one can easily specify a finite number of gaps and specify the Dirichlet eigenvalues within each gap, and compute the associated potential and its evolution under the KdV flow.  To demonstrate this we make first make the following choice for the gaps {in the $\lambda$-plane}.

Choosing $\alpha_1 = 0.1$ we then set for $j = 1,2,\ldots,g$
\begin{align}\label{eq:slow_gaps}
  \alpha_{j+1} - \beta_j = \begin{cases} j^{-1}  & j \text{ is odd},\\
   3 j^{-1} & j \text{ is even},
    \end{cases} \quad \beta_{j} = 2(j-1)^2 + \frac{4}{10}.
\end{align}
To fully specify the a solution we set $\gamma_j(x_0 = 0) = \beta_j$ for all $j$. 

While, as we demonstrate in the next section, we can compute the corresponding solution $q(x,t)$ for $g$ large, the solution oscillates wildly and is difficult to visualize.  For this reason we plot the solution for smaller values of $g$ over short time ranges. {See Figure~\ref{fig:dressing-waterfall}}.

\begin{figure}[tbp]
  \centering
  \begin{overpic}[width = 0.45\linewidth]{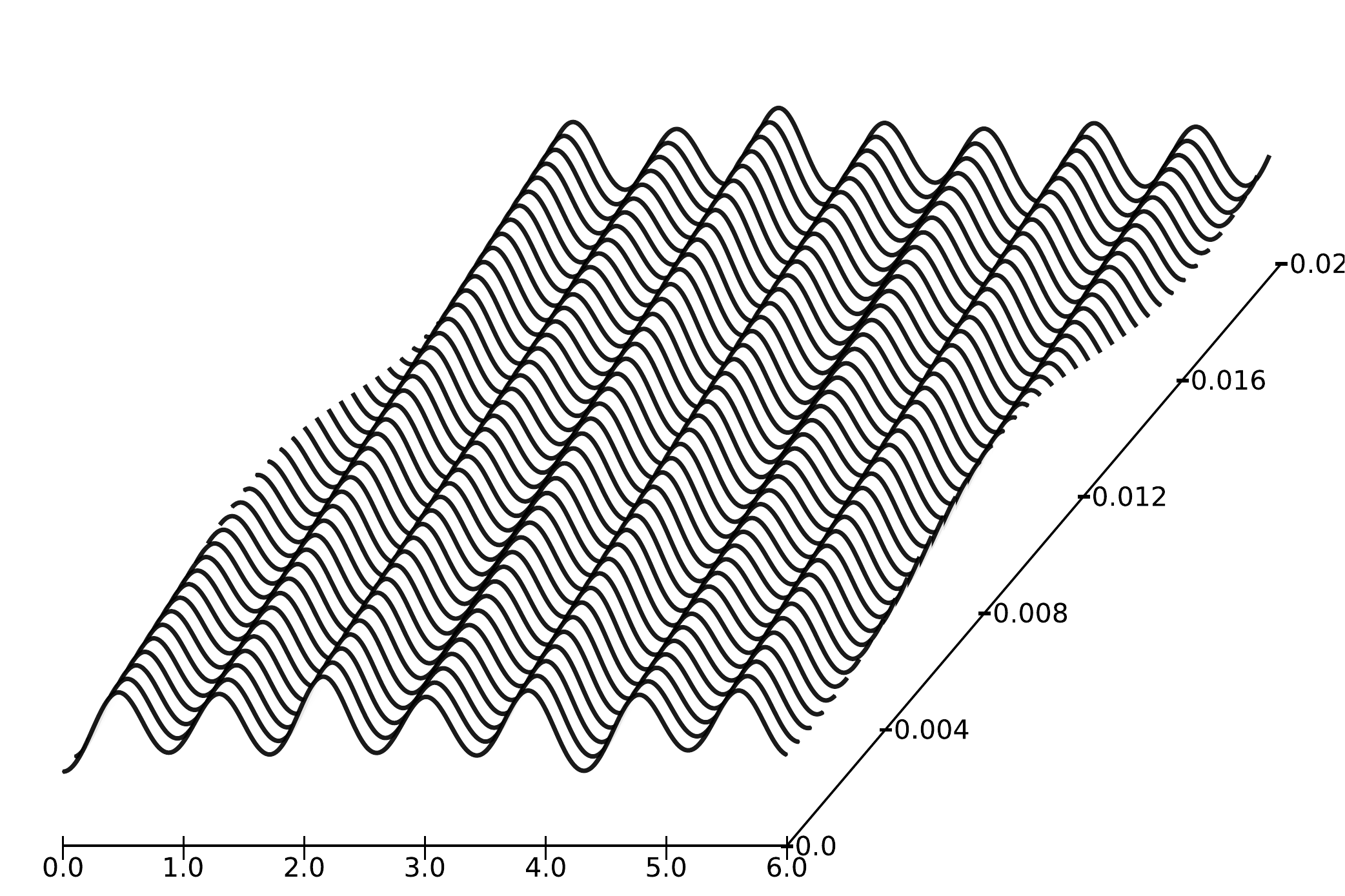}
    \put(28,-3){{$x$}}
    \put(84,22){{$t$}}
    \put(10,40){{$g = 5$}}
  \end{overpic}
  \begin{overpic}[width = 0.45\linewidth]{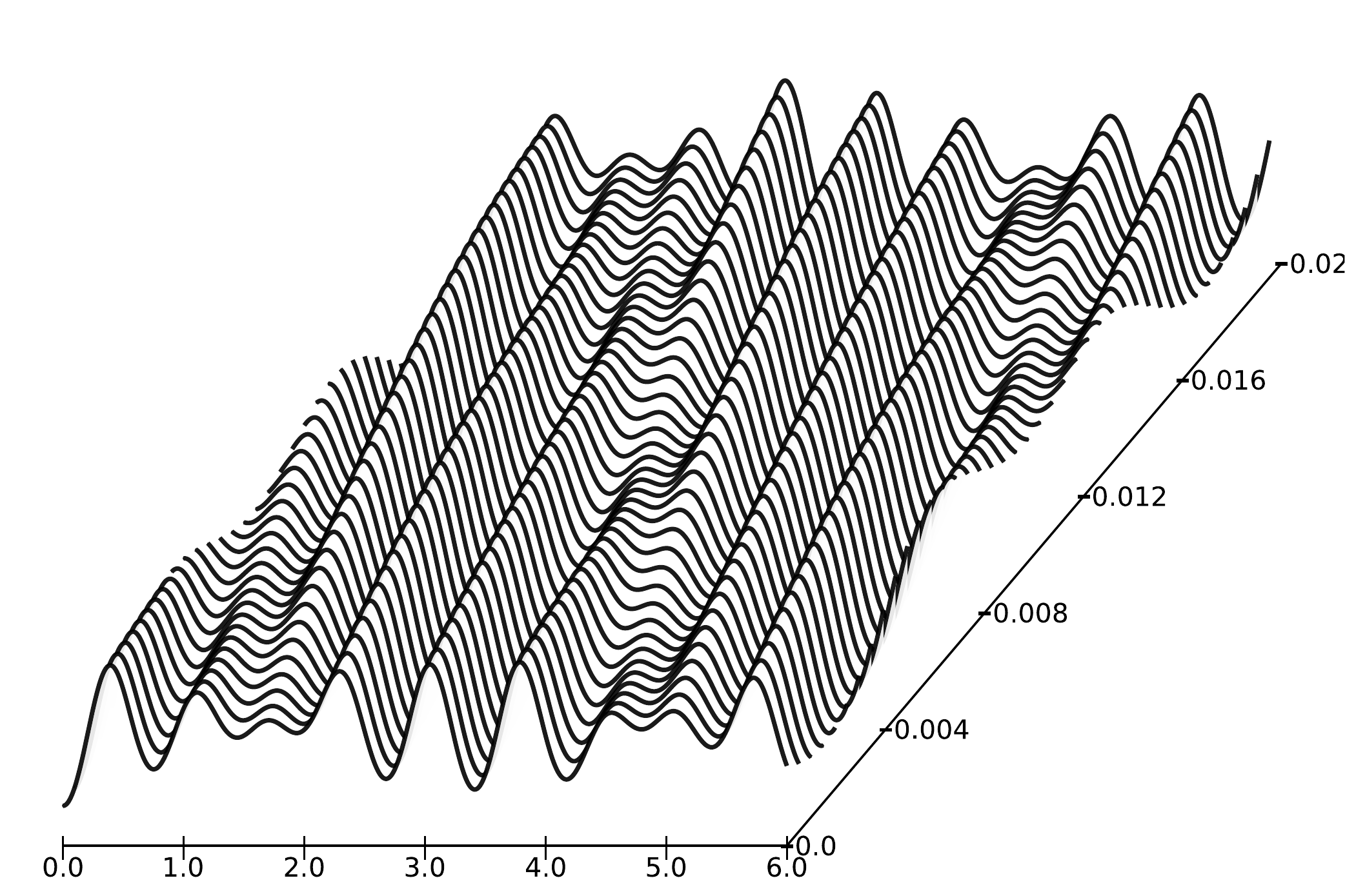}
   \put(28,-3){{$x$}}
   \put(84,22){{$t$}}
    \put(10,40){{$g = 7$}}
  \end{overpic} \\

   \begin{overpic}[width = 0.45\linewidth]{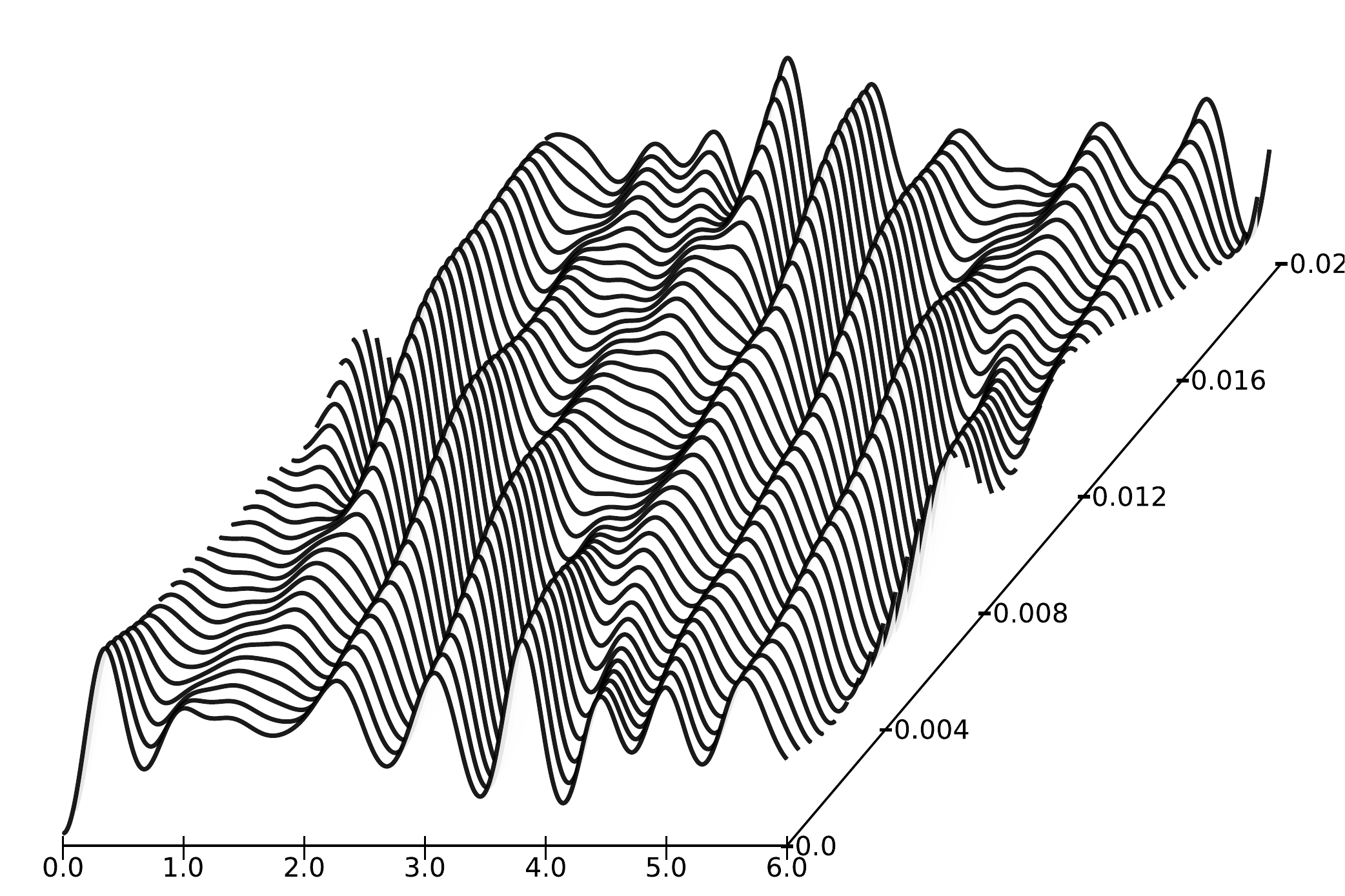}
   \put(28,-3){{$x$}}
   \put(84,22){{$t$}}
    \put(8,40){{$g = 10$}}
  \end{overpic}

  \caption{Numerical solutions of the KdV equation with slowly shrinking gaps, see \eqref{eq:slow_gaps}.  The number of collocation points per contour is chosen adaptively using the methodology in Section~\ref{sec:adaptive}. }
  \label{fig:dressing-waterfall}
\end{figure}
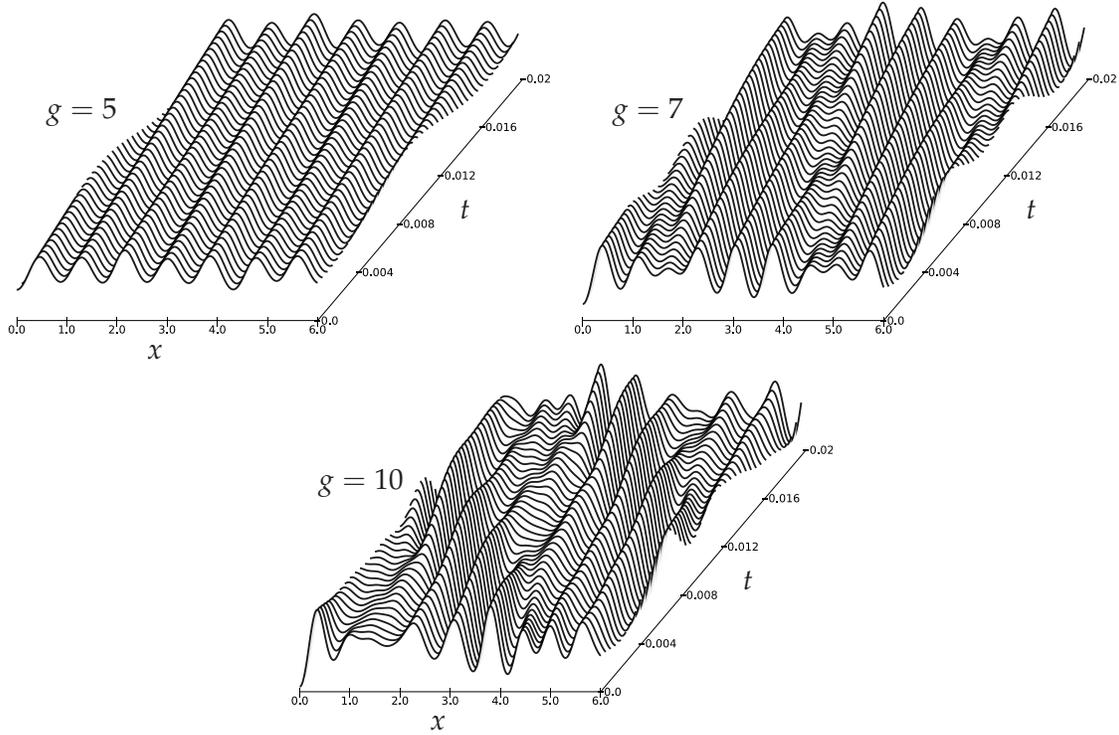

\subsection{High-genus solutions with dressing: Dense gaps and universality}
\label{s:numerics-dressing-univ}

It is also interesting to ask what happens if an increasing number of gaps are put into a fixed interval.  Fix, for convenience $\alpha_1 = 0$, and $\beta _1 > 0$.  Also, suppose that $\alpha_{g+1} \to \alpha > \beta _1$ as $g$ increases. Now suppose $\sigma({y}) = \int_{\beta_1}^t \varrho({y}) \dd {y}$, where $\varrho({y})$ is positive {and continuous} on $[\beta_1,\alpha]$, increases from $0$ to $1$ over the interval $[\beta_1,\alpha]$.  Given a sequence $w_1,\ldots,w_g$ with $0 < w_j < 1$, define $\alpha_2,\ldots, \alpha_{g}$ and $\beta_2,\ldots, \beta_{g}$ through
  \begin{align*}
    \sigma(\alpha_j) = \frac{j-1}{g-1}, \quad \sigma(\beta_j) = \frac{j - 1+ w_j}{g-1}.
  \end{align*}

  The following lemma will be of use.
  \begin{lemma}
    The rational function
    \begin{align*}
     B(\lambda) := \prod_{j=2}^{g} \frac{\lambda - \alpha_j}{\lambda - \beta_j}
    \end{align*}
    satisfies
        \begin{equation}
      \log (B(\lambda)) - \frac{1}{g-1} \sum_{j=2}^{g} \frac{w_j}{\varrho(\beta_j)} \frac{1}{\lambda - \beta_j}  = o(1), \quad g \to \infty,
    \end{equation}
    uniformly on compact subsets of $\mathbb C \setminus [\beta_1,\alpha]$.
  \end{lemma}
  \begin{proof}
    Expanding $\log(z)$ at $z = 1$ gives
    \begin{equation}
      \log \left(\frac{\lambda - \alpha_j}{\lambda - \beta_j}\right) = \frac{\beta_j - \alpha_j}{\lambda - \beta_j} + O((\beta_j-\alpha_j)^2).
    \end{equation}
    By the mean value theorem, $\beta_j - \alpha_j = \frac{w_j}{g-1} \varrho(\xi_j)^{-1}$, where $\alpha_j \leq \xi_j \leq \beta_j$.  This gives
    \begin{equation}
      \left|\log (B(\lambda)) - \frac{1}{g-1} \sum_{j=2}^g \frac{w_j}{\varrho(\xi_j)} \frac{1}{\lambda - \beta_j} \right| \leq C  \sum_{j=2}^g \frac{w_j^2}{(g-1)^2} \varrho(\xi_j)^{-2}.
    \end{equation}
    Then because $\varrho$ is uniformly continuous, for any ${\epsilon>0}$ there exists a $g_0> 0$ such that $|\varrho(\xi_j) - \varrho(\beta_j)| < \epsilon$ for all $j$ if $g > g_0$, so that
    \begin{align*}
      \left| \frac{1}{\varrho(\xi_j)} - \frac{1}{\varrho(\beta_j)} \right| \leq \frac{\epsilon}{\varrho(\xi_j)\varrho(\beta_j)},
    \end{align*}
    and the claim follows.
  \end{proof}

  Now, if in the above lemma, $w_j = v(\beta_j)$ for some continuous function $v: [\beta_1,\alpha] \to (0,1)$ we find that  
  \begin{align*}
    \log(B(\lambda)) = \int_{\beta_1}^{\alpha} \frac{v(y)}{\lambda -y} \dd y + o(1), \quad g \to \infty.
  \end{align*}
  Thus, the distribution of individual locations $\alpha_j$, $\beta_j$ do not influence the limiting behavior of {$B(\lambda)$ as $g$ becomes large}. But rather, the distribution of the lengths of the bands is the most important quantity.

  To see how $B(\lambda)$ will arise in a Riemann--Hilbert problem consider the above choice for $\alpha_j$ and $\beta_j$, for given functions ${\sigma}$ and $v$.
  Previously, we have moved poles in the gap $[\beta_j,\alpha_{j+1}]$ on one sheet of the Riemann surface to the point $\lambda=\alpha_j$.  This was for numerical convenience.  Here, for analytical convenience, we put the poles at $\lambda=\beta_j$. {We diagonalize the twist jump matrix for ${\boldsymbol\Psi}(\lambda;x,t)$:}
  {
  \begin{equation}
      \sigma_1 = \mathbf{Q} \sigma_3 \mathbf{Q}^{-1},\qquad \mathbf{Q}:=\frac{1}{\sqrt{2}}\begin{bmatrix} 1 & - 1 \\ 1 & 1 \end{bmatrix},\quad \mathbf{Q}^{-1} = \mathbf{Q}^\top.
  \end{equation}
}
  Then
  \begin{align*}
    \mathbf{W}(\lambda) := \mathbf{Q} \begin{bmatrix} 1 & 0 \\ 0 &  {B(\lambda)^{1/2}}\sqrt{\frac{\lambda - \alpha_1}{\lambda- \beta_1}}\sqrt{\alpha_{g+1} - \lambda}  \end{bmatrix} \mathbf{Q}^{\top}, 
  \end{align*}
  can be used\footnote{The {power} function {$B(\lambda)^{1/2}$} is chosen to have its branch cut on $\cup_{j=1}^g [\alpha_j,\beta_j]$ with $B(\lambda)^{1/2}\to 1$ as $\lambda\to\infty$.} to remove the jumps of $\mathbf{\Psi}$. Define
  \begin{align*}
    \check{\mathbf{\Psi}}(\lambda;x,t) = \begin{cases} \mathbf{\Psi}(\lambda;x,t) \mathbf{W}(\lambda)^{-1} & \lambda \in \Omega,\\
                                       \mathbf{\Psi}(\lambda;x,t)  & \text{otherwise}, \end{cases}
  \end{align*}
  where $\Omega$ is the {disk}\footnote{This domain is taken for concreteness, any other reasonable region containing all finite bands and gaps will suffice.} centered at $\lambda={\alpha_1} = 0$ with radius $\alpha + c$, $c > 0${, and we orient the circle $\partial\Omega$ counter-clockwise}. Then one finds that $ \check{\mathbf{\Psi}}$ satisfies the following jump condition
  \begin{align*}
     \check{\mathbf{\Psi}}^+(\lambda;x,t) = \check{\mathbf{\Psi}}^-(\lambda;x,t)\mathbf{W}(\lambda)^{-1}, \quad \lambda \in \partial \Omega.
  \end{align*}
  For $\lambda \in \partial \Omega$, supposing that $\beta_1 \to \beta$, $0 \leq \beta < \alpha$, one has
  \begin{align*}
    \mathbf{W}(\lambda)^{-1} \overset{g \to \infty}{\longrightarrow} \mathbf{W}_\infty(\lambda)^{-1} : = \mathbf{Q} \begin{bmatrix} 1 & 0 \\ 0 & \frac{1}{\sqrt{\alpha - \lambda}} \sqrt{\frac{\lambda - \beta}{\lambda}}\exp \left( -\frac 1 2 \int_{\beta}^{\alpha} \frac{v(s)}{s - \lambda} \dd s\right)  \end{bmatrix} \mathbf{Q}^{\top}.
  \end{align*}
  Bringing the jump from $\partial\Omega$ back to the real axis, we find the following Riemann--Hilbert problem for a limiting ${\boldsymbol\Psi}_\infty$
\begin{align*}
{\boldsymbol\Psi}_\infty^+(\lambda;x,t)&={\boldsymbol\Psi}_\infty^-(\lambda;x,t) \sigma_1
,\quad \lambda \in (0,\beta) \cup (\alpha,\infty), \\
{\boldsymbol\Psi}_\infty^+(\lambda;x,t)&={\boldsymbol\Psi}_\infty^-(\lambda;x,t) \begin{bmatrix} 
\frac{1 + \ee^{ \ii \pi  v(\lambda)} }{2} & \frac{ \ee^{\ii \pi v(\lambda)} -1 }{2}\\
\frac{ \ee^{\ii \pi v(\lambda)} -1 }{2}  & \frac{1 + \ee^{\ii \pi v(\lambda)} }{2} \end{bmatrix}, \quad \lambda \in (\beta,\alpha),
\end{align*}
and has the asymptotic behavior
\begin{equation}\label{eq:psi_inf_inf}
{\boldsymbol\Psi}_\infty(\lambda;x,t) = \begin{bmatrix} \ee^{\ii \lambda^{\frac{1}{2}}(x + 4 \lambda t)}& \ee^{-\ii \lambda^{\frac{1}{2}}(x + 4 \lambda t)}\end{bmatrix}\left(\mathbb{I} + o(1) \right),\quad \lambda\to\infty.
\end{equation}
This construction can be immediately generalized to
\begin{align*}
{\boldsymbol\Psi}_\infty^+(\lambda;x,t)&={\boldsymbol\Psi}_\infty^-(\lambda;x,t) \sigma_1
,\quad \lambda \in (\alpha_{g+1},\infty) \cup \bigcup_{j=1}^g (\alpha_j,\beta_j), \\
{\boldsymbol\Psi}_\infty^+(\lambda;x,t)&={\boldsymbol\Psi}_\infty^-(\lambda;x,t) \begin{bmatrix} 
\frac{1 + \ee^{ \ii \pi  v_j(\lambda)} }{2} & \frac{ \ee^{\ii \pi v_j(\lambda)} -1 }{2}\\
\frac{ \ee^{\ii \pi v_j(\lambda)} -1 }{2}  & \frac{1 + \ee^{\ii \pi v_j(\lambda)} }{2} \end{bmatrix}, \quad \lambda \in (\beta_j,\alpha_{j+1}),
\end{align*}
where $v_j: [\beta_j,\alpha_{j+1}] \to (0,1)$ is continuous and ${\boldsymbol\Psi}_\infty$
and has the asymptotic behavior \eqref{eq:psi_inf_inf}. While full exploration of such Riemann--Hilbert problems is beyond the scope of the current paper, potentials for $v(\lambda) \equiv 1/2$ are given in Figure~\ref{fig:uniform} in the case where $\beta_1 \to 0$ as $g \to \infty$ and the evolution is plotted in Figure~\ref{fig:uniform-contour}.

  \begin{figure}[tbp]
  \centering
  \begin{overpic}[width = 0.45\linewidth]{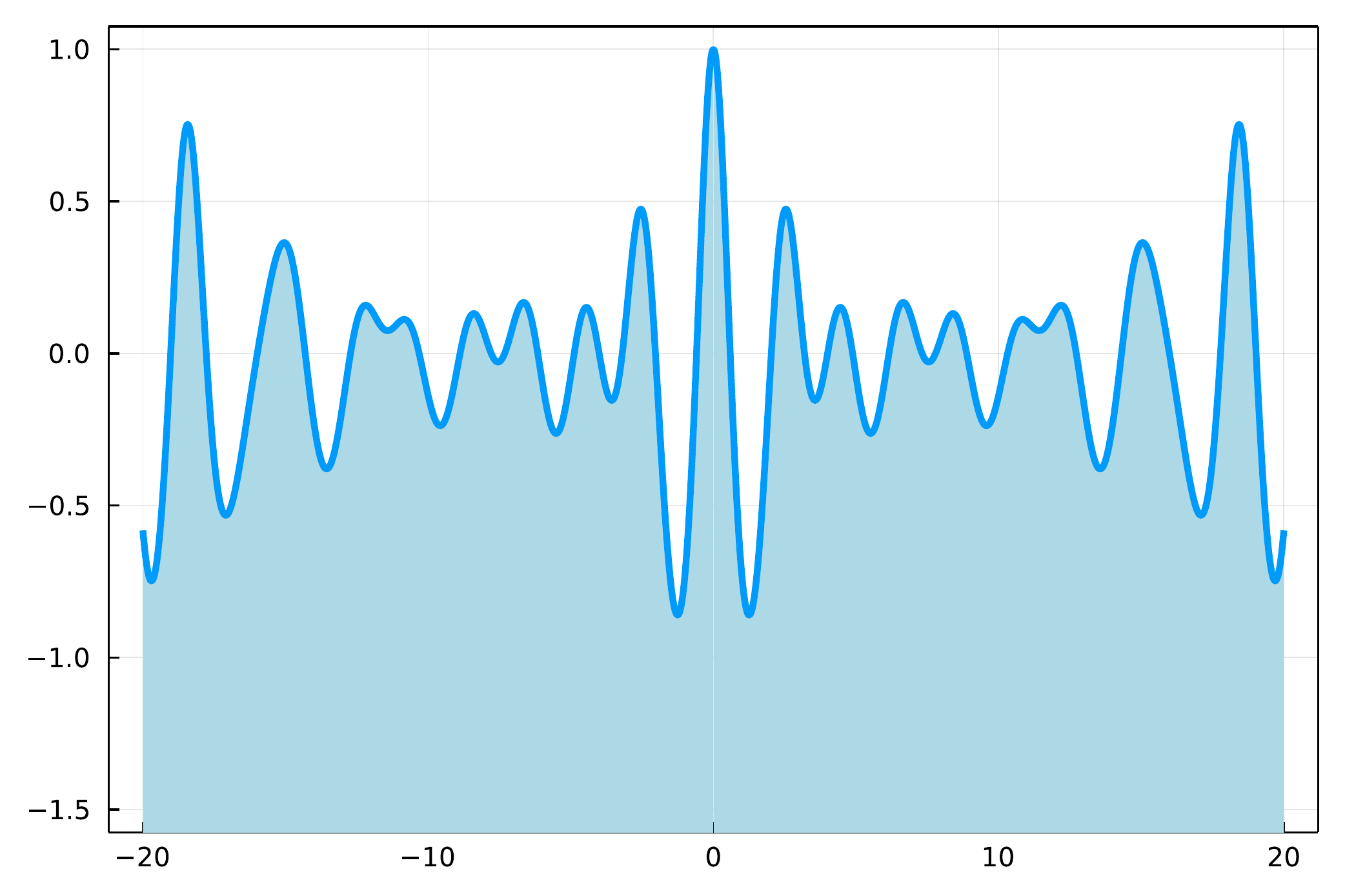}
    \put(80,58){$g = 5$}
    \put(-3,25){\rotatebox{90}{$q(x,0)$}}
    \put(55,-1){{$x$}}
  \end{overpic}\hspace{0.1in}
  \begin{overpic}[width = 0.45\linewidth]{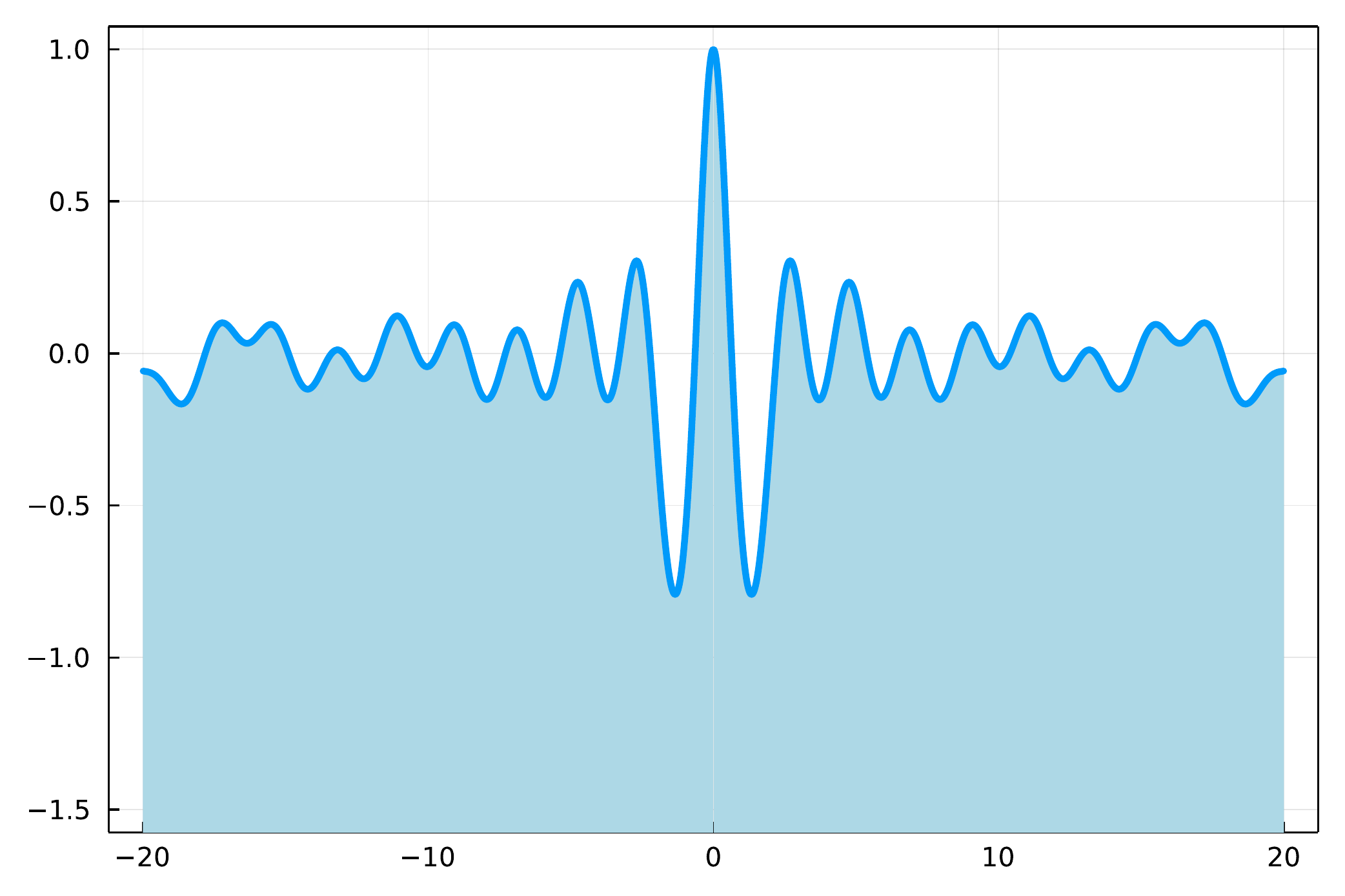}
    \put(80,58){$g = 10$}
    \put(-3,25){\rotatebox{90}{$q(x,0)$}}
    \put(55,-1){{$x$}}  
  \end{overpic}\\  
  \vspace{.1in}
  
  \begin{overpic}[width = 0.45\linewidth]{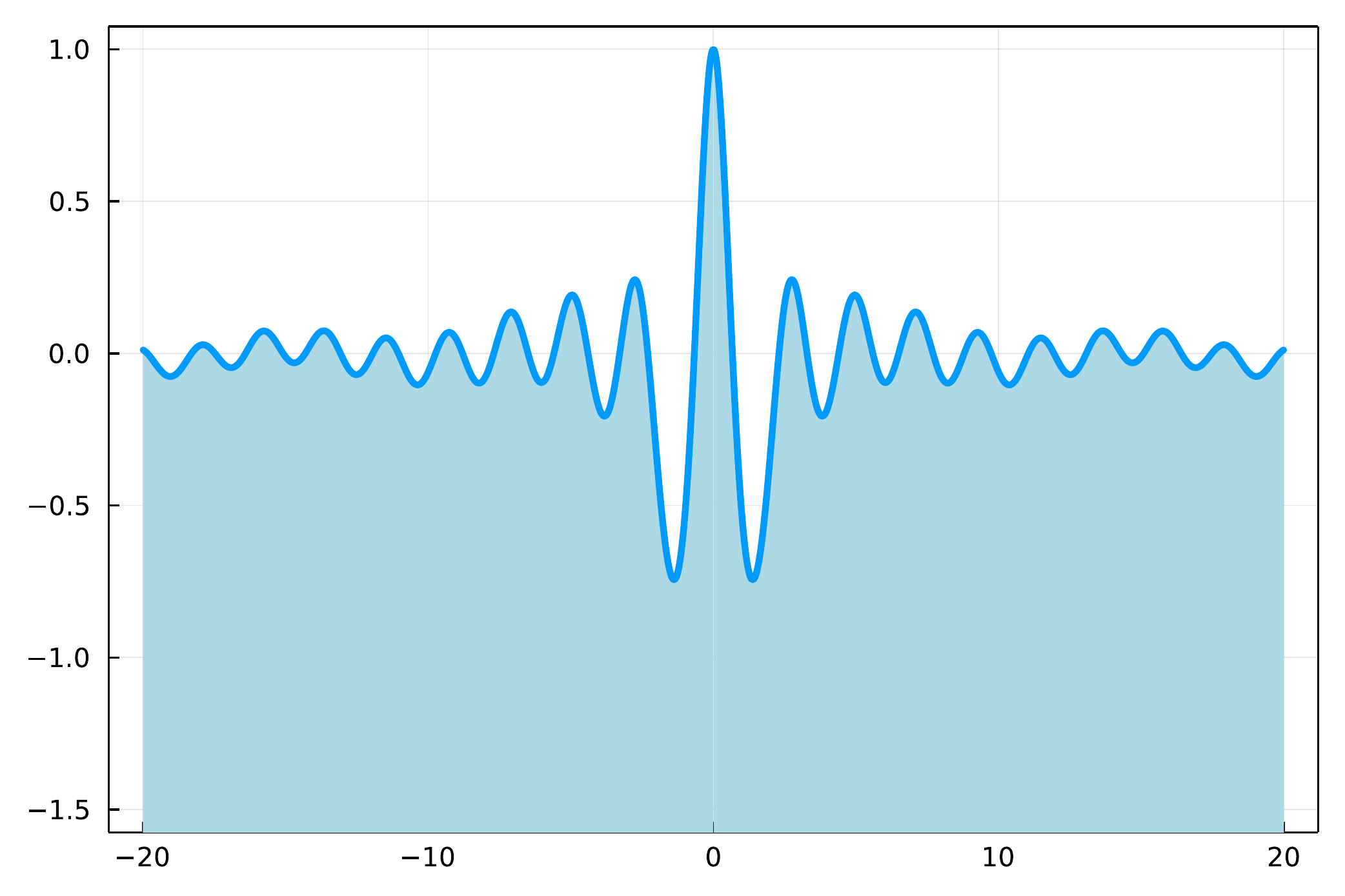}
    \put(80,58){$g = 20$}
    \put(-3,25){\rotatebox{90}{$q(x,0)$}}
    \put(55,-1){{$x$}}
  \end{overpic}\hspace{0.1in}
  \begin{overpic}[width = 0.45\linewidth]{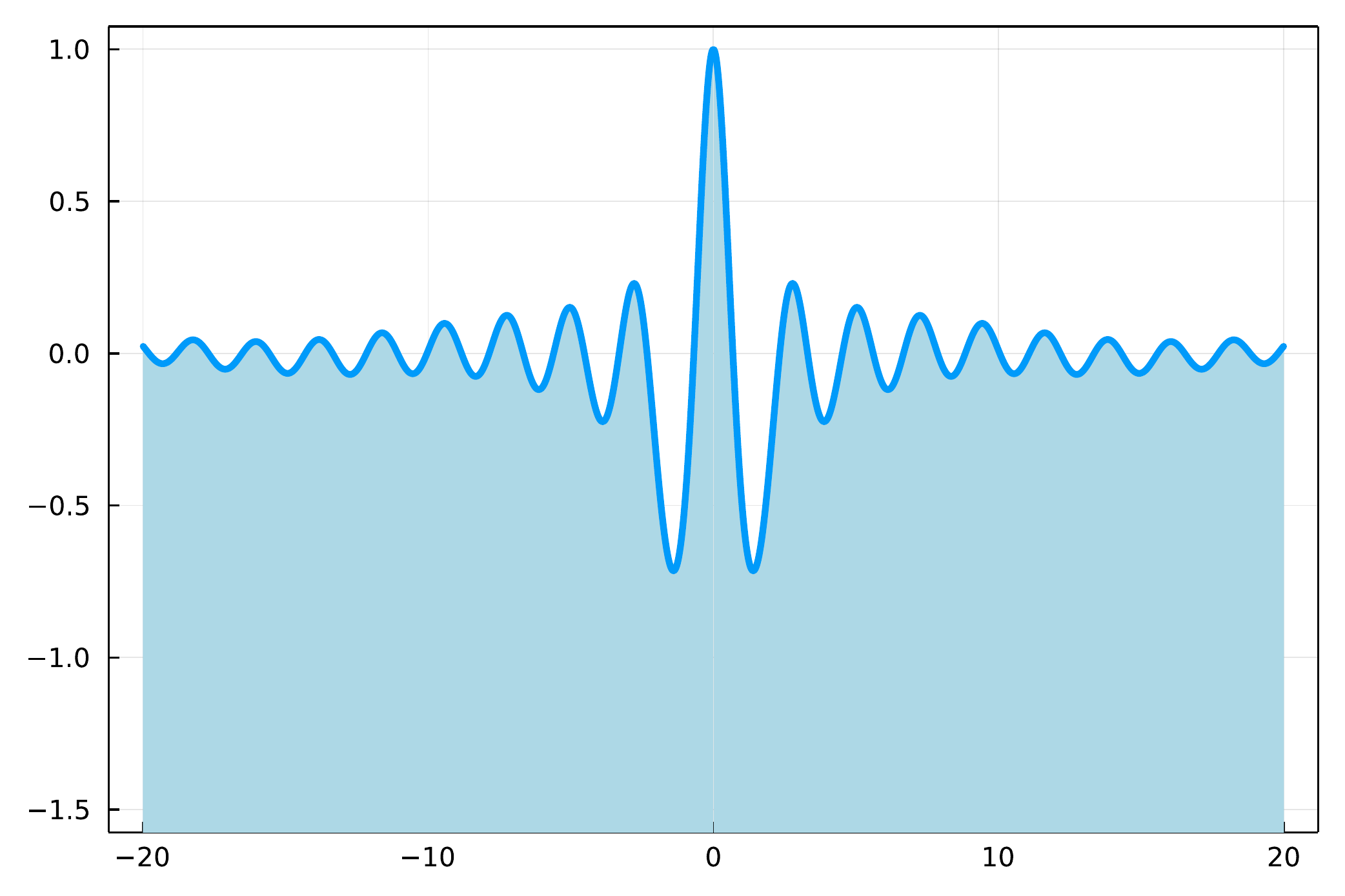}
    \put(80,58){$g = 40$}
    \put(-3,25){\rotatebox{90}{$q(x,0)$}}
    \put(55,-1){{$x$}}
  \end{overpic}
  \caption{\label{fig:uniform} Numerical solutions of the KdV equation as $g$ increases with $v$ in Section~\ref{s:numerics-dressing-univ} being constant, equal to $1/2$ and $\alpha_{g+1} = 2 + 1/g$, $\beta = 3/g$.  Here we also, for simplicity, take $\varrho$ to be uniform on $[\beta_1,\alpha]$.  The number of collocation points on each interval is obtained adaptively using the methodology in Section~\ref{sec:adaptive}.  It is clear from these panels that the solution converges as $g \to \infty$. The limit may be related to the so-called primitive potentials \cite{Dyachenko2016} but the connection is not immediate..  }  
\end{figure}

\begin{figure}[tbp]
  \centering
  \begin{overpic}[width = 0.6\linewidth]{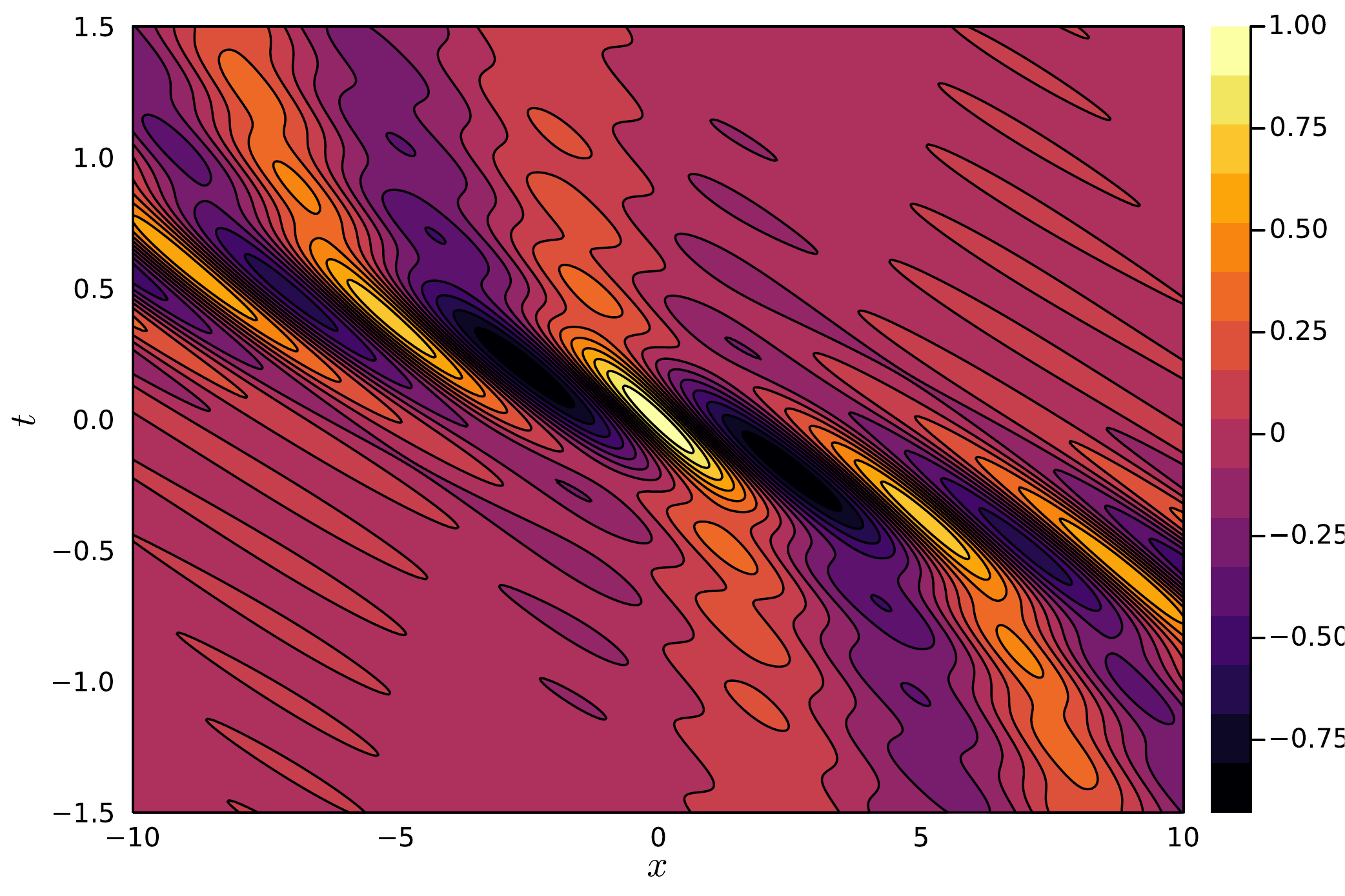}
  \end{overpic}
  \caption{\label{fig:uniform-contour} Evolution of the solution of the KdV equation with $g = 30$ and $v$ in Section~\ref{s:numerics-dressing-univ} being constant, equal to $1/2$ and $\alpha_{g+1} = 2$.}
  \end{figure}

\subsection{Initial-value problem with smooth data}
\label{s:numerics-smooth}

We consider the classical problem of Zabusky--Kruskal \cite{Zabusky1965}
\begin{align}\label{eq:uZK}
  u_t + u u_x + \delta^2 u_{xxx} = 0,\\
  u(x,0) = \cos \pi x,
\end{align}
for $x \in (-2,2]$.  Based on Remark~\ref{r:symmetry}, since we are set to solve $q_x + 6 qq_x + q_{xxx} = 0$, we choose
\begin{align*}
  q(x,0) = a^{-1} u( x/b,0), \quad b = \frac{1}{\sqrt{6} \delta}, \quad a = \frac{6c}{b}, \quad c = b^3 \delta^2,
\end{align*}
and then $u(x,t) = a q(b x,ct)$.

We choose $\delta = 0.08$ and use an error tolerance of $10^{-13}$ (see Section~\ref{sec:adaptive}) to choose the number of collocation points on each interval $I_j$.  We then plot the error in computing $u(x,1)$ as $g$ increases.  To estimate the true error we use the exponential integrator method discussed in \cite{Bilman2018} motivated from the work in \cite{Klein2008} to compute the ``true'' solution.  Exponential convergence is seen in Figure~\ref{fig:ZK-g}.  The evolution of the corresponding solution is given in Figure~\ref{fig:ZK-evolve}.

\begin{remark}
To be able to compute this solution, one needs to be able to compute the spectrum.  We use the Fourier--Floquet--Hill method \cite{Deconinck2006b} to compute the periodic/anti-periodic eigenvalues and use a Chebyshev method to compute the Dirichlet eigenvalues.  This latter method can be found implemented in both {\tt Chebfun} \cite{chebfun} and {\tt ApproxFun} \cite{approxfun}.  
\end{remark}

\begin{figure}[tbp]
  \centering
  \begin{overpic}[width = 0.6\linewidth]{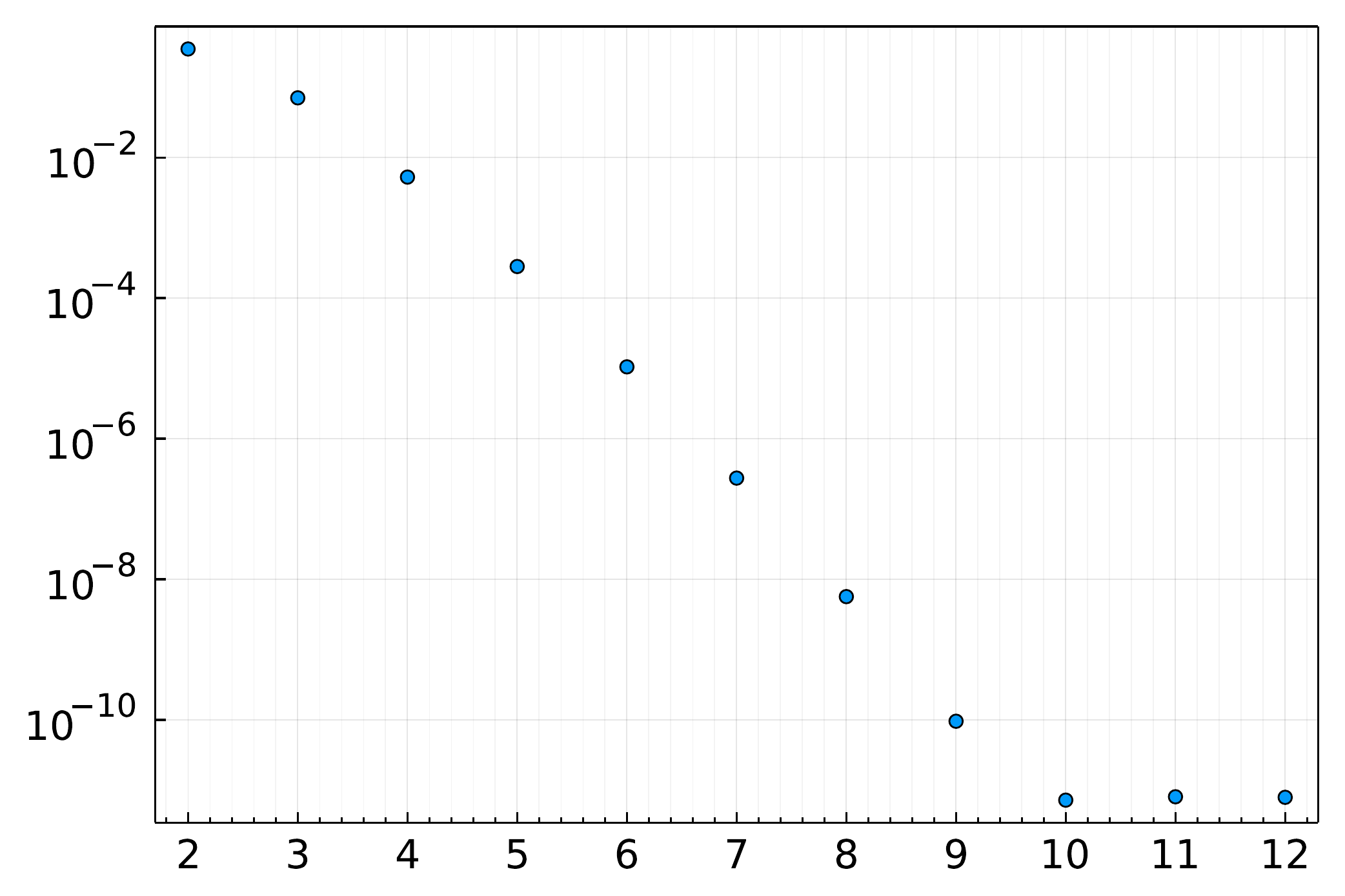}
    \put(-1,30){\rotatebox{90}{Error}}  
    \put(55,-1){{$g$}}
  \end{overpic}
  \caption{\label{fig:ZK-g} The error in the computation of $u(x,1)$ where $u$ is the solution of \eqref{eq:uZK} as measured by comparing the numerical solution with that obtained by direct integration of the KdV equation using an exponential time integration method, see~\cite{Klein2008} for a general reference and \cite{Bilman2018} for the precise method used. The error is on the order of $10^{-14}$ for sufficiently large genus and we are unable to distinguish which numerical solution is giving the dominant contribution to the error.  The number of collocation points per contour is chosen using the methodology in Section~\ref{sec:adaptive}. }
  \end{figure}

  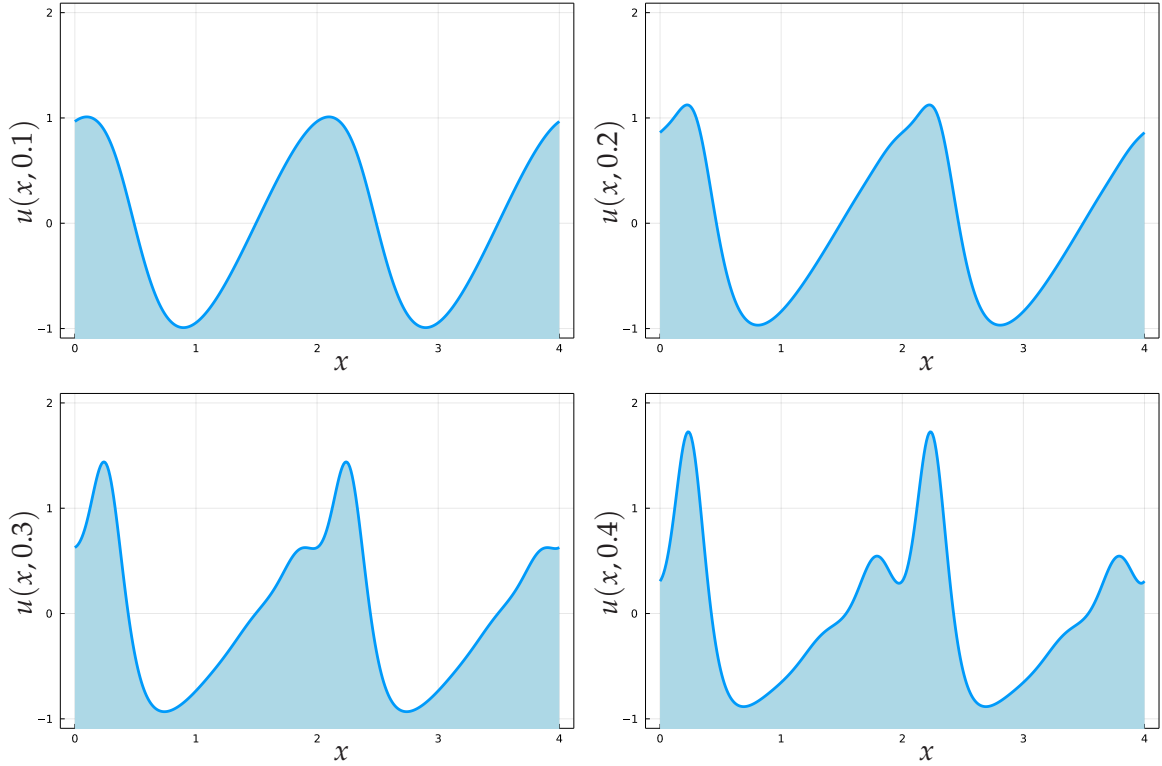
\begin{figure}[tbp]
  \centering
  \begin{overpic}[width = 0.45\linewidth]{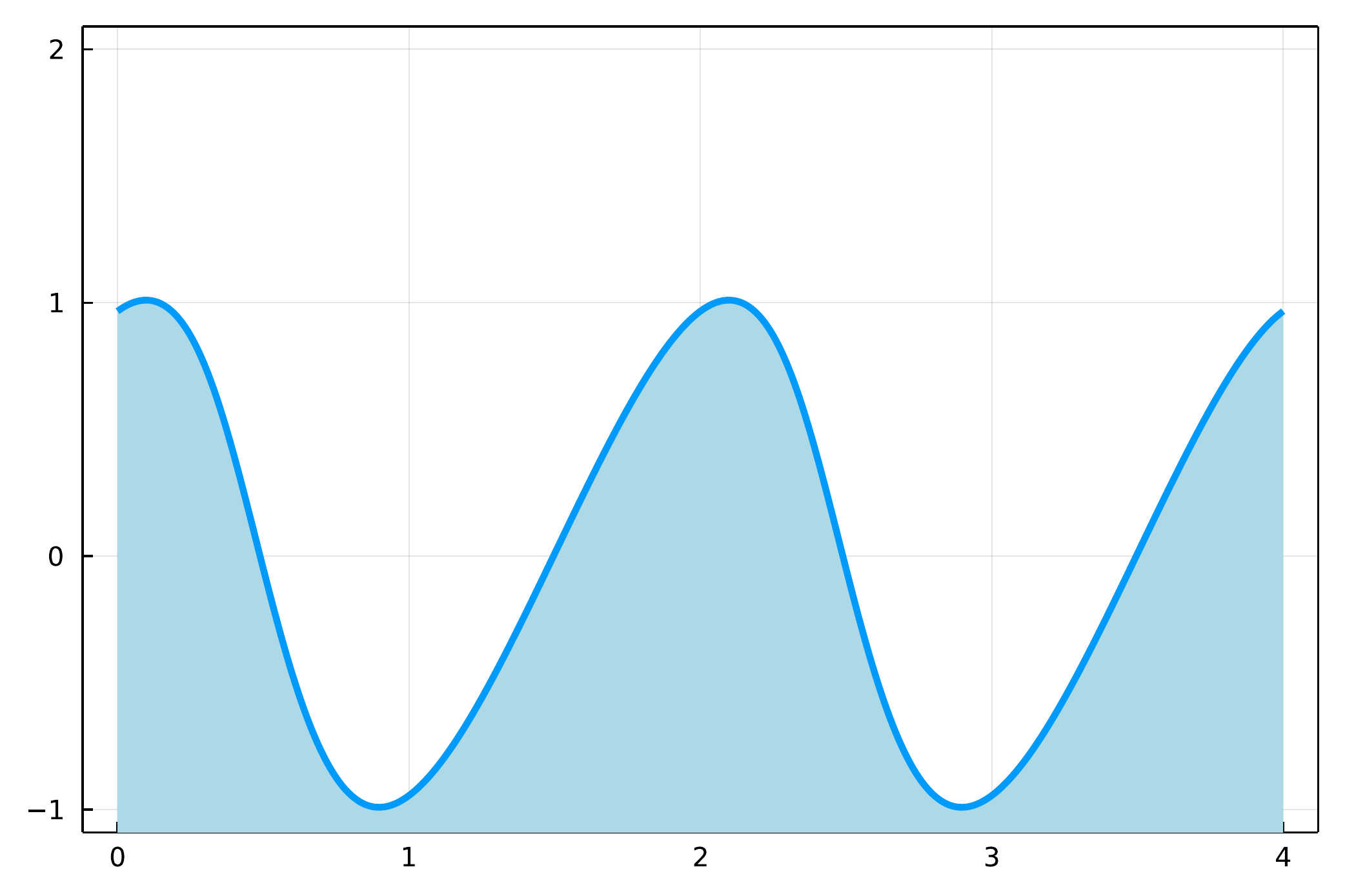}
    \put(-3,25){\rotatebox{90}{$u(x,0.1)$}}
    \put(55,-1){{$x$}}
  \end{overpic}\hspace{0.1in}
  \begin{overpic}[width = 0.45\linewidth]{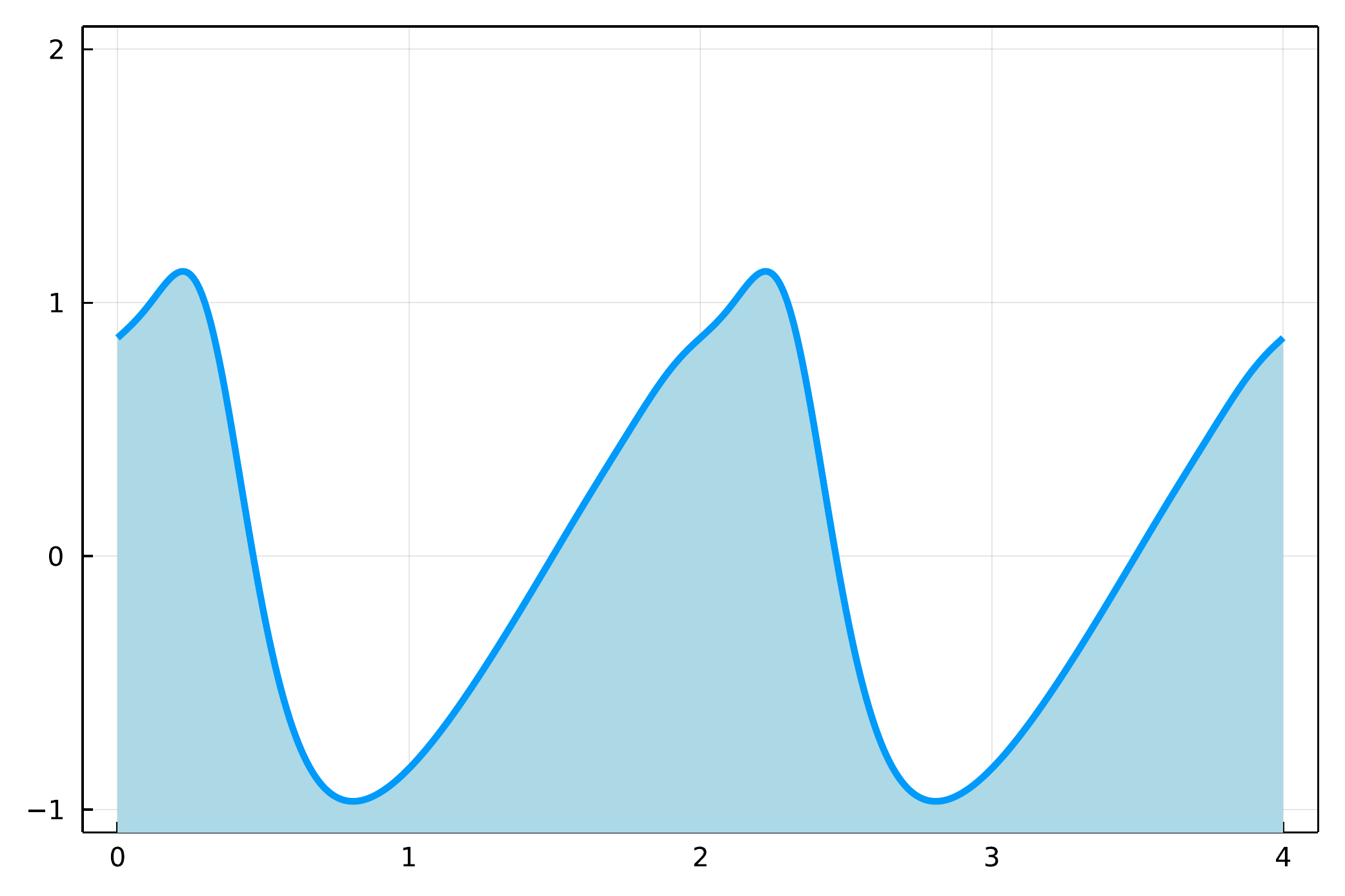}
    \put(-3,25){\rotatebox{90}{$u(x,0.2)$}}
    \put(55,-1){{$x$}}
  \end{overpic}\\
  \vspace{.1in}
  
  \begin{overpic}[width = 0.45\linewidth]{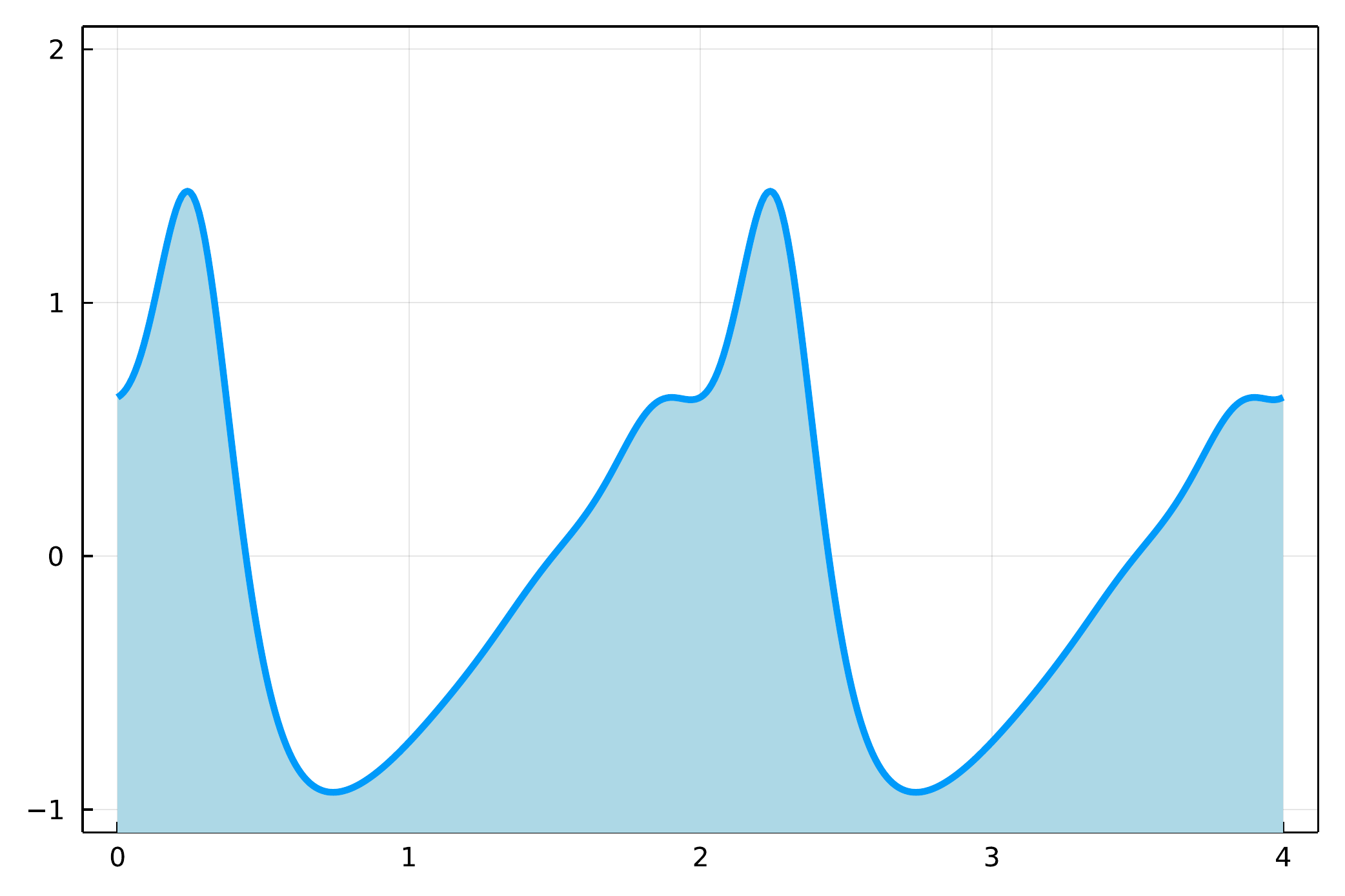}
    \put(-3,25){\rotatebox{90}{$u(x,0.3)$}}
    \put(55,-1){{$x$}}
  \end{overpic}\hspace{0.1in}
   \begin{overpic}[width = 0.45\linewidth]{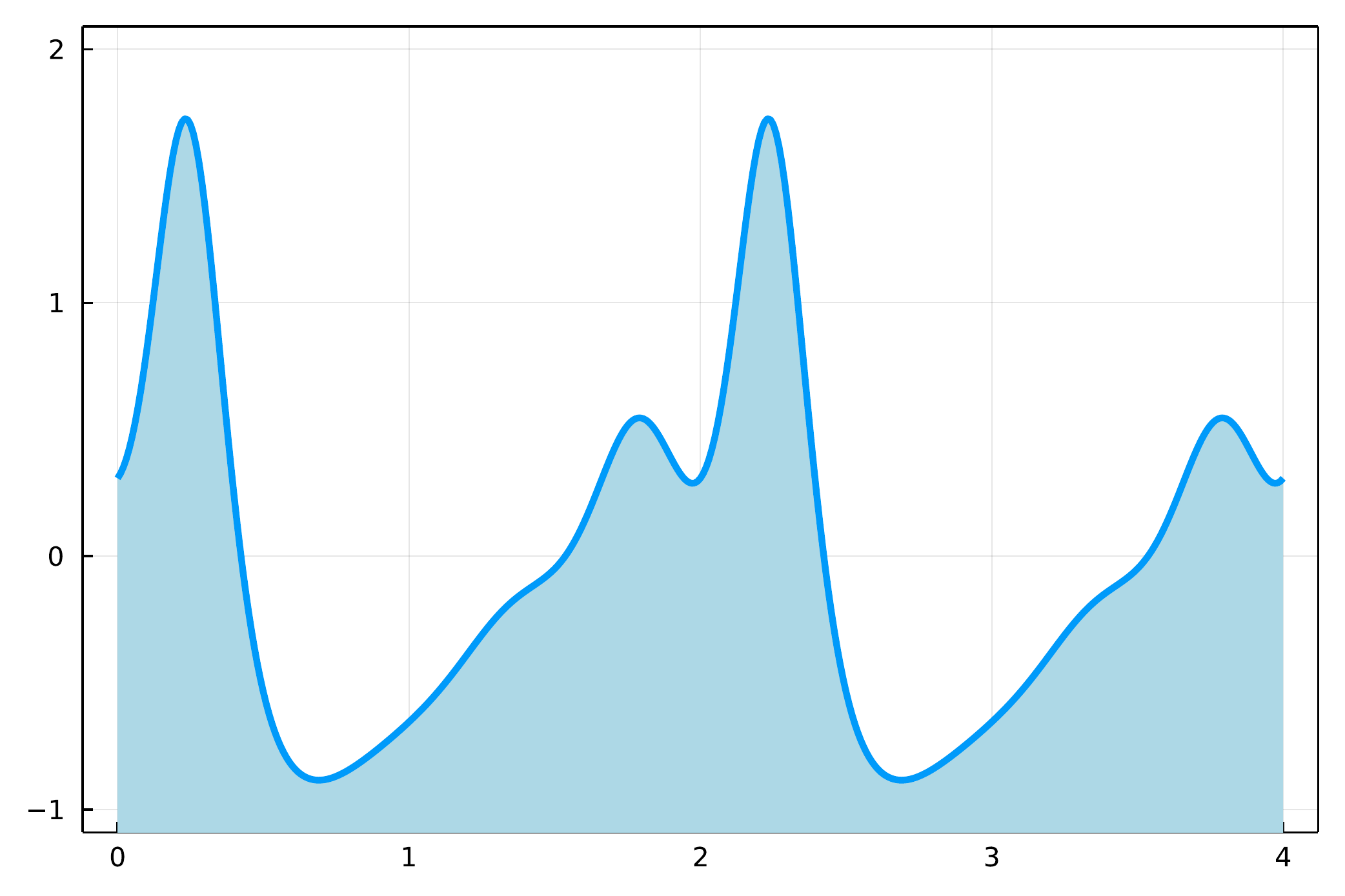}
    \put(-3,25){\rotatebox{90}{$u(x,0.4)$}}
    \put(55,-1){{$x$}}
  \end{overpic}
  \caption{\label{fig:ZK-evolve} The numerically compute evolution of \eqref{eq:uZK} using a genus $g =12$ approximation.  The number of collocation point on each contour is chosen adaptively using the methodology in Section~\ref{sec:adaptive}. }
\end{figure}

\subsection{Initial-value problem with ``box'' data} \label{s:numerics-box}
To be able to solve the initial value problem for the box initial condition, $q_0(x) = q_0(x + L)$,
\begin{equation}
q_0(x) = \begin{cases} 0 & x \in (0,w) \\ -h & x \in (w,L) \end{cases}
\end{equation}
we need to compute the forward spectral theory for \eqref{Schrodinger-op}.
In this case a basis of solutions $c(x;\lambda)$ and $s(x;\lambda)$ to \eqref{Lax-x} normalized as in \eqref{eq:c} and \eqref{eq:s} can be computed explicitly, and these solutions contain all information needed to compute the forward spectral theory.
In particular, we have
\begin{align*}
  \Delta(\lambda) &= 2 \cos (w \sqrt{\lambda}) \cos( (L - w) \sqrt{\lambda - h}) - \frac{2 \lambda - h}{\sqrt{\lambda}} \sqrt{\lambda - h} \sin (w \sqrt{\lambda}) \sin((L- w) \sqrt{\lambda - h}),\\
  s(L;\lambda) &= \frac{1}{\sqrt{\lambda}} \cos((L-w) \sqrt{\lambda - h}) \sin (w \sqrt \lambda)  + \frac{1}{\sqrt{\lambda -h}} \cos(w \sqrt{\lambda}) \sin((L-w)\sqrt{\lambda - h}).
\end{align*}

The Dirichlet eigenvalues $\gamma_1(x_0 = 0)<\gamma_2(x_0 = 0)<\gamma_3(x_0 = 0)<\dots$ of \eqref{Schrodinger-op} are then the zeros of
\begin{equation} c(L;\lambda) = \frac{1}{\sqrt{\lambda}} \sin(\sqrt{\lambda})\cos(\sqrt{\lambda-1}(x-1)) + \frac{1}{\sqrt{\lambda-1}} \cos(\sqrt{\lambda}) \sin(\sqrt{\lambda-1}).\end{equation}
The band ends $\alpha_0 < \beta_0 < \cdots \alpha_g < \beta_g < \cdots$
are the zeros of
\begin{equation} \Delta(\lambda)^2-1. \end{equation}
These are easily computed using standard root-finding techniques.  In practice, we find it convenient to use high-precision arithmetic here so that one is sure where future errors are incurred.  Since this is an infinite genus potential, we specify a finite $g$ to truncate the spectrum, setting $\beta_{g+1} = \infty$, resulting in an approximate solution $u(x,t;g)$.  The convergence of this approximation is slow, but reliable, and this is investigated in Figure~\ref{fig:box-conv}.

\begin{figure}[tbp]
  \centering
  \begin{overpic}[width = 0.75\linewidth]{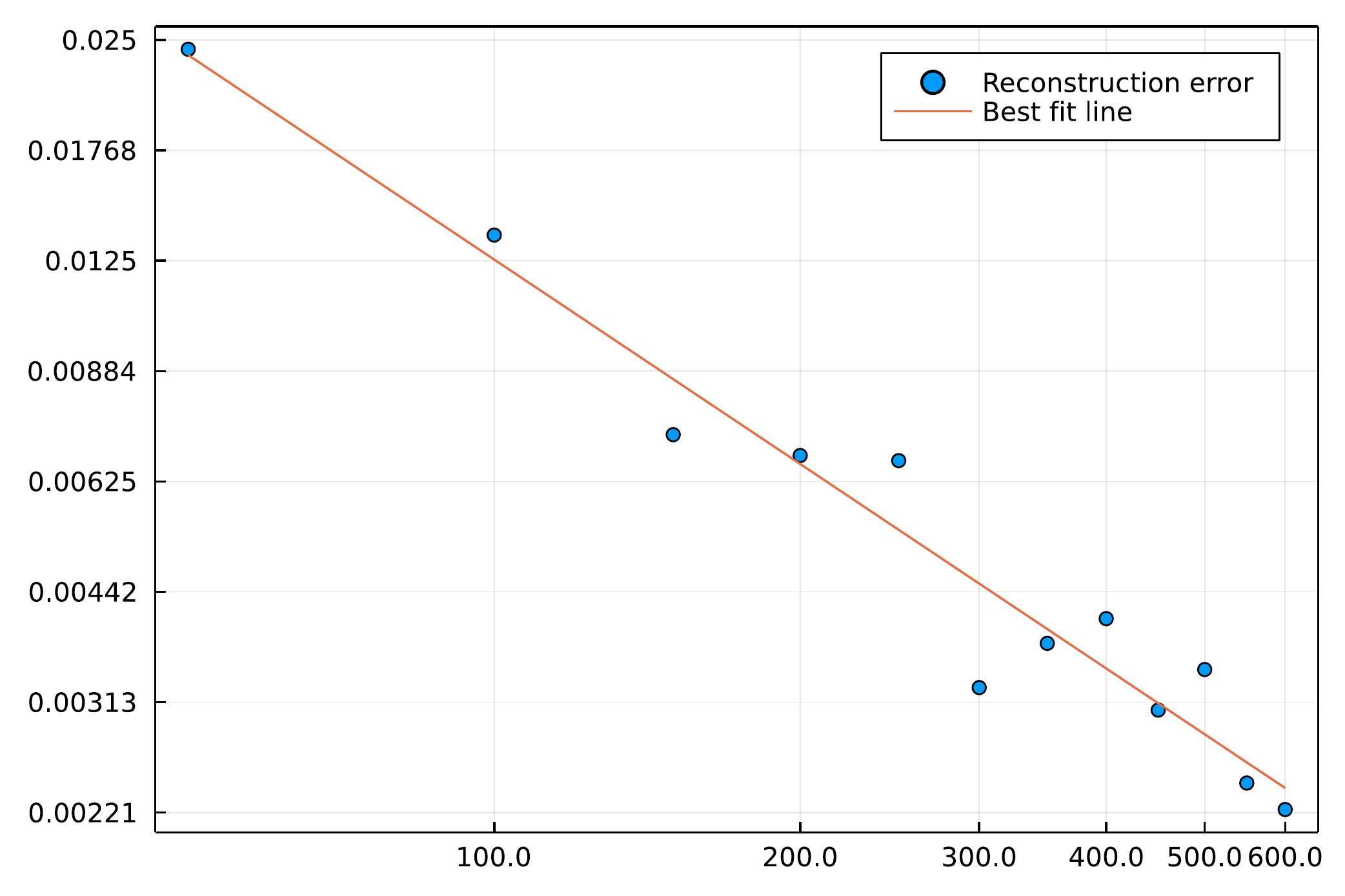}
    \put(-4,20){\rotatebox{90}{$\text{Error:}\, \|u(\vec x,0;g)\|_\infty$}}
    \put(50,-2){{$g$}}
  \end{overpic}\hspace{0.16in}
  
  \caption{\label{fig:box-conv} A test of the convergence of the numerical approximation of $u(x,0;g)$ of  $u(x,0)$ as $g$ increases.  Because of expected non-uniform convergence due to a Gibbs-like phenomenon, we take a uniform grid $\vec x$ on $[0.1,\pi - .1]$ and use $\|u(\vec x,0;g)\|_\infty$ as a proxy for the error ($u(x,0) = 0$ on this interval).  Due to increased oscillations as $g$ increases, this can only be thought of as an estimate for the true error.  The best fit line is given by $\ell(g) \approx 0.89 \times g^{-0.93}$ which appears to be consistent with $O(g^{-1})$ convergence.}
\end{figure}

We investigate various aspects of computing $q(x,t)$ with initial data $q_0$ as above. We focus on the case discussed by Chen and Olver \cite{Chen2014a}.  Specifically, if one chooses $w = \pi/\sqrt{6}$, $L = 2\pi/\sqrt{6}$, $h = -1/2$, then
\begin{align}\label{eq:ukdv}
  u(x,t) := -q(6^{-1/2} x, 6^{-3/2} t),
\end{align}
is the solution of $u_t + u_{xxx} = u u_x$ with initial data
\begin{align}\label{eq:udata}
  u(x,0) = \begin{cases} 0  & 0< x < \pi, \\ 1/2 & \pi < x < 2\pi, \end{cases}
\end{align}
extended periodically.  This allows us to reproduce much of the phenomenon in \cite{Chen2014a}.  In Figure~\ref{fig:dq} we demonstrate dispersive quantization.  When $t$ is a rational multiple of $\pi$, $u(x,t)$ as a function of $x$ appears to be piecewise smooth and slowly varying.  When $t$ is an irrational multiple of $\pi$ the solution appears to have a fractal nature.  One interesting observation we make here is that while Chen and Olver remark in \cite{Chen2014a} that it is not clear from their numerical method if there are truly oscillations between jumps at rational-times-$\pi$ times, our numerics in Figure~\ref{fig:zoomed} indicate that these oscillations will disappear as the genus increases.

We also use this problem to illustrate some important aspects of the numerical method we have developed. First, recall the matrix ${\mathbf{A}}$ in \eqref{eq:Amatrix}.  We plot the eigenvalues of $A(0,0)$ in the left panel of Figure~\ref{fig:precond}.  Here we use $10$ collocation points for $I_{j}$ if $|j| \leq 4$ and $3$ collocation points otherwise.  Then in the right panel of  Figure~\ref{fig:precond} we show the preconditioned matrix ${\tilde{\mathbf{A}}}(0,0)^{-1} {\mathbf{A}}(0,0)$ where ${\tilde{\mathbf{A}}}$ is obtained from a discretization of \eqref{eq:preconder}.  It becomes clear that the eigenvalues become localized near $\lambda = 2$. This problem is extremely well conditioned and GMRES will converge in just a few iterations.

Lastly, in Figure~\ref{fig:coefs} we display how the magnitude of the computed Chebyshev coefficients $\gamma_{i,j}$ in \eqref{eq:expand} depends on $I_j$.  In the top-left panel of Figure~\ref{fig:coefs} we see that $\gamma_{0,j} = O(1)$.  This is not unexpected because the recovery formula \eqref{eq:recovery} weights these coefficients by the gap lengths.  What is rather surprising is how, for large $|j|$, the second coefficient in \eqref{eq:expand} decays rapidly, see the top-right panel of Figure~\ref{fig:coefs}.  The third coefficient decays even more rapidly, see the bottom panel of Figure~\ref{fig:coefs}.

\begin{figure}[tbp]
  \centering
  \begin{overpic}[width = 0.45\linewidth]{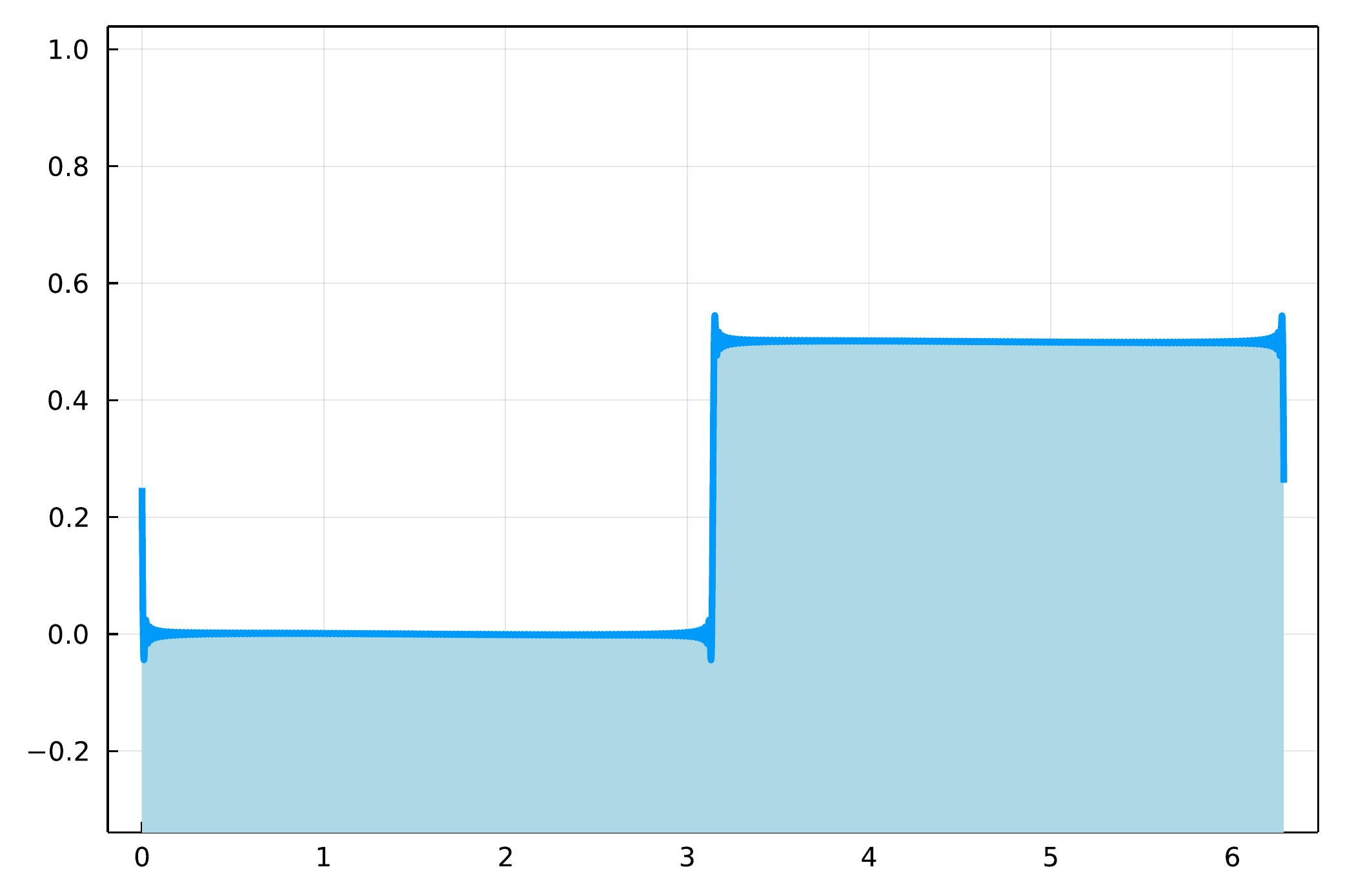}
    \put(-3,25){\rotatebox{90}{$u(x,0)$}}
    \put(55,-1){{$x$}}
  \end{overpic}\hspace{0.1in}
  \begin{overpic}[width = 0.45\linewidth]{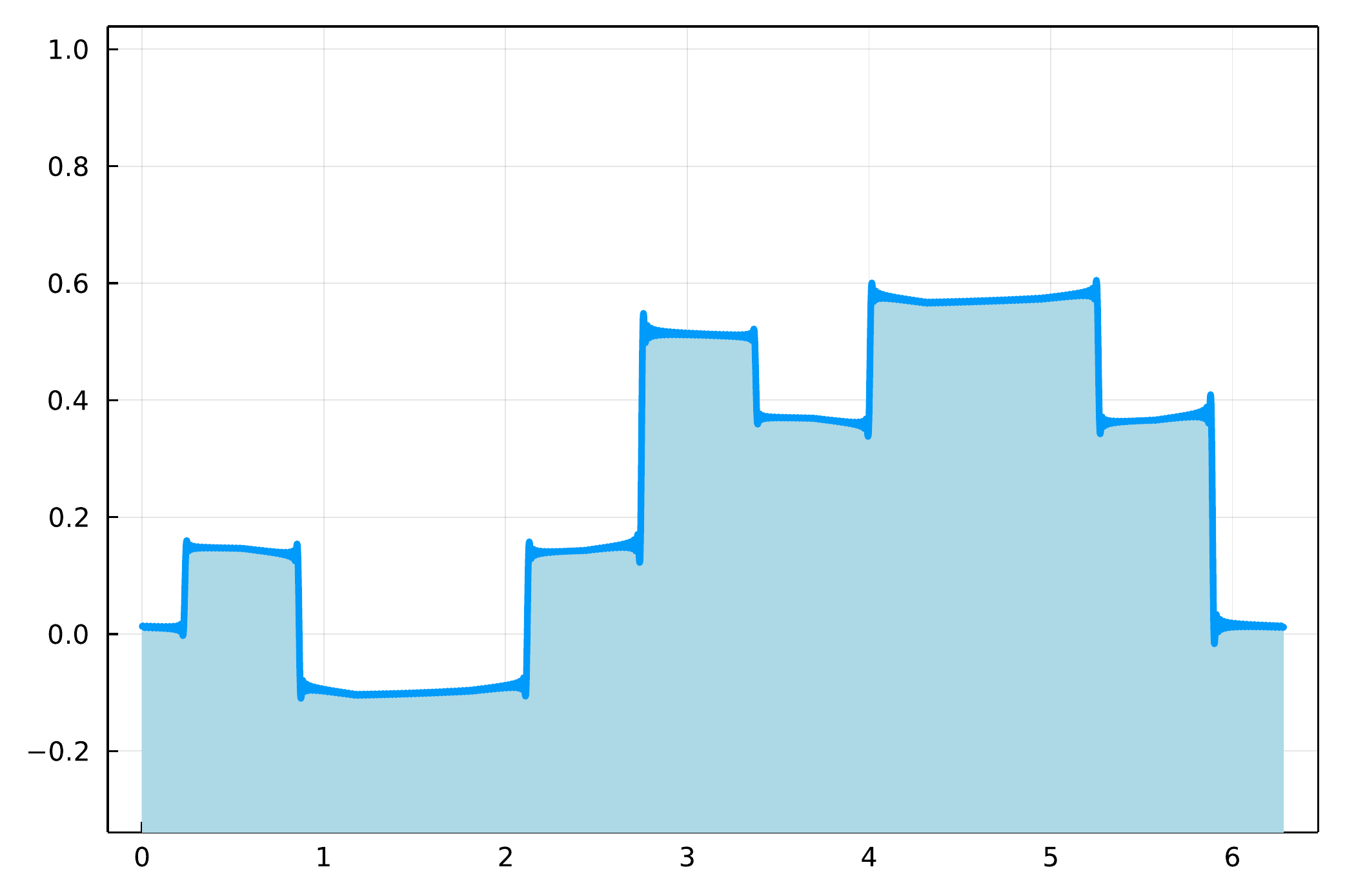}
    \put(-3,25){\rotatebox{90}{$u(x,0.1\pi)$}}
    \put(55,-1){{$x$}}
  \end{overpic}\\
  \vspace{.1in}
  
  \begin{overpic}[width = 0.45\linewidth]{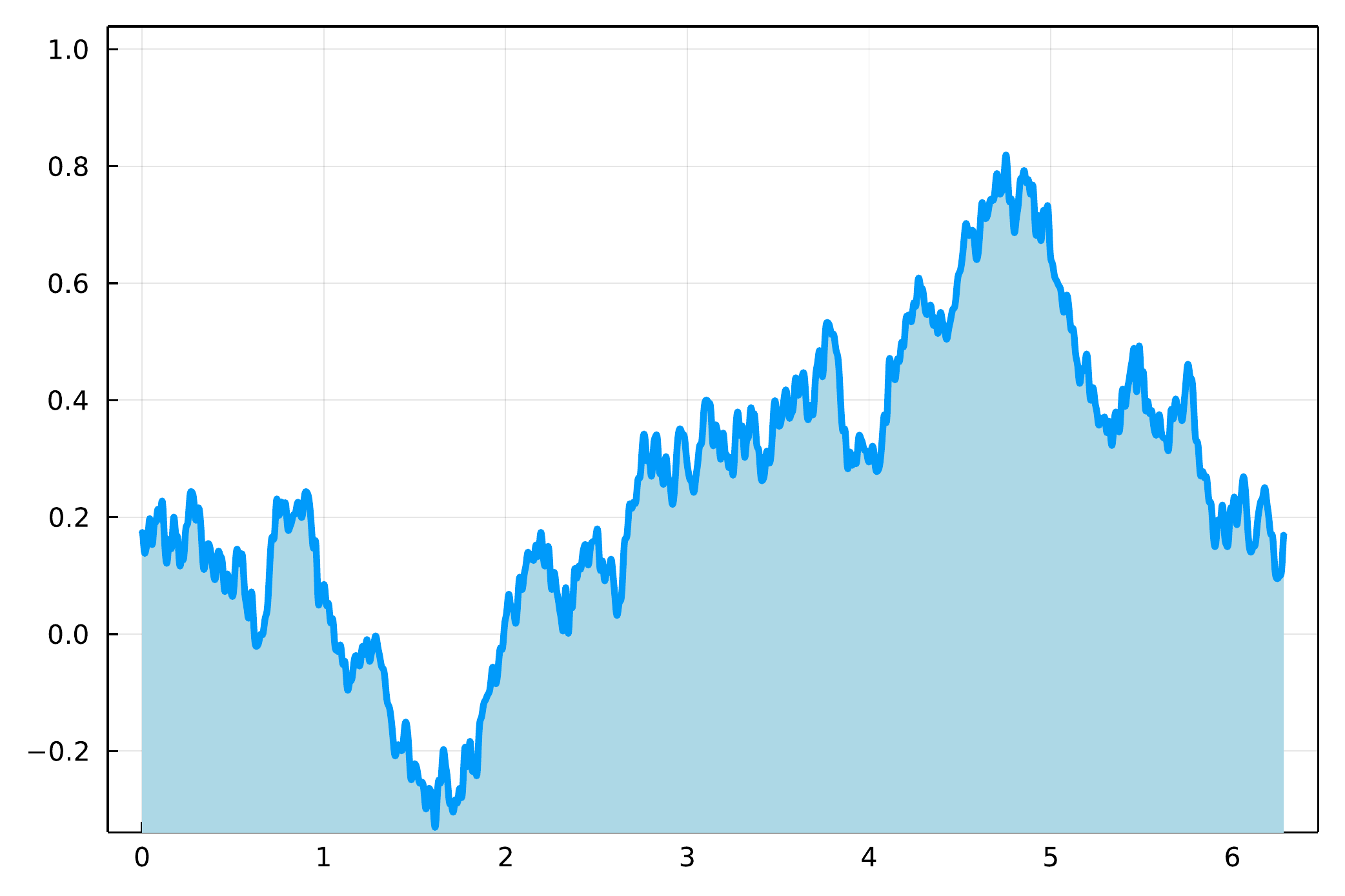}
    \put(-3,25){\rotatebox{90}{$u(x,0.1)$}}
    \put(55,-1){{$x$}}
  \end{overpic}\hspace{0.1in}
   \begin{overpic}[width = 0.45\linewidth]{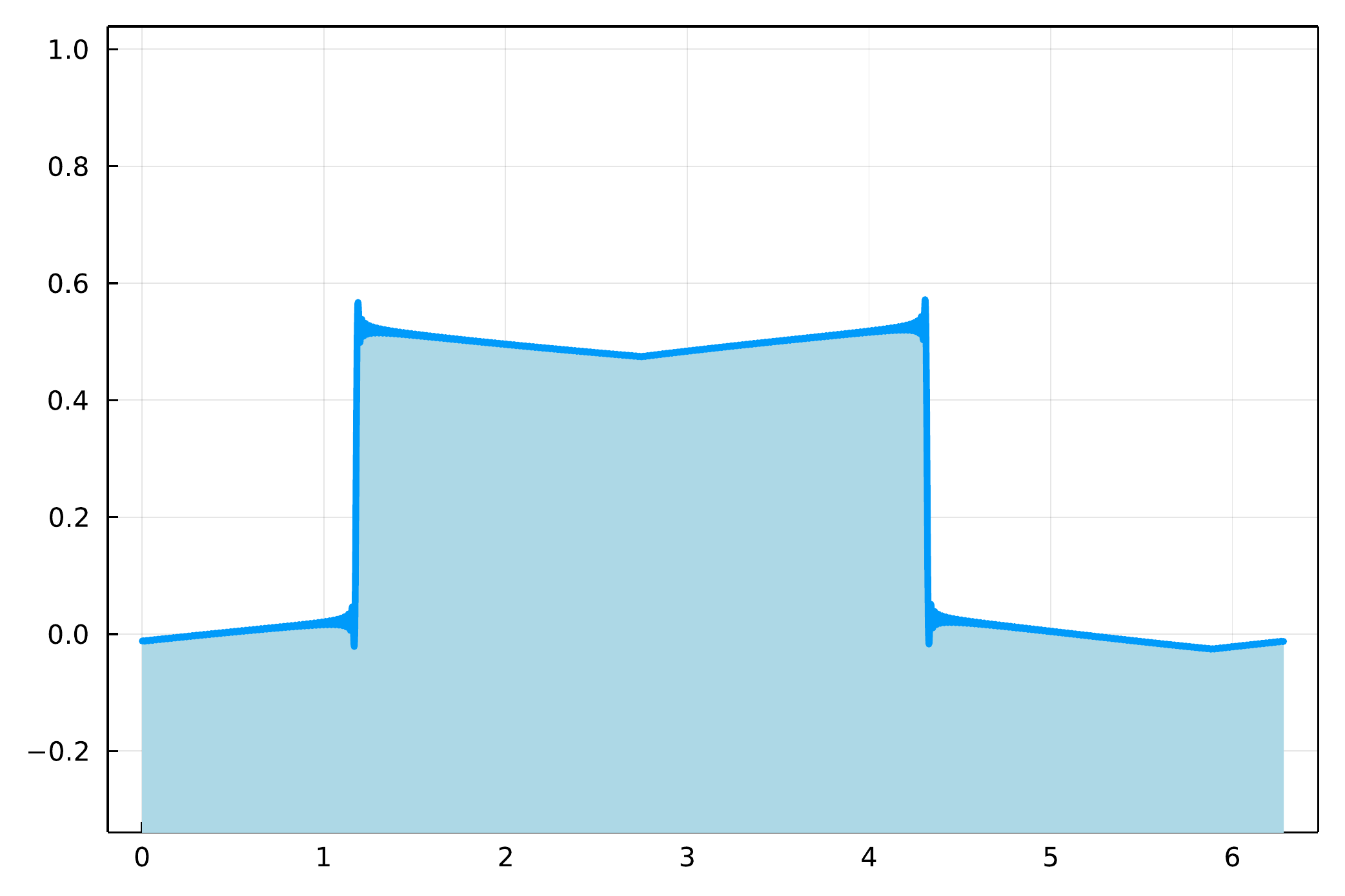}
    \put(-3,25){\rotatebox{90}{$u(x,0.5\pi)$}}
    \put(55,-1){{$x$}}
  \end{overpic}\\
  \vspace{.1in}
  \begin{overpic}[width = 0.45\linewidth]{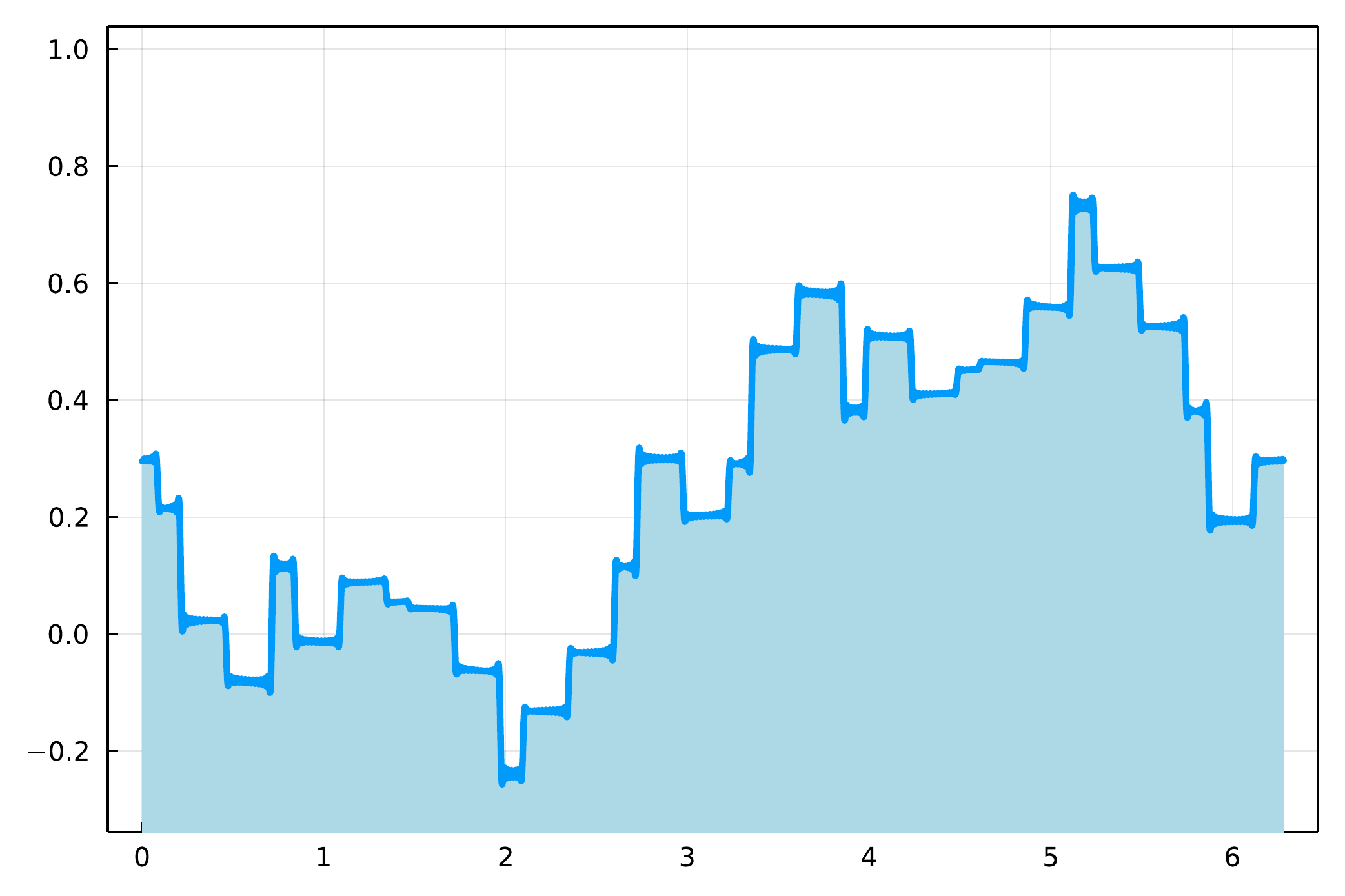}
    \put(-3,25){\rotatebox{90}{$u(x,0.01\pi)$}}
    \put(55,-1){{$x$}}
  \end{overpic}
  \caption{\label{fig:dq} The evolution of $u$ in \eqref{eq:ukdv} with \eqref{eq:udata}. These plots show dispersive quantization where the solution appears to be piecewise smooth at rational-times-$\pi$ times and fractal otherwise.  This was first observed by Chen and Olver, see \cite{Chen2014a}, for example. These plots are produced using a genus $g = 300$ approximation and using $10$ collocation points on $I_j$ if $|j| \leq 4$ and two collocation points otherwise.  This choice is justified by Figure~\ref{fig:coefs}. }
\end{figure}

\begin{figure}[tbp]
  \centering
  \begin{overpic}[width = 0.45\linewidth]{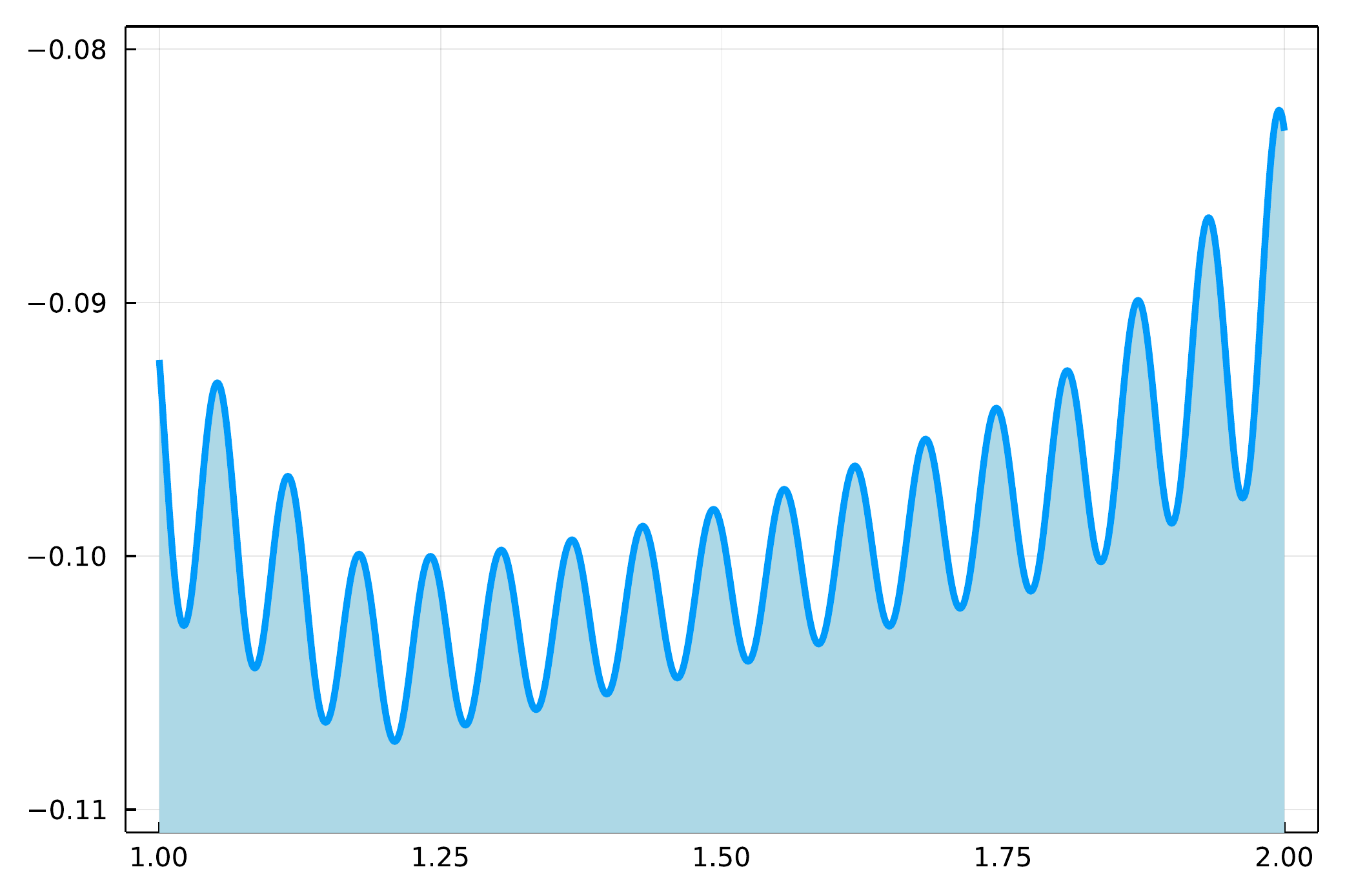}
    \put(70,58){$g = 100$}
    \put(-3,25){\rotatebox{90}{$u(x,0.1\pi)$}}
    \put(55,-1){{$x$}}
  \end{overpic}\hspace{0.1in}
  \begin{overpic}[width = 0.45\linewidth]{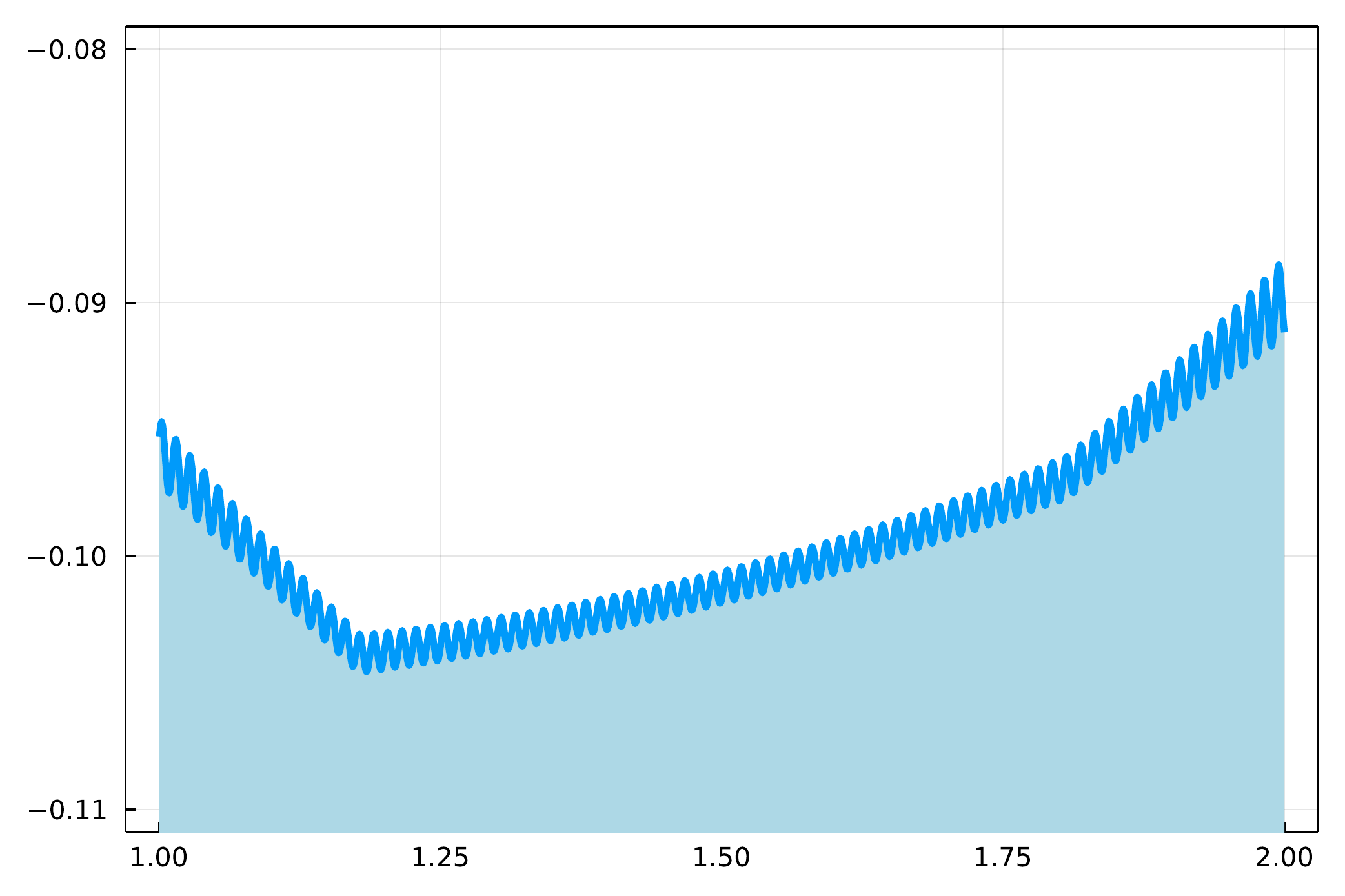}
     \put(70,58){$g = 500$}
    \put(-3,25){\rotatebox{90}{$u(x,0.1\pi)$}}
    \put(55,-1){{$x$}}
  \end{overpic}\\
  \vspace{.1in}
  
  \begin{overpic}[width = 0.45\linewidth]{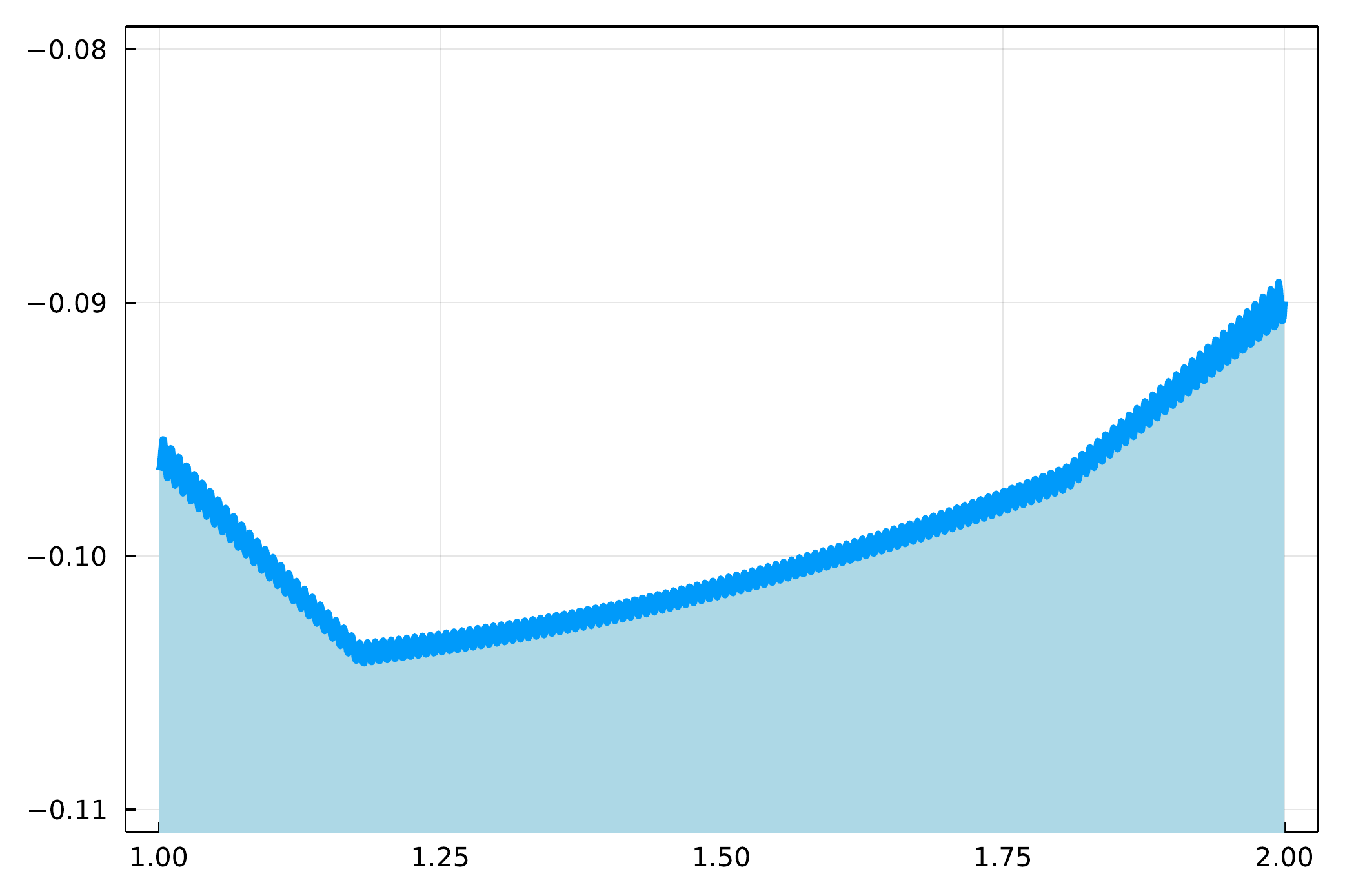}
     \put(70,58){$g = 900$}
    \put(-3,25){\rotatebox{90}{$u(x,0.1\pi)$}}
    \put(55,-1){{$x$}}
  \end{overpic}
  \caption{\label{fig:zoomed} A zoomed view of $u(x,0.1\pi)$ as $g$ increases.  These plots indicate that the amplitude of the oscillations decrease as the genus increases.  This leads to the conjecture that the limiting solution profile is piecewise smooth and slowly varying. }
\end{figure}

\begin{figure}[tbp]
  \centering
  \begin{overpic}[width = 0.45\linewidth]{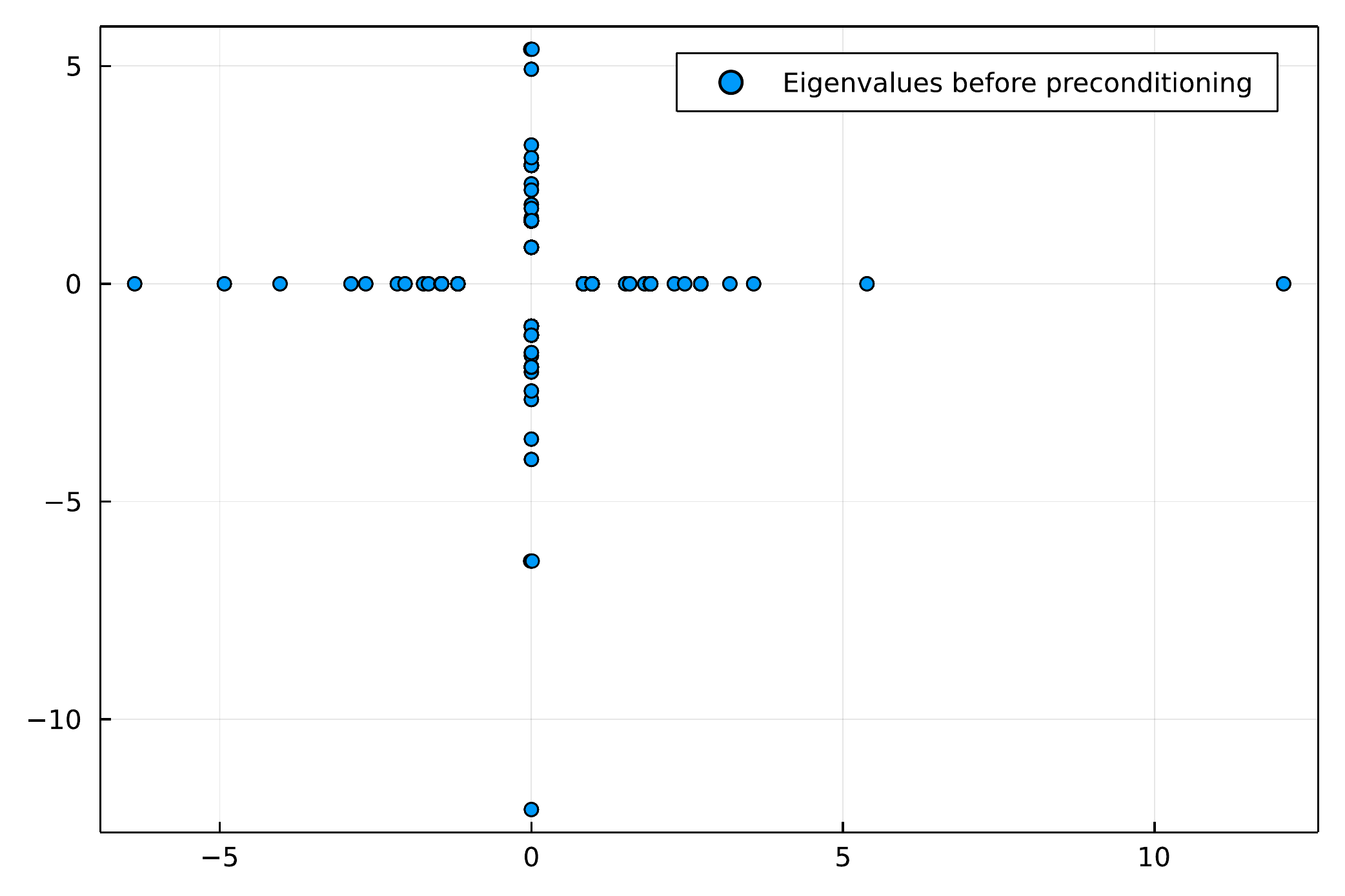}
    \put(-3,30){\rotatebox{90}{{$\mathrm{Im}\, \lambda$}}}
    \put(50,-3){{$\mathrm{Re}\, \lambda$}}
  \end{overpic}\hspace{0.1in}
 \begin{overpic}[width = 0.45\linewidth]{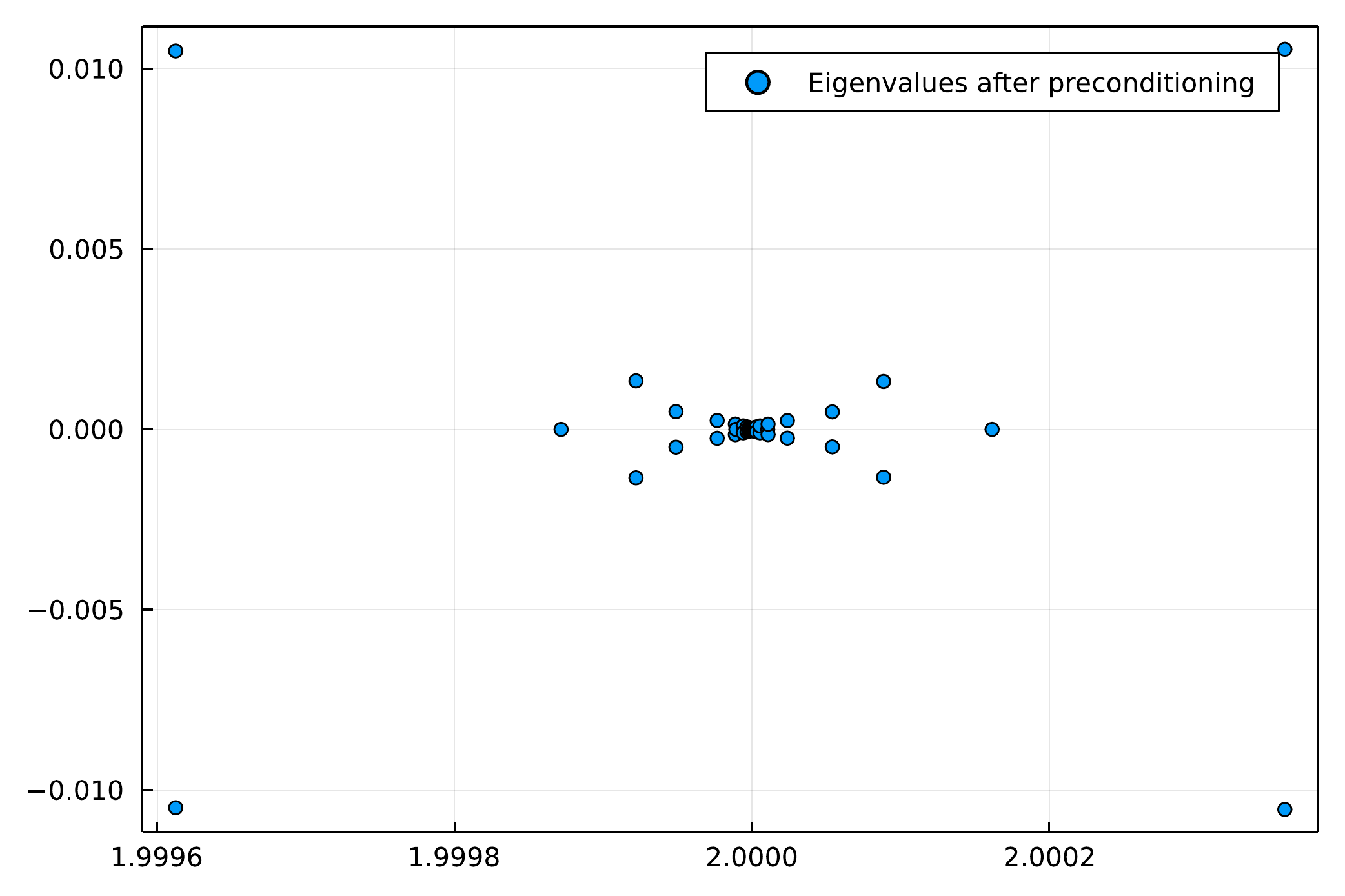}
    \put(-3,30){\rotatebox{90}{{$\mathrm{Im}\, \lambda$}}}
    \put(50,-3){{$\mathrm{Re}\, \lambda$}}
  \end{overpic}
  
  \caption{\label{fig:precond} Left panel: The eigenvalues of $\mathbf{A}(0,0)$ from \eqref{eq:Amatrix} for the potential \eqref{eq:udata}.   We use $10$ collocation points for $I_{j}$ if $|j| \leq 4$ and $3$ collocation points otherwise.  Right panel: The preconditioned matrix $\tilde{\mathbf{A}}(0,0)^{-1} \mathbf{A}(0,0)$ where $\tilde{\mathbf{A}}$ is obtained from a discretization of \eqref{eq:preconder}.  The eigenvalues become localized near $\lambda = 2$. This problem is extremely well conditioned and GMRES will converge in just a few iterations. }
\end{figure}

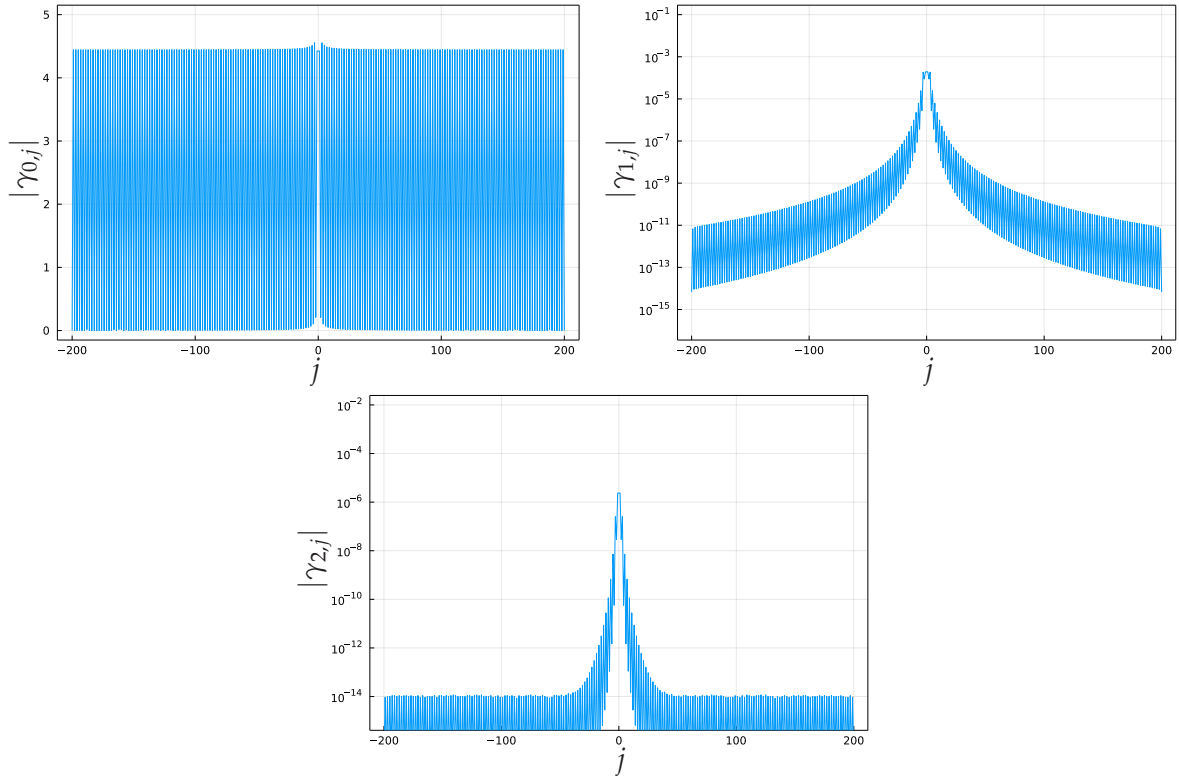
\begin{figure}[tbp]
  \centering
  \begin{overpic}[width = 0.45\linewidth]{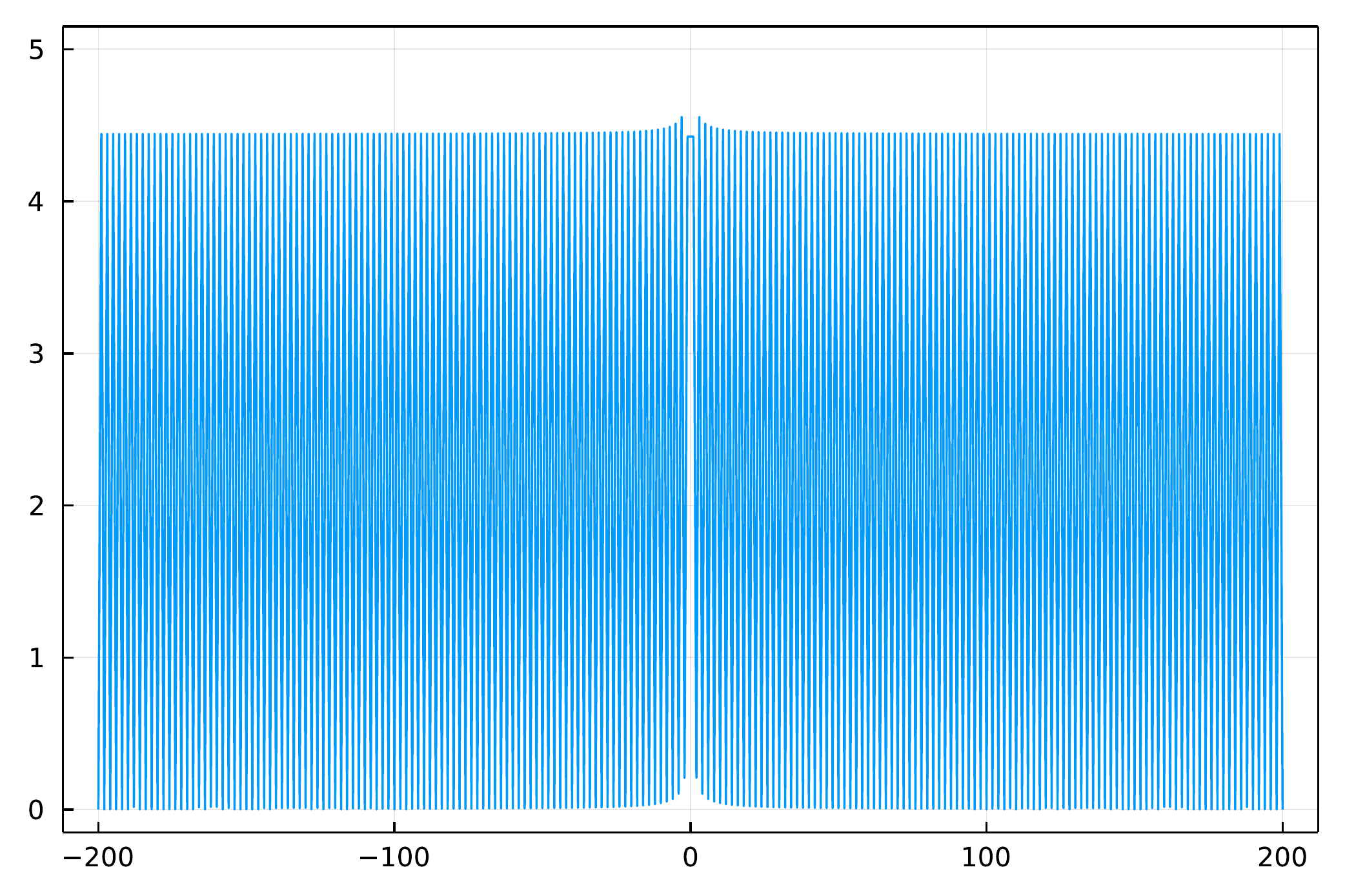}
    \put(-4,30){\rotatebox{90}{$|\gamma_{0,j}|$}}
    \put(50,-2){{$j$}}
  \end{overpic}\hspace{0.16in}
  \begin{overpic}[width = 0.45\linewidth]{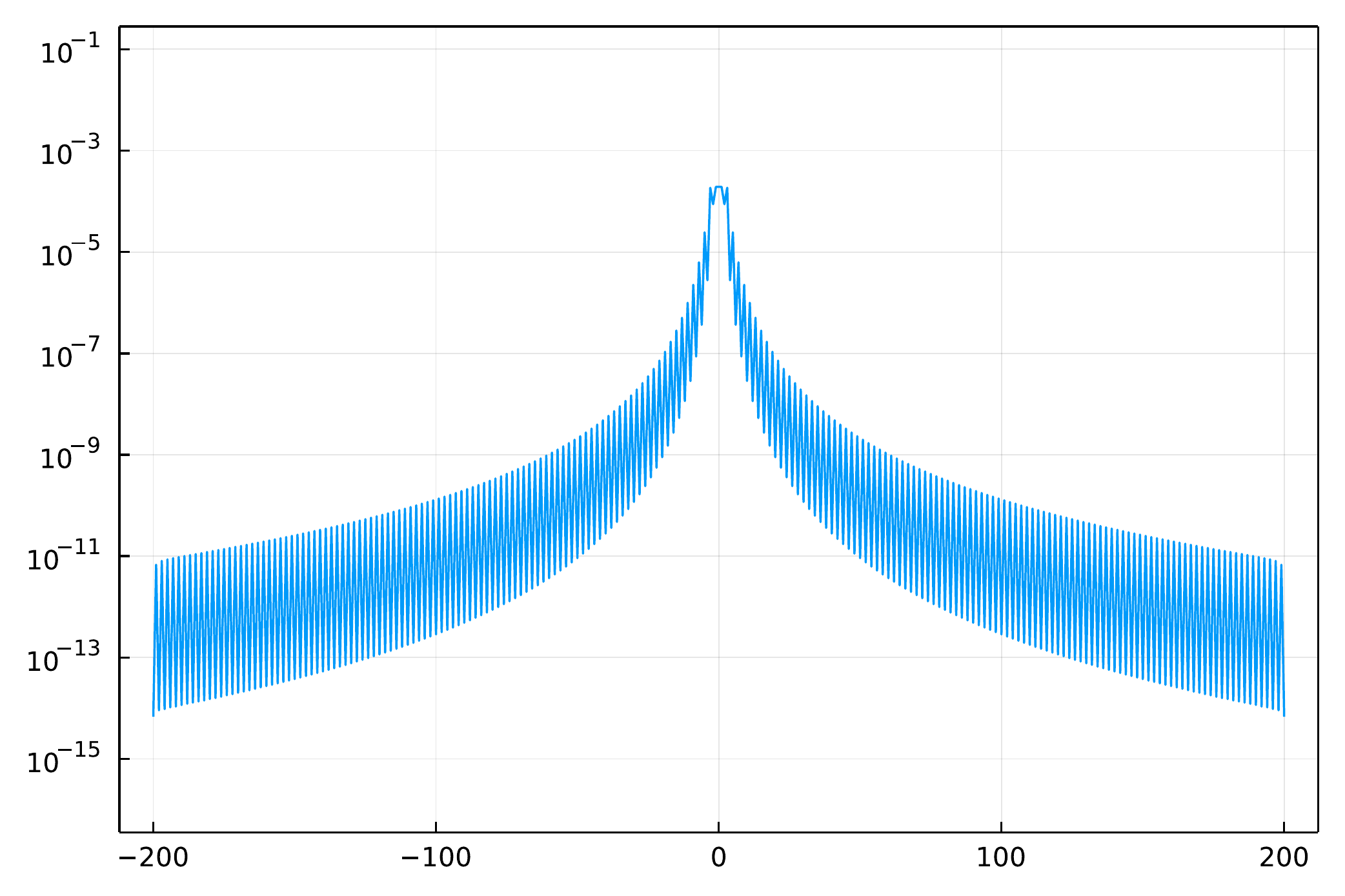}
    \put(-4,30){\rotatebox{90}{$|\gamma_{1,j}|$}}
    \put(53,-2){{$j$}}
  \end{overpic}\\
  \vspace{.1in}
  
  \begin{overpic}[width = 0.45\linewidth]{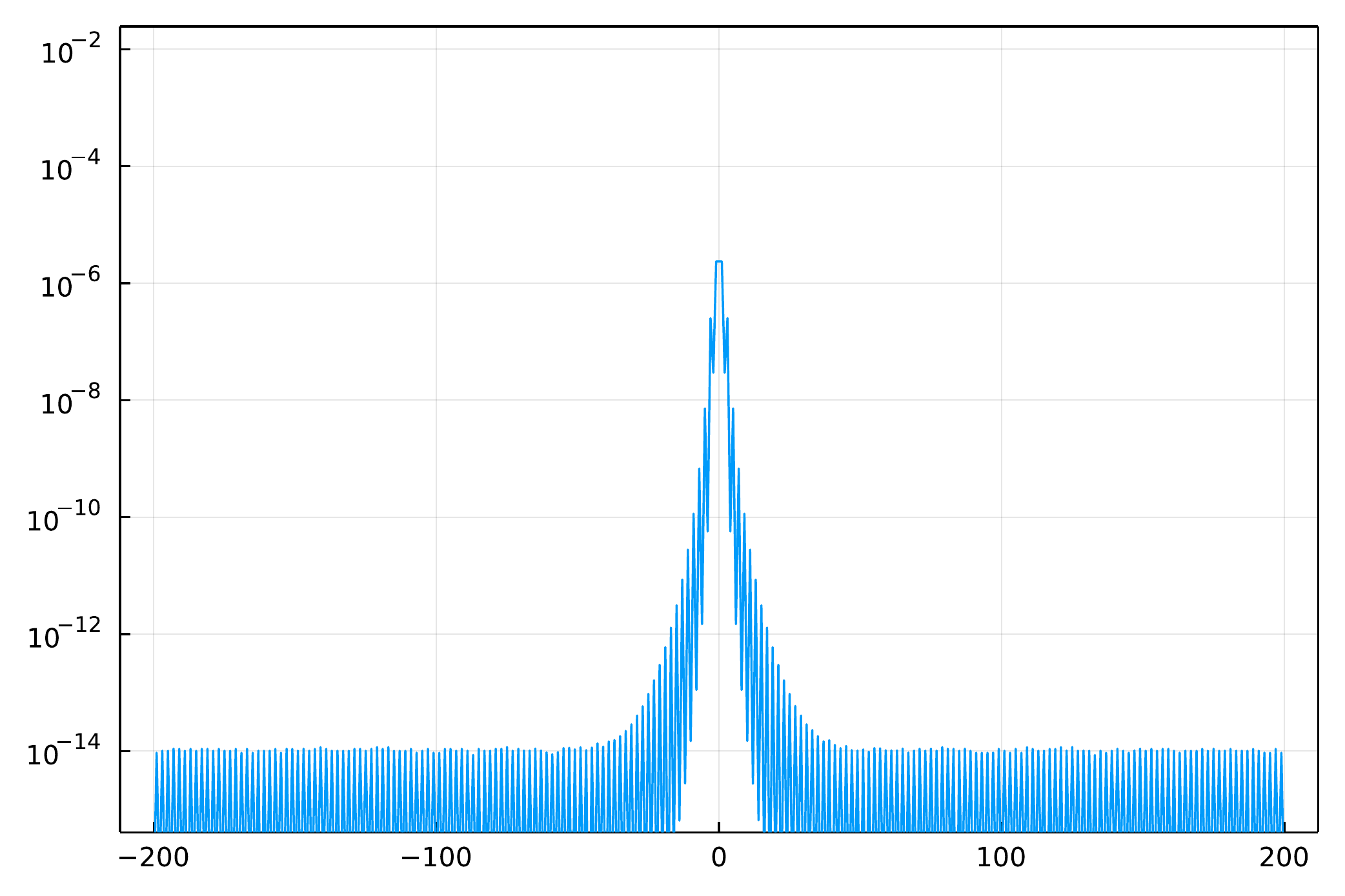}
    \put(-4,30){\rotatebox{90}{$|\gamma_{2,j}|$}}
    \put(53,-2){{$j$}}
  \end{overpic}\hspace{0.1in}
  
  \caption{\label{fig:coefs} The magnitude how of the computed Chebyshev coefficients $\gamma_{i,j}$, $i = 0,1,2$ in \eqref{eq:expand} depends on $I_j$, $j = \pm 1, \pm 2, \ldots, \pm g$, for the potential \eqref{eq:udata}.  As $i$ increases the decay rate with respect to $|j|$ is extremely rapid.}
\end{figure}
\clearpage

\bibliographystyle{abbrv}
\bibliography{library}  

\begin{thebibliography}{10}

\bibitem{chebfun}
Z.~Battles and L.~N. Trefethen.
\newblock {An extension of MATLAB to continuous functions and operators}.
\newblock {\em SIAM J. Sci. Comput.}, 25:1743--1770, 2004.

\bibitem{Algebro}
E.~D. Belokolos, A.~I. Bobenko, V.~Z. Enol'skii, A.~R. Its, and V.~B. Matveev.
\newblock {\em {Algebro-Geometric Approach to Nonlinear Integrable Equations}}.
\newblock Springer, 1994.

\bibitem{Berry1996}
M.~V. Berry and S.~Klein.
\newblock {Integer, fractional and fractal Talbot effects}.
\newblock {\em Journal of Modern Optics}, 43(10):2139--2164, oct 1996.

\bibitem{Bilman2018}
D.~Bilman and T.~Trogdon.
\newblock {On numerical inverse scattering for the Korteweg-de Vries equation
  with discontinuous step-like data}.
\newblock {\em Nonlinearity}, 33(5):2211--2269, may 2020.

\bibitem{Bottcher1997}
A.~B{\"{o}}ttcher and Y.~I. Karlovich.
\newblock {\em {Carleson Curves, Muckenhoupt Weights, and Toeplitz Operators}}.
\newblock Birkh{\"{a}}user Basel, Basel, 1997.

\bibitem{Carrier1988}
J.~Carrier, L.~Greengard, and V.~Rokhlin.
\newblock {A fast adaptive multipole algorithm for particle simulations}.
\newblock {\em SIAM Journal on Scientific and Statistical Computing},
  9(4):669--686, jul 1988.

\bibitem{Chen2014a}
G.~Chen and P.~{J. Olver}.
\newblock {Numerical simulation of nonlinear dispersive quantization}.
\newblock {\em Discrete {\&} Continuous Dynamical Systems - A},
  34(3):991--1008, 2014.

\bibitem{Deconinck-theta}
B.~Deconinck, M.~Heil, A.~Bobenko, M.~van Hoeij, and M.~Schmies.
\newblock {Computing Riemann theta functions}.
\newblock {\em Math. Comp.}, 73:1417--1442, 2004.

\bibitem{Deconinck2006b}
B.~Deconinck and J.~N. Kutz.
\newblock {Computing spectra of linear operators using the Floquet-Fourier-Hill
  method}.
\newblock {\em Journal of Computational Physics}, 219:296--321, 2006.

\bibitem{Dubrovin}
B.~A. Dubrovin.
\newblock {Inverse problem for periodic finite zoned potentials in the theory
  of scattering}.
\newblock {\em Func. Anal. and Its Appl.}, 9:61--62, 1975.

\bibitem{Dubrovin75-fan1}
B.~A. Dubrovin.
\newblock {Inverse problem of scattering theory for periodic finite-zone
  potentials in the theory of scattering}.
\newblock {\em Funktsionaln. Analiz i ego Prilozhenija}, 9:65--66, 1975.

\bibitem{Dubrovin75-fan2}
B.~A. Dubrovin.
\newblock {Periodic problem for the Korteweg-de Vries equation in the class of
  finite-zone potentials}.
\newblock {\em Funktsionaln. Analiz i ego Prilozhenija}, 9:41--51, 1975.

\bibitem{DubrovinNotes}
B.~A. Dubrovin.
\newblock {Integrable Systems and Riemann Surfaces Lecture Notes}.
\newblock http://people.sissa.it/{\~{}}dubrovin/rsnleq{\_}web.pdf, 2009.

\bibitem{DubrovinN74}
B.~A. Dubrovin and S.~P. Novikov.
\newblock {Periodic and conditionally periodic analogs of the many-soliton
  solutions of the korteweg-de vries equation}.
\newblock {\em Soviet Phys. JETP}, 67:1058--1063, 1974.

\bibitem{Duren}
P.~Duren.
\newblock {\em {Theory of {\$}H{\^{}}p{\$} Spaces}}.
\newblock Academic Press, 1970.

\bibitem{Dyachenko2016}
S.~Dyachenko, D.~Zakharov, and V.~Zakharov.
\newblock {Primitive potentials and bounded solutions of the KdV equation}.
\newblock {\em Physica D: Nonlinear Phenomena}, 333:148--156, oct 2016.

\bibitem{klein}
J.~Frauendiener and C.~Klein.
\newblock {Hyperelliptic theta-functions and spectral methods: KdV and KP
  solutions}.
\newblock {\em Lett. Math. Phys.}, 76:249--267, 2006.

\bibitem{Klein2008}
C.~Klein.
\newblock {Fourth order time-stepping for low dispersion Korteweg--de Vries and
  nonlinear Schroedinger equations}.
\newblock {\em Electronic Transactions on Numerical Analysis}, 29:116--135,
  2008.

\bibitem{Lax68}
P.~D. Lax.
\newblock {Integrals of Nonlinear Equations of Evolution and Solitary Waves}.
\newblock {\em Communications on Pure and Applied Mathematics}, 21:467--490,
  1968.

\bibitem{Lax}
P.~D. Lax.
\newblock {Periodic solutions of the KdV equation}.
\newblock {\em Comm. Pure Appl. Math.}, 28:141--188, 1975.

\bibitem{McLaughlin2018}
K.~T.-R. McLaughlin and P.~V. Nabelek.
\newblock {A Riemann--Hilbert Problem Approach to Infinite Gap Hill's Operators
  and the Korteweg--de Vries Equation}.
\newblock {\em International Mathematics Research Notices}, 2021(2):1288--1352,
  jan 2021.

\bibitem{TheoryOfSolitons}
S.~P. Novikov, S.~V. Manakov, L.~P. Pitaevskii, and V.~E. Zakharov.
\newblock {\em {Theory of Solitons}}.
\newblock Constants Bureau, New York, 1984.

\bibitem{Olver2010a}
P.~J. Olver.
\newblock {Dispersive Quantization}.
\newblock {\em The American Mathematical Monthly}, 117(7):599, 2010.

\bibitem{approxfun}
S.~Olver.
\newblock {ApproxFun.jl}.
\newblock 2014.

\bibitem{Olver2020}
S.~Olver, R.~M. Slevinsky, and A.~Townsend.
\newblock {Fast algorithms using orthogonal polynomials}.
\newblock {\em Acta Numerica}, 29:573--699, may 2020.

\bibitem{Olver2013a}
S.~Olver and A.~Townsend.
\newblock {A Fast and Well-Conditioned Spectral Method}.
\newblock {\em SIAM Review}, 55(3):462--489, jan 2013.

\bibitem{Osborne}
A.~R. Osborne.
\newblock {\em {Nonlinear ocean waves and the inverse scattering transform}},
  volume~97 of {\em International Geophysics Series}.
\newblock Elsevier/Academic Press, Boston, MA, 2010.

\bibitem{Talbot1836}
H.~F. Talbot.
\newblock {LXXVI. Facts relating to optical science. No. IV}.
\newblock {\em Philosophical Magazine Series 3}, 9(56):401--407, dec 1836.

\bibitem{TrefethenATAP}
L.~N. Trefethen.
\newblock {Approximation Theory and Approximation Practice}, 2019.

\bibitem{Trogdon2022}
T.~Trogdon, D.~Bilman, and P.~V. Nabelek.
\newblock {https://github.com/tomtrogdon/PeriodicKdV.jl}.
\newblock 2022.

\bibitem{TrogdonFiniteGenus}
T.~Trogdon and B.~Deconinck.
\newblock {A Riemann--Hilbert problem for the finite-genus solutions of the KdV
  equation and its numerical solution}.
\newblock {\em Physica D}, 251:1--18, 2013.

\bibitem{Trogdon2013a}
T.~Trogdon and B.~Deconinck.
\newblock {Numerical computation of the finite-genus solutions of the
  Korteweg-de Vries equation via Riemann-Hilbert problems}.
\newblock {\em Applied Mathematics Letters}, 26(1), 2013.

\bibitem{TrogdonDressing}
T.~Trogdon and B.~Deconinck.
\newblock {A numerical dressing method for the nonlinear superposition of
  solutions of the KdV equation}.
\newblock {\em Nonlinearity}, 27(1):67--86, jan 2014.

\bibitem{TrogdonSOBook}
T.~Trogdon and S.~Olver.
\newblock {\em {Riemann--Hilbert Problems, Their Numerical Solution and the
  Computation of Nonlinear Special Functions}}.
\newblock SIAM, Philadelphia, PA, 2016.

\bibitem{Zabusky1965}
N.~Zabusky and M.~Kruskal.
\newblock {Interaction of "Solitons" in a collisionless plasma and the
  recurrence of initial states}.
\newblock {\em Physical Review Letters}, 15(6):240--243, aug 1965.

\end{thebibliography}

\end{document}